\documentclass[10pt]{article}
\usepackage[left=1.5in,right=1.5in]{geometry}
\pdfoutput=1
\usepackage{graphicx}
\usepackage{mathptmx}
\usepackage{mathrsfs}
\usepackage{amssymb}
\usepackage{amsmath}
\usepackage{amsmath,amsfonts}
\usepackage{algorithm}
\usepackage{algorithmic}
\usepackage[numbers,square]{natbib}
\usepackage{amscd}
\usepackage[mathscr]{eucal}
\usepackage{tikz}
\usepackage[doublespacing]{setspace}
\usepackage{fancyhdr}
\usepackage{lineno}

\usepackage{enumitem}
\usepackage{stackrel}

\newtheorem{theorem}{Theorem}[section]

\numberwithin{theorem}{section}
\newtheorem{lemma}[theorem]{Lemma}
\newtheorem{exam}[theorem]{Example}
\newtheorem{remark}[theorem]{Remark}
\newtheorem{defn}[theorem]{Definition}
\newtheorem{corollary}[theorem]{Corollary}

\newenvironment{proof}{\noindent\\ \noindent\relax{\sc
     Proof}}{{\samepage\par\nopagebreak\hbox
     to\hsize{\hfill$\Box$}}}
\newcommand{\be}{\begin{equation}} \newcommand{\ee}{\end{equation}}
\newcommand{\bd}{\begin{displaymath}} \newcommand{\ed}{\end{displaymath}}
\newcommand{\ben}{\begin{enumerate}} \newcommand{\een}{\end{enumerate}}
\newcommand{\bi}{\begin{itemize}} \newcommand{\ei}{\end{itemize}}

\newcommand{\ud}{\mathrm{d}}
\newcommand{\E}[1]{\operatorname{E}\left[ #1 \right]}

\newcommand{\Expectation}[1]{\operatorname{E}\left[ #1 \right]}

\newcommand{\variance}[1]{\operatorname{Var}\left[ #1 \right]}
\newcommand{\covariance}[2]{\operatorname{Cov}\left[ #1,#2 \right]}

\newcommand*\circledchar[1]{
\begin{tikzpicture}
\node[draw,circle,inner sep=2.5pt] {#1};
\end{tikzpicture}}
\newcounter{rEPP}

\newcounter{rEexpPsi}

\newcounter{rEtauP}

\begin{document}


\title{A Central Limit Theorem for Punctuated Equilibrium}

\date{}

\author{K. Bartoszek \thanks{{krzysztof.bartoszek@liu.se, krzbar@protonmail.ch, 
Department of Computer and Information Science, Link\"oping University, 581 83 Link\"oping, Sweden }}}

\maketitle

\begin{abstract}
Current evolutionary biology models usually assume that a phenotype undergoes
gradual change. This is in stark contrast to biological intuition, which 
indicates that change can also be punctuated---the phenotype can jump.
Such a jump could especially occur at speciation, i.e. dramatic 
change occurs that drives the species apart. Here we derive 
a Central Limit Theorem for punctuated equilibrium. We show that,
if adaptation is fast, for weak convergence to normality to hold, 
the variability in the occurrence of change has to disappear with time.
\end{abstract}

\noindent
Keywords : 
Branching diffusion process, Conditioned branching process,
Central Limit Theorem, L\'evy process,
Punctuated equilibrium, Yule--Ornstein--Uhlenbeck with jumps process
\\~\\ \noindent
AMS subject classification : 
60F05, 60J70, 60J85, 62P10, 92B99

\section{Introduction}

A long--standing debate in evolutionary biology is whether 
changes take place at times of speciation 
(punctuated equilibrium \citet[][]{NEldSGou1972,SGouNEld1993})
or gradually over time (phyletic gradualism, see references in \citet[][]{NEldSGou1972}). 
Phyletic gradualism is in line with Darwin's original envisioning of evolution
(\citet{NEldSGou1972}).
On the other hand, the theory of punctuated equilibrium was an answer to  what
fossil data was indicating (\citet{NEldSGou1972,SGouNEld1977,SGouNEld1993}).
A complete unbroken fossil series was rarely observed, rather
distinct forms separated by
long periods of stability (\citet{NEldSGou1972}). 
Darwin saw 
``the fossil record more 
as an embarrassment than as an aid to his theory'' (\citet{NEldSGou1972})
in the discussions with Falconer at the birth of the theory of evolution.
Evolution with jumps was proposed
under the name ``quantum evolution'' (\citet{GSim1947}) to the scientific community.
However, only later  (\citet{NEldSGou1972}) was punctuated equilibrium re--introduced 
into contemporary mainstream evolutionary theory.
Mathematical modelling of punctuated evolution on phylogenetic trees
seems to be still in its infancy 
(but see \citet{FBok2002,FBok2008,FBok2010,TMatFBok2008,AMooDSch1998,AMooSVamDSch1999}).
The main reason is that we do not seem to have sufficient understanding
of the stochastic properties of these models. 
An attempt was made in this direction (\citet{KBar2014})---to derive the tips' mean,
variance, covariance and interspecies correlation for a branching
Ornstein--Uhlenbeck (OU) process with jumps at speciation, alongside 
a way of quantitatively
assessing the effect of both types of evolution. 
Very recently \citet{PBasMMarSRob2017} considered the problem
from a statistical point of view and proposed an Expectation--Maximization
algorithm for a phylogenetic Brownian motion with jumps
and OU with jumps in the drift function model. This work
is very important to indicate as it includes estimation software
for a punctuated equilibrium model, something not readily available earlier.
\citet{SPenMHofAOli2017} also recently looked into estimation
procedures for bifurcating Markov chains.

Combining jumps with an OU process is attractive
from a biological point of view. It is consistent with the
original motivation behind punctuated equilibrium.
At branching, dramatic events occur that drive species 
apart. But then stasis between these jumps 
does not mean that no change takes place, rather
that during it ``fluctuations of little or no accumulated consequence'' occur (\citet{SGouNEld1993}).
The OU process fits into this idea because
if the adaptation rate is large enough, then the process reaches stationarity 
very quickly and oscillates around the optimal state.
This then can be interpreted as stasis between the jumps---the small
fluctuations.  
\citet{EMay1982}  supports this sort of reasoning by hypothesizing  that 
``The further removed in time a species from the original speciation 
event that originated it, the more its genotype will
have become stabilized and the more it is likely to resist change.'' 
It should perhaps be noted at this point, that a Reviewer pointed out that the modelling
approach presented in this work is not the same as the ``classical
view of punctuated equilibrium''. One would expect the jump to take place
in the direction of the optimum trait value. However, here at speciation
the jump is allowed to take place in any direction, also away from the optimum.
Then, after the jump, a relaxation period occurs and the trait
is allowed to evolve back to the optimum. Such a view on the jumps
is similar to e.g. \citet{FBok2002,FBok2008}'s modelling approach, however, there
the Brownian motion (BM) process was considered so no optimum parameter was present.
All the presented here results, concern the balance between the relaxation phenomena
and the jumps' magnitudes and chances of occurring. One way of maybe thinking
about jumps going against the optimum, is that at the speciation event
a short--lived (as afterwords evolution goes again in the direction
of the previous optimum), random environmental niche appeared that allowed part of the species'
population to break--off and form a new species. However, to make this any more formal
one would have to link it with models for the environment, fitness and trait
dependent speciation, which is beyond the scope of this paper. 

In this work we build up on previous results (\citet{KBar2014,KBarSSag2015}) and
study in detail the asymptotic behaviour of the average of the tip
values of a branching OU process, with jumps at speciation points,
evolving on a pure birth tree. 
To the best of our knowledge the work here is one of the 
first to consider the effect of jumps on a 
branching OU process in a phylogenetic context (but also look at \citet{PBasMMarSRob2017}).
It is possible that some of the results could be special
subcases of general results on branching Markov
processes (e.g. \citet{RAbrJDel2012,VBanJDelLMarVTra2011,BCloMHai2015,AMar2019,YRenRSonRZha2014,YRenRSonRZha2015,YRenRSonRZha2017}).
However, these studies use a very heavy functional analysis
apparatus, which unlike the direct one here, 
could be difficult for the applied reader. \citet{VBanJDelLMarVTra2011,JGuy2007,SPenHDjeAGui2014}'s works
are worth pointing out as they connect their results on bifurcating Markov processes  
with biological settings where branching phenomena are applicable, e.g. cell growth. 

In the work here we can observe the (well known)
competition between the tree's speciation and OU's adaptation (drift)
rates, resulting in a phase transition when the latter is half the former
(the same as in the no jumps case  
\citet{RAdaPMil2014,RAdaPMil2015,CAneLHoSRoc2017,KBarSSag2015}).
We show here that if variability in jump occurrences disappears with time 
or the model is in the critical regime (plus a bound assumption on the jumps' magnitude and chances of occurring),
then the contemporary sample mean will be asymptotically normally distributed. Otherwise
the weak limit can be characterized as a ``normal distribution with a
random variance''. Such probabilistic characterizations are important 
as tools for punctuated phylogenetic models are starting to be developed
(e.g. \citet{PBasMMarSRob2017}). 
This is partially due to an uncertainty of what is estimable, especially
whether the contribution of gradual and punctuated change
may be disentangled (but \citet{FBok2010} indicates that they should be distinguishable).
Large sample size distributional approximations will allow for choosing seeds
for numerical maximum likelihood procedures and sanity checks if the results of numerical
procedures make sense. For example in the one--dimensional OU case it is known that (for a Yule tree)
the sample average is a consistent estimator of the long term mean and the sample variance of 
the OU process' stationary variance (\citet{KBarSSag2015}). In the BM (Yule tree) case one can have
a consistent estimator of the diffusion coefficient \citep{KBarSSag2015JTB}.
Hence, from these sample statistics one can construct starting values for numerical estimation
procedures (as e.g. mvSLOUCH does now, \citep{KBarJPiePMosSAndTHan2012}).

Often a key ingredient in studying branching Markov processes is a 
``Many--to--One'' formula---the law of the trait
of an uniformly sampled individual in an ``average'' population (e.g. \citet{AMar2019}). The approach
in this paper is that on the one hand we condition
on the population size, $n$, but then to obtain the law (and its limit) of the contemporary
population, we consider moments of uniformly sampled species and the covariance between a 
uniformly sampled pair of species.

The strategy to study the limit behaviour is to first condition
on a realization of the Yule tree and jump pattern (on which branches after speciation did the jump
take place). This is, as 
conditional on the phylogeny and jump locations, the collection of the contemporary tips' trait values
will have a multivariate normal distribution, and hence their sample average will be normally distributed. 
We are able to represent (under the above conditioning) the variance of the sample average in terms of 
transformations of the number of speciation events on randomly selected lineage, 
time to coalescent of randomly selected pair of tips
and the number of common speciation events for a randomly selected pair of tips.
We consider the conditional (on the tree and jump pattern) expectation
of these transformations and then look at the rate of decay to $0$ 
of the variances of these conditional expectations. If this rate of decay
is fast enough, then they will converge to a constant and 
the normality of an appropriately scaled average of tips species will be retained in the limit.
Very briefly this rate of decay depends on how the product of the probability and variance
of the jump behaves along the nodes of the tree. We do not necessarily 
assume (as previously in \citet{KBar2014}) that the jumps are homogeneous on the whole tree.

First, in Section \ref{secNotation} we provide a series of formal definitions 
that introduce key random variables associated with the phylogeny
that are necessary for this study.
Afterwords, in Section \ref{secModelIntro} we introduce the considered probabilistic model
and the concepts from Section \ref{secNotation} in a more intuitive manner.
Then, in Section \ref{secMainRes} we present the main results. Section 
\ref{secKeyConvLem} is devoted to a series of technical convergence lemmata that
characterize the speed of decay of the effect of jumps 
on the variance and covariance of tip species. Finally, in Section \ref{secProof}
we calculate the first two moments of a scaled sample average,
introduce a particular random variable related to the model and put this
together with the previous convergence lemmata to prove the 
Central Limit Theorems (CLTs) of this paper.
It should be acknowledged at this point that 
in the original arXiv preprint of this paper  the convergence
to normality results were stated in an incomplete manner.
In particular the limiting normality in the critical regime was not described correctly.
The current characterization was noticed during the 
collaboration with Torkel Erhardsson \citep{KBarTErh2019arXiv} and 
more details on the previous mischaracterization can be found in Remark \ref{remTH}.

\section{Notation}\label{secNotation}
We first introduce two separate labellings for the tip and internal nodes
of the tree. Let the origin of the tree have label ``$0$''. 
Next we label from ``$1$'' to ``$n-1$'' the internal nodes
of the tree in their temporal order of appearance. The root
is ``$1$'', the node corresponding to the second speciation event
is ``$2$'' and so on. We label the tips of the
tree from ``$1$'' to ``$n$'' in an arbitrary fashion.
This double usage of the numbers ``$1$'' to ``$n-1$''
does not generate any confusion as it will always
be clear whether one refers to a tip or internal node.

\begin{defn}
~

$$
N_{\mathrm{Tip}}(t)=\{\mathrm{set~ of~ tip~ nodes~ at~ time}~t \}
$$
\end{defn}

\begin{defn}
~

$$
U^{(n)}=\inf\{ t\ge 0: \vert N_{\mathrm{Tip}}(t) \vert =n\},
$$
where $\vert A \vert$ denotes the cardinality of set $A$.
\end{defn}

\begin{defn}
For $i\in N_{\mathrm{Tip}}(U^{(n)})$,

$$
\Upsilon^{(i,n)}: \mathrm{number~of~nodes~on~the~path~from~the~root}~(\mathrm{internal~node}~1\mathrm{,~including~it})~\mathrm{to~tip~node}~i 
$$
\end{defn}

\begin{defn}
For $i\in N_{\mathrm{Tip}}(U^{(n)})$, define the finite sequence of length $\Upsilon^{(i,n)}$
as

$$
\mathrm{I}^{(i,n)}=\left(\mathrm{I}_{j}^{(i,n)}:\mathrm{I}_{j}^{(i,n)} 
\mathrm{is~a~node~on~the~root~to~tip~node}~i~\mathrm{path~and}~
\mathrm{I}_{j}^{(i,n)}<\mathrm{I}_{k}^{(i,n)}~\mathrm{for}~1\le j < k \le \Upsilon^{(i,n)}
\right)_{j=1}^{\Upsilon^{(i,n)}}
$$
\end{defn}

\begin{defn}
For $i\in N_{\mathrm{Tip}}(U^{(n)})$ and $r \in \{1,\ldots,n-1 \}$, 
let $\mathbf{1}^{(i,n)}_{r}$ be a binary random variable such that

$$
\mathbf{1}^{(i,n)}_{r}=1~\mathrm{iff}~ r \in \mathrm{I}^{(i,n)},
$$
where the $\in$ should be understood in the natural way that there exists a position $j$ 
in the sequence $\mathrm{I}^{(i,n)}$ s.t. $\mathrm{I}^{(i,n)}_{j}=r$.
\end{defn}

\begin{defn}
For $i\in N_{\mathrm{Tip}}(U^{(n)})$ and $r \in \{1,\ldots,\Upsilon^{(i,n)} \}$, 
let $J^{(i,n)}_{r}$ be a binary random variable equalling $1$
iff a jump (an event that will be discussed in more detail Section \ref{secModelIntro}) 
took place just after the $r$--th speciation event
in the sequence $\mathrm{I}^{(i,n)}$.
\end{defn}

\begin{defn}
For $i\in N_{\mathrm{Tip}}(U^{(n)})$ and $r \in \{1,\ldots,n-1 \}$, 
let $Z^{(i,n)}_{r}$ be a binary random variable equalling $1$
iff $\mathbf{1}^{(i,n)}_{r}=1$ and 
$J^{(i,n)}_{k}=1$, where $\mathrm{I}^{(i,n)}_{k}=r$.
\end{defn}

\begin{defn}
For $i,j\in N_{\mathrm{Tip}}(U^{(n)})$,

$$
\mathrm{I}^{(i,j,n)}=\mathrm{I}^{(i,n)}\cap \mathrm{I}^{(j,n)},
$$
where for two sequences $a=(a_{j})$ and $b=(b_{j})$ we define the operation 

$$
a \cap b = (a_{j}: a_{j}=b_{j}) 
$$
or in other words $a \cap b$ is the common prefix of sequences $a$ and $b$.
\end{defn}

\begin{defn}
For $i,j\in N_{\mathrm{Tip}}(U^{(n)})$,

$$
\upsilon^{(i,j,n)}=\vert \mathrm{I}^{(i,j,n)} \vert -1,
$$
where for a finite sequence $v$, $\vert v \vert$ means its length.
\end{defn}

\begin{remark}\label{remupsmin1}
We have the $-1$ in the above definition of $\upsilon^{(i,j,n)}$ as we are interested
in counting the speciation events that could have a jump common to both lineages.
As the jump occurs after a speciation event, the jumps connected to the coalescent
node of tip nodes $i$ and $j$ cannot affect both of these tips (see Section \ref{sbsecBrphen}).
\end{remark}

\begin{defn}
For $i,j\in N_{\mathrm{Tip}}(U^{(n)})$ and $r \in \{1,\ldots, \max(I^{(i,j,n)})-1 \}$, 
let $\mathbf{1}^{(i,j,n)}_{r}$ be a binary random variable such that

$$
\mathbf{1}^{(i,j,n)}_{r}=1~\mathrm{iff}~ r \in \mathrm{I}^{(i,j,n)}.
$$
For a sequence $a$, the operation $\max (a)$ chooses the maximum value present in the sequence.
\end{defn}

\begin{defn}
For $i,j\in N_{\mathrm{Tip}}(U^{(n)})$,

$$
\tau^{(i,j,n)}=U^{(n)}-\inf\{ t\ge 0: N_{\mathrm{Tip}}(t) = \max \left(\mathrm{I}^{(i,j,n)}\right)\}.
$$
\end{defn}

\begin{defn}
For $i,j\in N_{\mathrm{Tip}}(U^{(n)})$ and $r \in \{1,\ldots,\upsilon^{(n)}_{i,j} \}$, 
let $J^{(i,j,n)}_{r}$ be a binary random variable equalling $1$
iff $J^{(i,n)}_{r}=1$ and $J^{(j,n)}_{r}=1$.
\end{defn}

\begin{defn}
For $i,j\in N_{\mathrm{Tip}}(U^{(n)})$ and $r \in \{1,\ldots,n-1 \}$, 
let $Z^{(i,j,n)}_{r}$ be a binary random variable equalling $1$
iff $Z^{(i,n)}_{r}=1$ and $Z^{(j,n)}_{r}=1$.
\end{defn}

\begin{defn}\label{dfpairRV}
Let $R$ be uniformly distributed on $\{1,\ldots,n\}$ and
$(R,K)$ be uniformly distributed on the set of ordered pairs 
drawn from $\{1,\ldots,n\}$ (i.e. $\mathrm{Prob}((R,K)=(r,k))=\binom{n}{2}^{-1}$,
for $1\le r < k \le n$)

$$
\begin{array}{l}
\tau^{(n)}= \tau^{(R,K,n)},~~
\Upsilon^{(n)}=\Upsilon^{(R,n)},~~
\upsilon^{(n)}=\upsilon^{(R,K,n)},~~
\mathrm{I}^{(n)}=\mathrm{I}^{(R,n)},~~
\mathrm{\tilde{I}}^{(n)}=\mathrm{I}^{(R,K,n)},\\
\mathbf{1}_{i}=\mathbf{1}^{(R,n)}_{i},~~
\mathbf{\tilde{1}}_{i}=\mathbf{1}^{(R,K,n)}_{i},~~
J_{i}= J^{(R,n)}_{i},~~
\tilde{J}_{i}= J^{(R,K,n)}_{i},~~
Z_{i}= Z^{(R,n)}_{i},~~
\tilde{Z}_{i}= Z^{(R,K,n)}_{i}. 
\end{array}
$$
\end{defn}
Some of the variables defined in Defn. \ref{dfpairRV} are illustrated in Figs. \ref{figTreeNotation}, \ref{figPhiPsi}
and further described in the captions. It might be also useful to refer to \citet{KBar2014}, especially
Fig. A.$8$, therein.

\begin{remark}
For the sequences $\mathrm{I}^{(n)}$, $\mathrm{I}^{(r,n)}$, 
$\mathrm{I}^{(R,n)}$, $\mathrm{\tilde{I}}^{(n)}$, $\mathrm{I}^{(r,k,n)}$, $\mathrm{I}^{(R,K,n)}$
the $i$--th element is naturally indicated as
$\mathrm{I}_{i}^{(n)}$, $\mathrm{I}_{i}^{(r,n)}$, $\mathrm{I}_{i}^{(R,n)}$, 
$\mathrm{\tilde{I}}_{i}^{(n)}$, $\mathrm{I}_{i}^{(r,k,n)}$, $\mathrm{I}_{i}^{(R,K,n)}$
respectively.
\end{remark}

\begin{remark}
We drop the $n$ in the superscript for the random variables
$\mathbf{1}_{i}$, $\mathbf{\tilde{1}}_{i}$, $J_{i}$, $\tilde{J}_{i}$,
$Z_{i}$ and $\tilde{Z}_{i}$ as their distribution will not depend
on $n$ (see Lemma \ref{lem1i} Section \ref{secModelIntro}). In fact, in principle, there will be 
no need to distinguish between
the version with and without the tilde. However, such a
distinction will make it more clear to what one is referring to in the subsequent derivations
in this work.
\end{remark}

\section{A model for punctuated stabilizing selection}\label{secModelIntro}
\subsection{Phenotype model}
Stochastic differential equations (SDEs) are today the standard language to
model continuous traits evolving on a phylogenetic tree. The general
framework is that of a diffusion process

\be
\label{eqSDEdiff}
\ud X(t) = \mu(t,X(t))\ud t + \sigma_{a}\ud B_{t}.
\ee
The trait, $X(t)\in \mathbb{R}$, follows Eq. \eqref{eqSDEdiff} along each branch 
of the tree (with possibly branch specific parameters). 
At speciation times
this process divides into two processes evolving independently from that point. 
A workhorse of contemporary phylogenetic comparative methods (PCMs) is the OU process

\be
\label{eqSDEOU}
\ud X(t) = -\alpha(X(t)-\theta)\ud t + \sigma_{a}\ud B_{t},
\ee
where sometimes the parameters $\alpha$, $\theta$, $\sigma_{a}$ are
allowed to vary over the tree
(see e.g. \citet{KBarJPiePMosSAndTHan2012, JBeauetal2012, MButAKin2004,THan1997, VMitKBarGAsiTSta2020, VMitKBarTSta2019}).
Without loss of generality, for the purpose of the results here, we could have taken $\theta=0$.
However, we choose to retain the parameter for consistency with previous literature.
In this work we keep all the parameters ($\alpha$, $\theta$, $\sigma_{a}$) identical over the whole tree. 

The probabilistic properties (e.g. \citet{BSapJYao2005})
and statistical procedures (e.g. \citet{RAzaFDufAPet2014})
for processes with jumps have of course been developed.
 In the phylogenetic context 
there have been a few attempts to go beyond the diffusion framework
into L\'evy process, including Laplace motion,
(\citet{KBar2012,PDucCLeuSSziLHarJEasMSchDWeg2016,MLanJSchMLia2013})
and jumps at speciation points (\citet{KBar2014,PBasMMarSRob2017,FBok2003,FBok2008}).
We follow in the spirit of the latter
and consider that just after a branching point, with a probability 
$p$, independently on each daughter lineage, a jump can occur. 
It is worth underlining here a key difference of this model from the one 
considered by \citet{PBasMMarSRob2017}. Here after speciation
each daughter lineage may with probability $p$ jump (independently of the other).
In \citet{PBasMMarSRob2017}'s model, in the OU case,
the jump is not in the trait value but in the drift function, $\theta$ of Eq. \eqref{eqSDEOU}.
We assume that the jump random variable, added to the trait's value, is normally distributed with mean $0$
and variance $\sigma_{c}^{2} < \infty$. In other words, if at time $t$ there
is a speciation event, then just after it, independently for each daughter lineage, the trait
process $X(t^{+})$ will be 

\be\label{eqProcJ}
X(t^{+}) = (1-Z)X(t^{-}) + Z(X(t^{-})+\zeta),
\ee
where $X(t^{-/+})$ means the value of $X(t)$ respectively just before and 
after time $t$, $Z$ is a binary random variable with probability $p$ of being 
$1$ (i.e. jump occurs) and $\zeta\sim \mathcal{N}(0,\sigma_{c}^{2})$.
The parameters $p$ and $\sigma_{c}^{2}$ can, in particular, differ between speciation events.
Taking $p=0$ or $\sigma_{c}^{2}=0$ we recover the YOU without jumps model and results
\citep[described by][]{KBarSSag2015}.

\subsection{The branching phenotype}\label{sbsecBrphen}
In this work we consider a fundamental model of phylogenetic tree growth --- the 
conditioned on number of tip species pure birth process (Yule tree).
We first make the notation from Section \ref{secNotation} more intuitive, 
illustrating it also in Figs. \ref{figTreeNotation} and \ref{figPhiPsi}
(see also  \citet{KBar2014,KBarSSag2015,SSagKBar2012}).
We consider a tree that has $n$ tip species.
Let $U^{(n)}$ be the tree height, $\tau^{(n)}$ the time from today
(backwards) to the coalescent of a pair of randomly chosen tip species,
$\tau^{(n)}_{ij}$ the time to coalescent of tips $i$, $j$,
$\Upsilon^{(n)}$ the number of speciation events on a random lineage, 
$\upsilon^{(n)}$ the number of common speciation events for a
random pair of tips minus one
and $\upsilon^{(n)}_{ij}$
the number of common speciation events for 
tips $i$, $j$ minus one.
The jumps take place after the speciation event so
any jump associated with the speciation event that
split the two lineages, e.g. in Fig. \ref{figTreeNotation}
speciation event $2$ for the pair of lineages $A$ and $B$,
cannot be common to the the two lineages. Hence, in the 
caption Fig. \ref{figTreeNotation}, we have 
$\upsilon^{(n)}_{AB}=1$, see also Remark \ref{remupsmin1}.

Furthermore, let $I^{(n)}$ be the sequence of nodes on a randomly chosen
lineage and $J^{(n)}$ be a binary sequence indicating if a jump took
place after each respective node in the $I^{(n)}$ sequence.
Finally, let $T_{k}$ be the time
between speciation events $k$ and $k+1$, 
$p_{k}$ and $\sigma_{c,k}^{2}$ be respectively the probability 
and variance of the jump just after the $k$--th speciation event
on each daughter lineage. It is worth recalling that (unlike in \citet{PBasMMarSRob2017}'s model)
both daughter lineages may jump independently of each other. It is also
worth reminding the reader that previously (in \citet{KBar2014})
the jumps were homogeneous over the tree, in this manuscript
we allow their properties to vary with the nodes of the tree.

The following simple, yet very powerful, lemma 
comes from the uniformity of the choice of pair to coalesce at the $i$--th
speciation event in the backward  description of the Yule process. The 
proof can be found in \citet{KBar2014} on p. $45$ (by no means do I claim this well known result as my own).

\begin{lemma}\label{lem1i}
Consider for a Yule tree the indicator random variables $\mathbf{1}_{i}$ that the $i$--th (counting from the root)
speciation event lies on a randomly selected lineage and 
$\tilde{\mathbf{1}}_{i}$ that the $i$--th speciation event lies on the path from the origin to 
the most recent common ancestor of a randomly selected pair of tips.
Then for all $i\in \{1,\ldots,n-1\}$

$$
\E{\tilde{\mathbf{1}}_{i}} = \E{\mathbf{1}_{i}} = \mathrm{Prob}(\mathbf{1}_{i}=1)=
\frac{2}{i+1}.
$$
\end{lemma}

\begin{figure}[!h]
\centering
\includegraphics[width=0.8\textwidth]{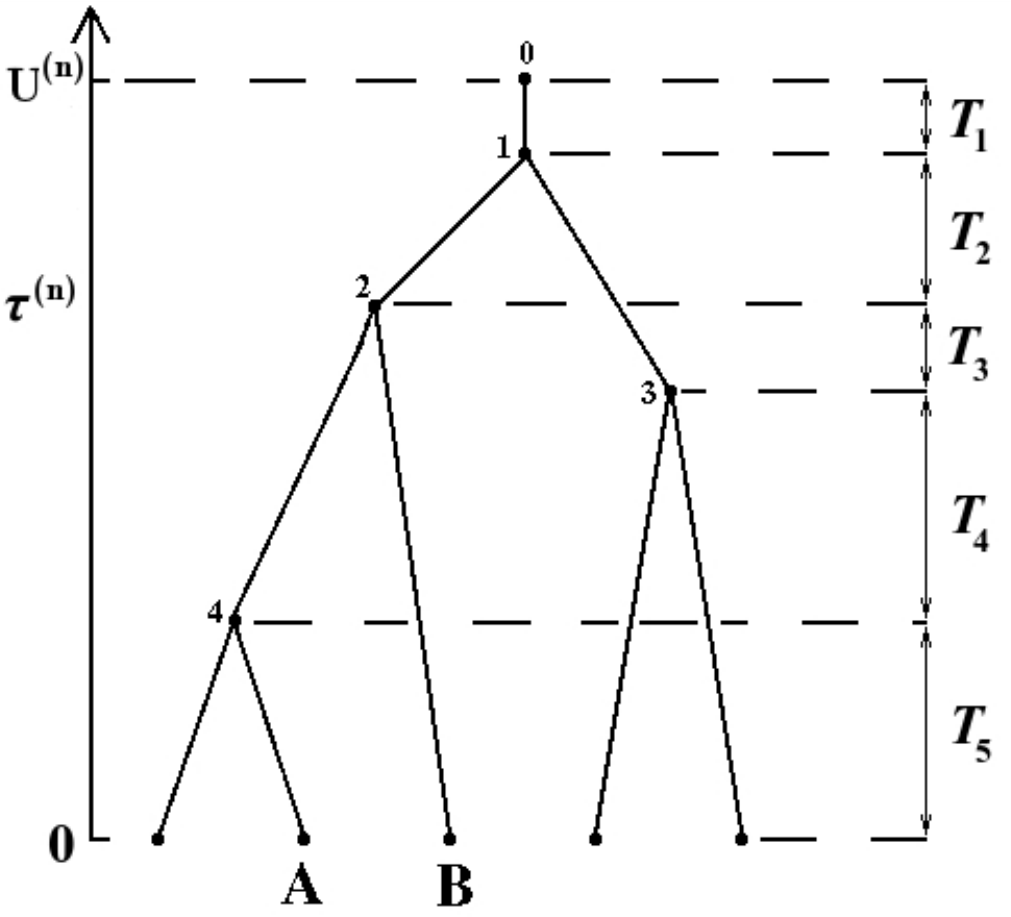}
\caption{
A pure--birth tree with the various time components marked on it. 
If we ``randomly sample'' node ``A'', then $\Upsilon^{(n)}=3$ and
the indexes of the speciation events on this random lineage are
$\mathrm{I}^{(n)}_{3}=4$, 
$\mathrm{I}^{(n)}_{2}=2$ and $\mathrm{I}^{(n)}_{1}=1$.
Notice that 
$\mathrm{I}^{(n)}_{1}=1$ always.
The between speciation times on this lineage are
$T_{1}$, $T_{2}$, $T_{3}+T_{4}$ and $T_{5}$.
If we ``randomly sample'' the pair of extant 
species ``A'' and ``B'', then $\upsilon^{(n)}=1$ and the two nodes coalesced at time 
$\tau^{(n)}=T_{3}+T_{4}+T_{5}$.
The random index of their joint speciation event is 
$\mathrm{\tilde{I}}_{1}=1$.
See also Fig. \ref{figPhiPsi} and \citet{KBar2014}'s Fig. A.$8$. for a more detailed discussion on relevant notation.
The internal node labellings $0$--$4$ are marked on the tree. The OUj process
evolves along the branches of the tree and we only observe the trait
values at the $n$ tips. For given tip, say ``A'' the value of the trait process
will be denoted $X^{(n)}_{A}$. Of course here $n=5$.
}
\label{figTreeNotation}
\end{figure}

We called the model a conditioned one.
By conditioning we consider stopping the tree growth just before
the $n+1$ species occurs, or just before the $n$--th speciation event. 
Therefore, the tree's height $U^{(n)}$ is a random stopping time. 
The asymptotics considered in this work are when $n\to \infty$. 

The key model parameter describing the tree component is $\lambda$, the birth rate. At the start,
the process starts with a single particle and then splits with rate $\lambda$.
Its descendants behave in the same manner. Without loss 
generality we take $\lambda=1$, as this is equivalent to rescaling time.

In the context of phylogenetic
methods this branching process has been intensively studied 
(e.g. \citet{KBarSSag2015,FCraMSuc2013,AEdw1970,TGer2008a,TGer2008b,WMulFCra2015,SSagKBar2012,MSteAMcK2001}), 
hence here we will just describe its key property.
The time between speciation events $k$ and $k+1$ is exponential with
parameter $k$. This is immediate from the memoryless
property of the process and the distribution of the minimum
of $k$ i.i.d. exponential random variables. From this we obtain some important
properties of the process.
Let $H_{n} = 1+1/2+\ldots+1/n$ be the $n$--th harmonic number, $x> 0$ and 
then their expectations and Laplace transforms are (\citet{KBarSSag2015,SSagKBar2012})

$$
\begin{array}{rcl}
\Expectation{U^{(n)}} & = & H_{n},\\
\Expectation{e^{-x U^{(n)}}} & = & b_{n,x},\\
\Expectation{\tau^{(n)}} & = & \frac{n+1}{n-1}H_{n}-\frac{2}{n-1},\\
\Expectation{e^{-x \tau^{(n)}}} & = & 
\left\{
\begin{array}{cc}
\frac{2-(n+1)(x+1)b_{n,x}}{(n-1)(x-1)} & x\neq 1, \\
\frac{2}{n-1}\left(H_{n}-1 \right)-\frac{1}{n+1}& x=1,
\end{array}
\right.
\end{array}
$$
where

$$
b_{n,x} = \frac{1}{x+1}\cdots\frac{n}{n+x}=\frac{\Gamma(n+1)\Gamma(x+1)}{\Gamma(n+x+1)} \sim \Gamma(x+1)n^{-x},
$$
$\Gamma(\cdot)$ being the gamma function.

Now let $\mathcal{Y}_{n}$ be the $\sigma$--algebra that
contains information on the Yule tree and jump pattern. 
By this we mean that conditional on $\mathcal{Y}_{n}$
we know exactly how the tree looks like (esp. the interspeciation times $T_{i}$)
and we know at what parts of the tree (at which lineage(s) just after which speciation 
events) did jumps take place. The motivation behind such conditioning is that
conditional on $\mathcal{Y}_{n}$ the contemporary tips sample is a multivariate normal one. 
When one does not condition on $\mathcal{Y}_{n}$ the normality does not hold---the randomness
in the tree and presence/absence of jumps distorts normality.

\citet{KBar2014} previously studied 
the branching Ornstein--Uhlenbeck with jumps (OUj) model 
and it was shown
(but, therein for constant $p_{k}$ and $\sigma_{c,k}^{2}$ and 
therefore there was no need to condition on the jump pattern) that,
conditional on the tree height and number of tip species
the mean and variance of the trait value of tip species $r$ (out of the $n$ contemporary),
$X^{(n)}_{r}\equiv X^{(n)}_{r}(U^{(n)})$ (see also Fig. \ref{figTreeNotation}), are

\be
\begin{array}{rcl}
\Expectation{X^{(n)}_{r} \vert \mathcal{Y}_{n}} & = & \theta + e^{-\alpha U^{(n)}}(X_{0}-\theta) \\
\variance{X^{(n)}_{r}  \vert \mathcal{Y}_{n}} & = & \frac{\sigma_{a}^{2}}{2\alpha}(1-e^{-2\alpha U^{(n)}}) + 
\sum\limits_{i=1}^{\Upsilon^{(r,n)}}\sigma_{c,\mathrm{I}^{(r,n)}_{i}}^{2}J^{(r,n)}_{i}
e^{-2\alpha (T_{n}+\ldots+T_{\mathrm{I}^{(r,n)}_{i}+1})},
\end{array}
\ee
$\Upsilon^{(r,n)}$, $I^{(r,n)}$ and $J^{(r,n)}$ are realizations
of the random variables $\Upsilon^{(n)}$, $I^{(n)}$ and $J^{(n)}$ when lineage $r$ is picked.
A key difference that the phylogeny brings in, is that the tip measurements
are correlated through the tree structure. One can easily show that conditional on  $\mathcal{Y}_{n}$, 
the covariance
between traits belonging to tip species $r$ and $k$, $X_{r}^{(n)}$ and $X_{k}^{(n)}$ is

\be
\covariance{X_{r}^{(n)}}{X_{k}^{(n)} \vert \mathcal{Y}_{n}} = 
\frac{\sigma_{a}^{2}}{2\alpha}(e^{-2\alpha \tau^{(r,k,n)}} - e^{-2\alpha U^{(n)}}) + 
\sum\limits_{i=1}^{\upsilon^{(r,k,n)}}\sigma_{c,\mathrm{I}^{(r,k,n)}_{i}}^{2}J^{(r,k,n)}_{i}
e^{-2\alpha (\tau^{(r,k,n)}+\ldots+T_{\mathrm{I}^{(r,k,n)}_{i}+1})},
\ee
where $J^{(r,k,n)}$, $I^{(r,k,n)}$ correspond to 
the realization of random variables $J^{(n)}$, $I^{(n)}$, but 
reduced to the common part of lineages $r$ and $k$,
while $\upsilon^{(r,k,n)}$, $\tau^{(r,k,n)}$
correspond to realizations of $\upsilon^{(n)}$, $\tau^{(n)}$
when the pair $(r,k)$ is picked. 
We will call, the considered model the 
Yule--Ornstein--Uhlenbeck with jumps (YOUj) process. 

\begin{remark}
Keeping the parameter
$\theta$ constant on the tree is not as simplifying as it might seem. Varying
$\theta$ models have been considered since the introduction of the OU process
to phylogenetic methods (\citet{THan1997}). However, it can very often happen
that the $\theta$ parameter is constant over whole clades, as these
species share a common optimum due to some common discrete characteristic. Therefore, understanding the model's
behaviour with a constant $\theta$ is a crucial first step.
Furthermore, if constant $\theta$ clades are apart far enough
one could think of them as independent samples 
and attempt to construct a test (based on normality of the species' averages)
if jumps have a significant effect 
(compare Thms. \ref{thmCLTYOUjpsConst} and \ref{thmCLTYOUjpsae0}).
For this one would have to make the very difficult to biologically justify 
assumption of constant model parameters between clades. Though, one can imagine
special situations where the levels of $\theta$ are connected to a discrete
characteristic common to many clades, e.g. fresh water or seawater. 
On the other hand CLTs and other asymptotical results for changing model
parameters and different levels of $\theta$ are an exciting
future research direction.
\end{remark}

\begin{remark}
It should be noted that the phylogeny could be introduced using 
a formal branching process approach and the 
offspring's' generating function (e.g. Ch. III.3, \citet{KAthPNey2004}).
Then, the branching trait model can be described (jointly with the tree)
as a ``Markov process in the space of integer--valued measures on $\mathbb{R}$''
(\citet{RAdaPMil2015}). However, in this work here we do not use any of the
machinery from that direction and so we refrain from defining the setup in that language
so as to avoid adding yet another layer of notation. On the other hand, the way of defining
the model used here is constructive---in the sense that it can be directly coded in
a simulation procedure.

\end{remark}
\subsection{Martingale formulation}
Our main aim is to study the asymptotic behaviour of the sample average and
it actually turns out to be easier to work with scaled trait values, for
each $r\in \{1,\ldots,n\}$,
$
Y^{(n)}_{r} = (X^{(n)}_{r}-\theta)/\sqrt{\sigma_{a}^{2}/2\alpha}.
$
Denoting 
$\delta=(X_{0}-\theta)/\sqrt{\sigma_{a}^{2}/2\alpha}$ we have

\be
\begin{array}{rcl}
\Expectation{Y^{(n)}} & = & \delta b_{n,\alpha}.
\end{array}
\ee
The initial condition of course will be $Y_{0}=\delta$. 

\begin{remark}
We remark, that here it becomes evident that the specific value of $\theta$, will
not play any role in obtaining the presented here results. What only matters is the
initial displacement from $\theta$, but even this will not contribute in any way
to the rate of convergence, only as a scaling constant for the 
expectation of $\overline{Y}_{n}$
(see Proof of Thm. \ref{thmCLTYOUjpsConst}).
\end{remark}
Just as was done by \citet{KBarSSag2015}
we may construct a martingale related to the average

$$
\overline{Y}_{n} = \frac{1}{n}\sum\limits_{i=1}^{n}Y_{i}^{(n)}.
$$
It is worth pointing out that $\overline{Y}_{n}$ is observed just before the $n$--th
speciation event. An alternative formulation would be to observe it just after
the $(n-1)$--st speciation event.
Then (cf. Lemma $10$ of \citet[][]{KBarSSag2015}), we define

$$
H_{n}:=(n+1)e^{(\alpha-1)U^{(n)}}\overline{Y}_{n},~~n\ge 0.
$$
This is a martingale with respect to $\mathcal{F}_{n}$, the $\sigma$--algebra
containing information on the Yule $n$--tree and the phenotype's evolution,
i.e. $\mathcal{F}_{n}=\sigma(\mathcal{Y}_{n},Y_{1},\ldots,Y_{n})$.

\section{Asymptotic regimes --- main results}\label{secMainRes}
Branching Ornstein--Uhlenbeck models commonly have three asymptotic regimes 
(\citet{RAdaPMil2014,RAdaPMil2015,CAneLHoSRoc2017,KBar2014,KBarSSag2015, YRenRSonRZha2014,  YRenRSonRZha2015}).
The dependency between the adaptation rate $\alpha$
and branching rate $\lambda=1$ governs in which regime the process is. 
If $\alpha > 1/2$, then the contemporary sample is similar to an i.i.d. sample,
in the critical case, $\alpha =  1/2$, we can, after appropriate rescaling,
still recover the ``near'' i.i.d. behaviour
and if $0 < \alpha < 1/2$, then the process has ``long memory''
(``local correlations dominate over the OU's ergodic properties'', \citet{RAdaPMil2014,RAdaPMil2015}).
In the context considered here by ``near'' and ``similar'' to i.i.d. we mean that the resulting CLTs resemble those
of an i.i.d. sample. For example the limit distribution of the normalized sample average
in the $\alpha>0.5$ YOU regime \citep[Thm. $1$ in][]{KBarSSag2015} is
$\mathcal{N}(0,(2\alpha+1)/(2\alpha-1))$ and taking $\alpha\to \infty$
we obtain the classical $\mathcal{N}(0,1)$ limit (as intuition
could suggest with instantaneous adaptation).
In the YOUj setup the same three asymptotic regimes can be observed, even though
\citet{RAdaPMil2014,RAdaPMil2015,YRenRSonRZha2014,YRenRSonRZha2015} 
assume that
the tree is observed at a given time point, $t$,
with $n_{t}$ being random. In what follows here, 
the constant $C$ may change between (in)equalities. It may in particular depend on 
$\alpha$. We illustrate the below Theorems in Fig. \ref{figYOUjCLT}.

We consider the process $\overline{Y}_{n} = (\overline{X}_{n}-\theta)/\sqrt{\sigma_{a}^{2}/2\alpha}$ 
which is the the normalized sample mean of the YOUj process with $\overline{Y}_{0}=\delta$. 
The next two Theorems consider its, 
depending on $\alpha$, asymptotic with $n$ behaviour.

\begin{theorem}\label{thmCLTYOUjpsConst}
Assume that the jump probabilities and jump variances are constant equalling $p$ and
$\sigma_{c}^{2}< \infty$ respectively.
\begin{enumerate}[label=(\Roman*)]
\item \label{rCLTthmpIpConst} If $0.5<\alpha$ and $0<p<1$,
then the conditional variance of the scaled sample mean
$\sigma_{n}^{2}:=n\variance{\overline{Y}_{n} \vert \mathcal{Y}_{n}}$
converges in $\mathbb{P}$ to a finite mean and variance random variable $\sigma_{\infty}^{2}$.
The scaled sample mean,
$\sqrt{n}~\overline{Y}_{n}$ converges weakly to random variable
whose characteristic function can be expressed in terms of the Laplace transform of 
$\sigma_{\infty}^{2}$

$$
\forall_{x \in \mathbb{R}} \lim\limits_{n \to \infty} \phi_{\sqrt{n}~\overline{Y}_{n}}(x) = \mathcal{L}(\sigma_{\infty}^{2})(x^{2}/2).
$$
\item \label{rCLTthmpIIpConst} If $0.5=\alpha$, 
then $\sqrt{(n/\ln n)}~\overline{Y}_{n}$ is asymptotically normally distributed with mean $0$ and
variance $2+4p\sigma_{c}^{2}/\sigma_{a}^{2}$.
In particular the conditional variance of the scaled sample mean
$\sigma_{n}^{2}:=n\ln^{-1} n\variance{\overline{Y}_{n} \vert \mathcal{Y}_{n}}$
converges in $L^{2}$ (and hence in $\mathbb{P}$) to the constant $2+4p\sigma_{c}^{2}/\sigma_{a}^{2}$.
\item \label{rCLTthmpIIIpConst} If $0<\alpha <0.5$, then 
$n^{\alpha}\overline{Y}_{n}$ converges almost surely and in $L^{2}$ to a random variable
$Y_{\alpha,\delta}$ with finite first two moments.
\end{enumerate}
\end{theorem}

\begin{remark}
For the a.s. and $L^{2}$ convergence to hold in Part \ref{rCLTthmpIIIpConst}, it suffices that the 
sequence of jump variances is bounded.
Of course, the first two moments will differ if the jump variance is not constant.
\end{remark}

\begin{remark}
After this remark we will define the concept of a sequence converging to $0$
with density $1$. Should the reader find it easier, they may forget that the 
sequence converges with density $1$, but think of the sequence simply
converging to $0$. The condition of convergence with density $1$ is 
a technicality that through ergodic theory allows us to slightly
weaken the assumptions of the theorem that gives a normal limit.
\end{remark}
\begin{defn}
A subset $E \subset \mathbb{N}$ of 
positive integers is said to have density $0$ (e.g. \citet{KPet1983}) if 

$$
\lim\limits_{n \to \infty}\frac{1}{n}\sum\limits_{k=0}^{n-1}\chi_{E}(k) =0,
$$
where $\chi_{E}(\cdot)$ is the indicator function of the set $E$.
\end{defn}

\begin{defn}
A sequence $a_{n}$ converges to $0$ with density $1$ if 
there exists a subset $E\subset \mathcal{N}$ of 
density $0$ such that 

$$
\lim\limits_{n \to \infty,n \notin E}a_{n} =0.
$$
\end{defn}

\begin{theorem}\label{thmCLTYOUjpsae0}
Assume that the sequence $\{\sigma_{c,k}^{4}p_{k}\}$ is bounded.
Then, depending on $\alpha$ the process $\overline{Y}_{n}$ has the following 
asymptotic with $n$ behaviour.

\begin{enumerate}[label=(\Roman*)]
\item \label{rCLTthmpI} If $0.5<\alpha$, 
$\sigma_{c,k}^{4}p_{k}(1-p_{k})$ goes to $0$ with density $1$
and the sequences $\{\sigma_{c,k}^{2}\}$, $\{p_{k}\}$ are such that
the sequences of expectations

$$
\begin{array}{l}
\Expectation{
\sum\limits_{k=1}^{\Upsilon^{(n)}} \sigma_{c,\mathrm{I}^{(n)}_{k}}^{2}J_{k}e^{-2\alpha (T_{n}+\ldots+T_{\mathrm{I}^{(n)}_{k}+1})}
}\to \sigma_{\Upsilon}^{2}
\\
n\Expectation{
\sum\limits_{k=1}^{\upsilon^{(n)}}\sigma_{c,\mathrm{\tilde{I}}^{(n)}_{k}}^{2}\tilde{J}_{k}e^{-2\alpha (\tau^{(n)}+\ldots+T_{\mathrm{\tilde{I}}^{(n)}_{k}+1})}
} \to \sigma_{\upsilon}^{2}
\end{array}
$$
converge, then the process $\sqrt{n}~\overline{Y}_{n}$ is asymptotically normally distributed with mean $0$ and
variance $(2\alpha+1)/(2\alpha-1)+(\sigma_{\Upsilon}^{2}+\sigma_{\upsilon}^{2})/(\sigma_{a}^{2}/(2\alpha))$.
\item \label{rCLTthmpII} If $0.5=\alpha$, 
and the sequences $\{\sigma_{c,k}^{2}\}$, $\{p_{k}\}$ are such that
the sequence of expectations

$$
(n\ln^{-1}n)\Expectation{
\sum\limits_{k=1}^{\upsilon^{(n)}}\sigma_{c,\mathrm{\tilde{I}}^{(n)}_{k}}^{2}\tilde{J}_{k}e^{- (\tau^{(n)}+\ldots+T_{\mathrm{\tilde{I}}^{(n)}_{k}+1})}
} \to \sigma_{\upsilon}^{2}
$$
converges,
then $\sqrt{(n/\ln n)}~\overline{Y}_{n}$ is asymptotically normally distributed with mean $0$ and
variance $2+\sigma_{\upsilon}^{2}/\sigma_{a}^{2}$.
\end{enumerate}
\end{theorem}
It is worth pointing out that Thm. \ref{thmCLTYOUjpsae0} covers the extreme cases $p=0$ and $p=1$.
The convergence conditions on the expectations look rather daunting, however they will 
simplify very compactly if $\sigma_{c,k}^{2}$ and $p_{k}$ are constant or 
$\sigma_{c,k}^{4}p_{k} \to 0$ (with density $1$). These we discuss after the proof of the theorem,
when we also mention why the  assumptions on these expectations are necessary.

\begin{remark}\label{remTH}
In the original arXiv preprint of this paper it was stated that convergence to normality
in the $\alpha \ge 0.5$ regimes will only take place if 
$\sigma_{c,k}^{4}p_{k}$ is bounded and goes to $0$ with density $1$.
Normality in the $\alpha=0.5$ and $p_{k}=1$ regimes was noticed thanks
to the collaboration with Torkel Erhardsson
\citep{KBarTErh2019arXiv} and then, the results and proofs in this manuscript
were adjusted. 
\end{remark}

\begin{remark}
The assumption $\sigma_{c,k}^{4}p_{k}(1-p_{k}) \to 0$ with density $1$ 
is an essential one for the limit to be a normal distribution, when $\alpha > 0.5$. This is 
visible from the proof of Lemma \ref{lemvarEexpPhi}. In fact, this is the
key difference that the jumps bring in---if their magnitude or
their uncertainty in occurrence is too large, then they will disrupt the weak convergence. 

One possible way of achieving the above condition is to keep $\sigma_{c,k}^{2}$ constant and allow $p_{k} \to 0$, 
the chance of jumping becomes smaller relative to the number of species. 
Alternatively, $\sigma_{c,k}^{2} \to 0$, which could
mean that with more and more species---smaller and smaller jumps occur at speciation.
Actually, one could intuitively think of this as biologically more realistic.
We are in the Yule, no extinction, case so with time there will be more and
more species (species here can be understood, if it helps intuition
as non--mixing, for some reason, populations). If they all live in some
spatially confined area, then as the number of species grows
there could be more and more competition. If one considers a trait
that is related to what is competed for, then 
smaller and smaller differences
in phenotype could drive the species apart. Specialization occurs and tinier
and tinier niches are filled.
This reasoning of course
further assumes that the number of individuals grows with the number 
of species. Furthermore, under the considered YOUj model the long time mean, $\theta$, is the same
for all species, so even though there is an initial displacement
(into a different niche) with time the trait will try to revert to its optimum.
Hence, the above is not aiming for making any authoritative biological
statements, nor provide an interpretation of the whole YOUj model. Rather,
it has as its goal of giving some intuition on jump variance decreasing to $0$
with time/number of species.
\end{remark}

\begin{remark}
In Thm. \ref{thmCLTYOUjpsae0} we do not consider the ``fast branching/slow adaptation'', $0<\alpha<0.5$
regime. By assuming $\sigma_{c,k}^{4}p_{k}\to 0$ with density $1$, it is possible
to make the influence of the jumps disappear asymptotically, just like in the $\alpha\ge 0.5$ case,
see Example \ref{exPsae0}. However, no further insights, than those in Thm. \ref{thmCLTYOUjpsConst}
will be readily available, similarly as \citet{KBarSSag2015} note for the YOU without jumps model. This 
is as the used here methods, do not seem to easily extend
to the $0<\alpha<0.5$ situation, beyond what is presented in this manuscript.
\end{remark}

\begin{figure}[!ht]
\begin{center}
\includegraphics[width=0.3\textwidth]{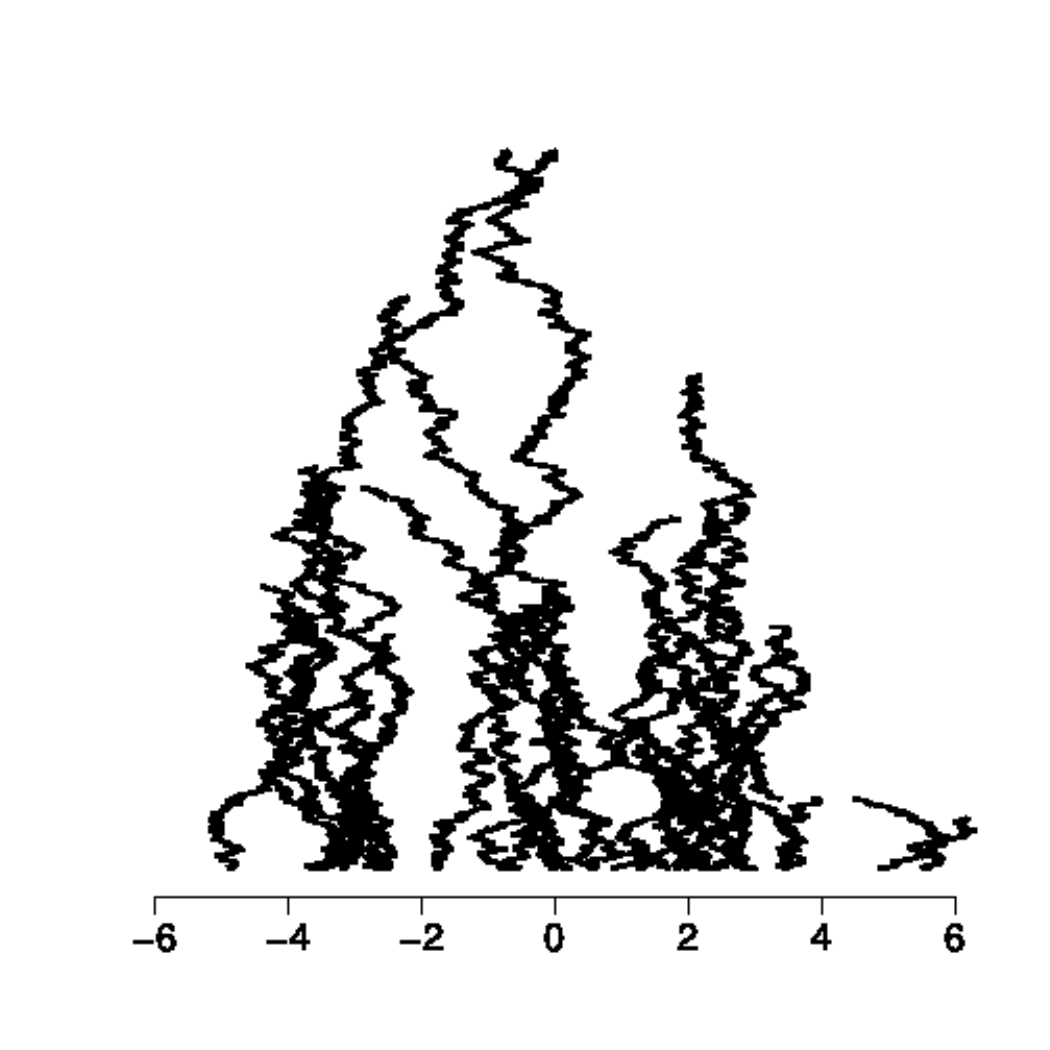}
\includegraphics[width=0.3\textwidth]{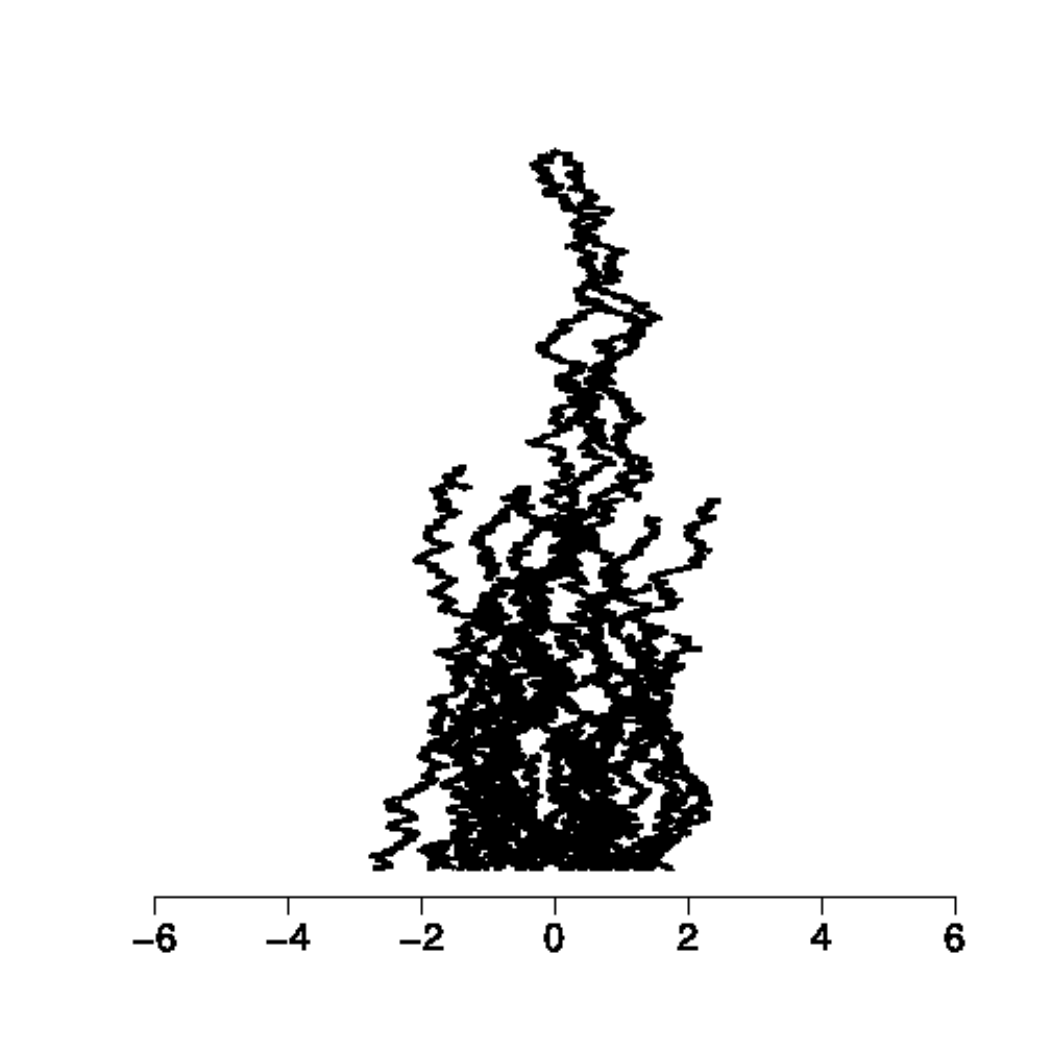}
\includegraphics[width=0.3\textwidth]{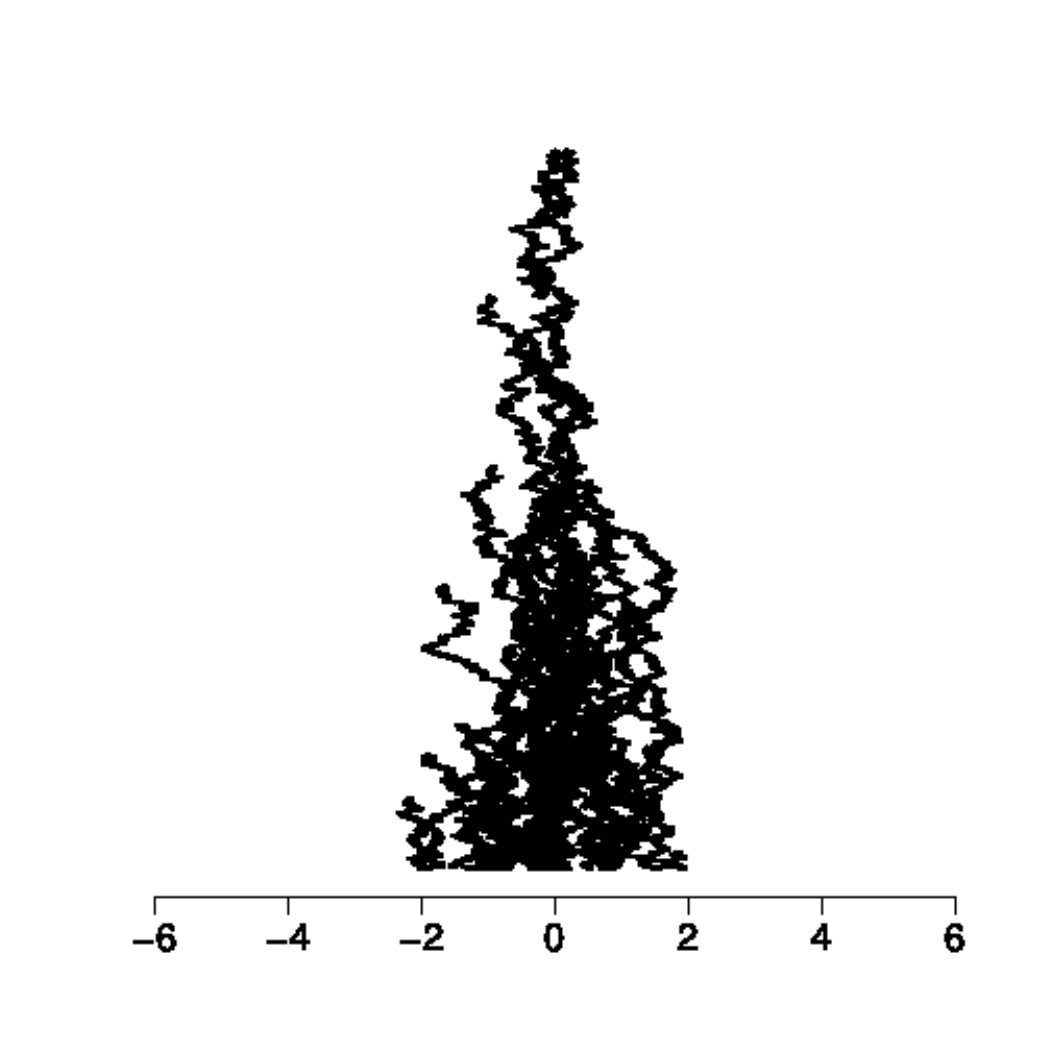} \\
\includegraphics[width=0.3\textwidth]{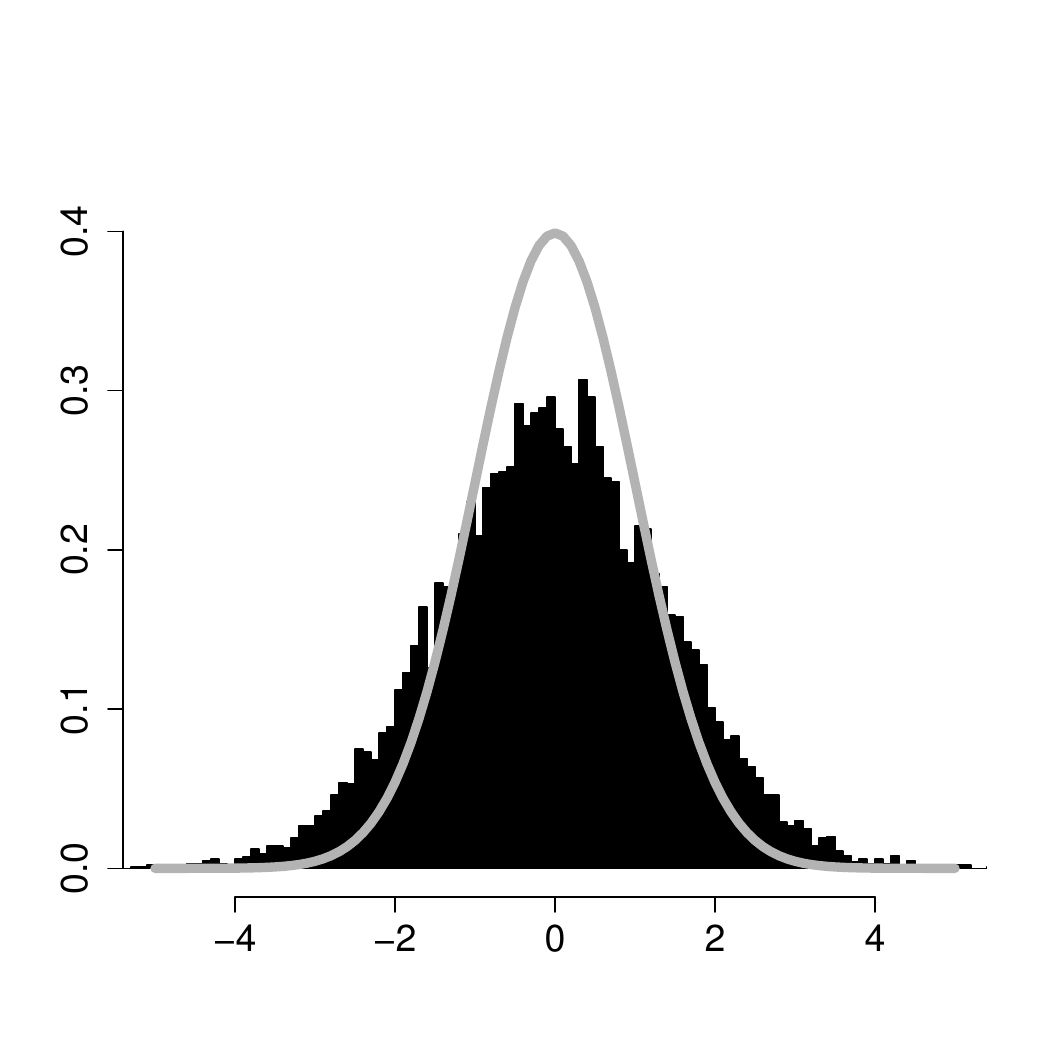}
\includegraphics[width=0.3\textwidth]{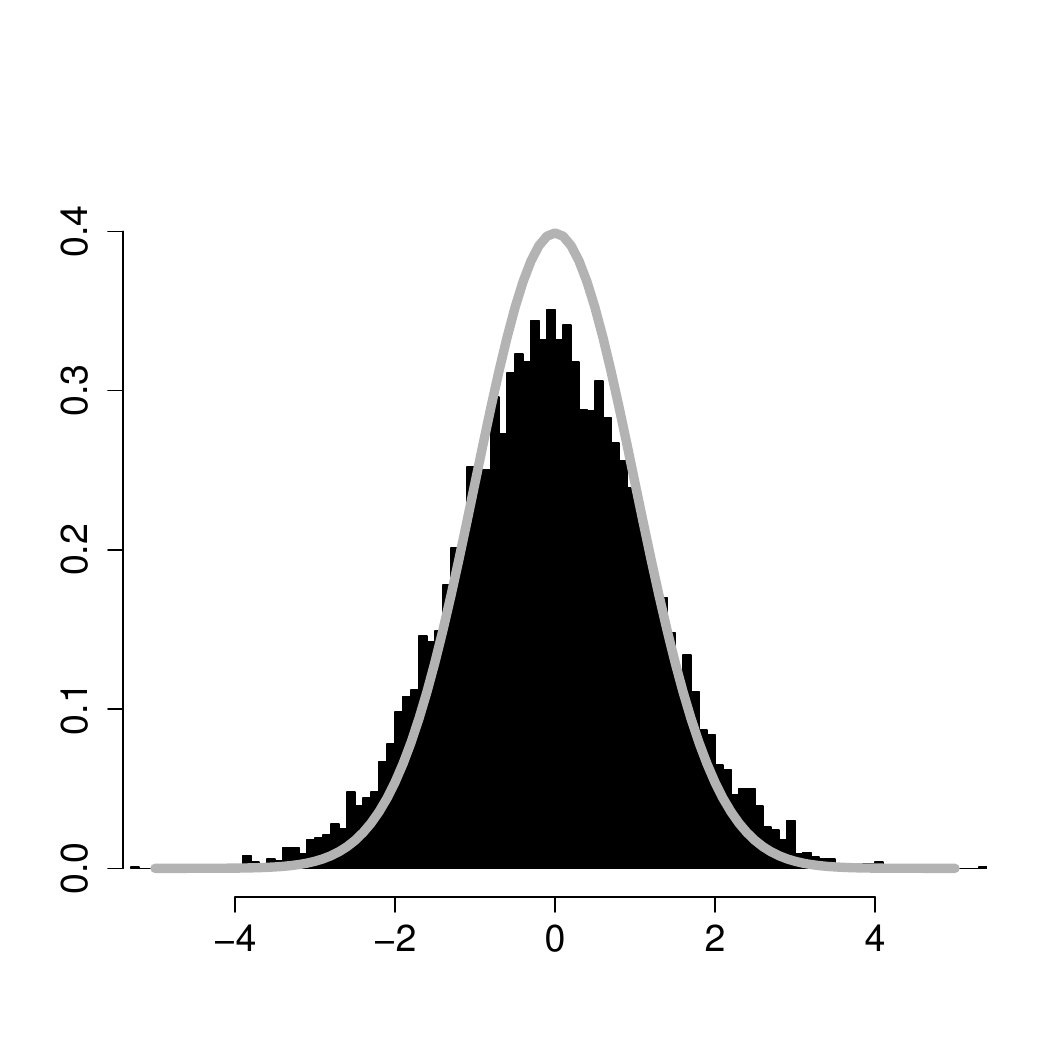}
\includegraphics[width=0.3\textwidth]{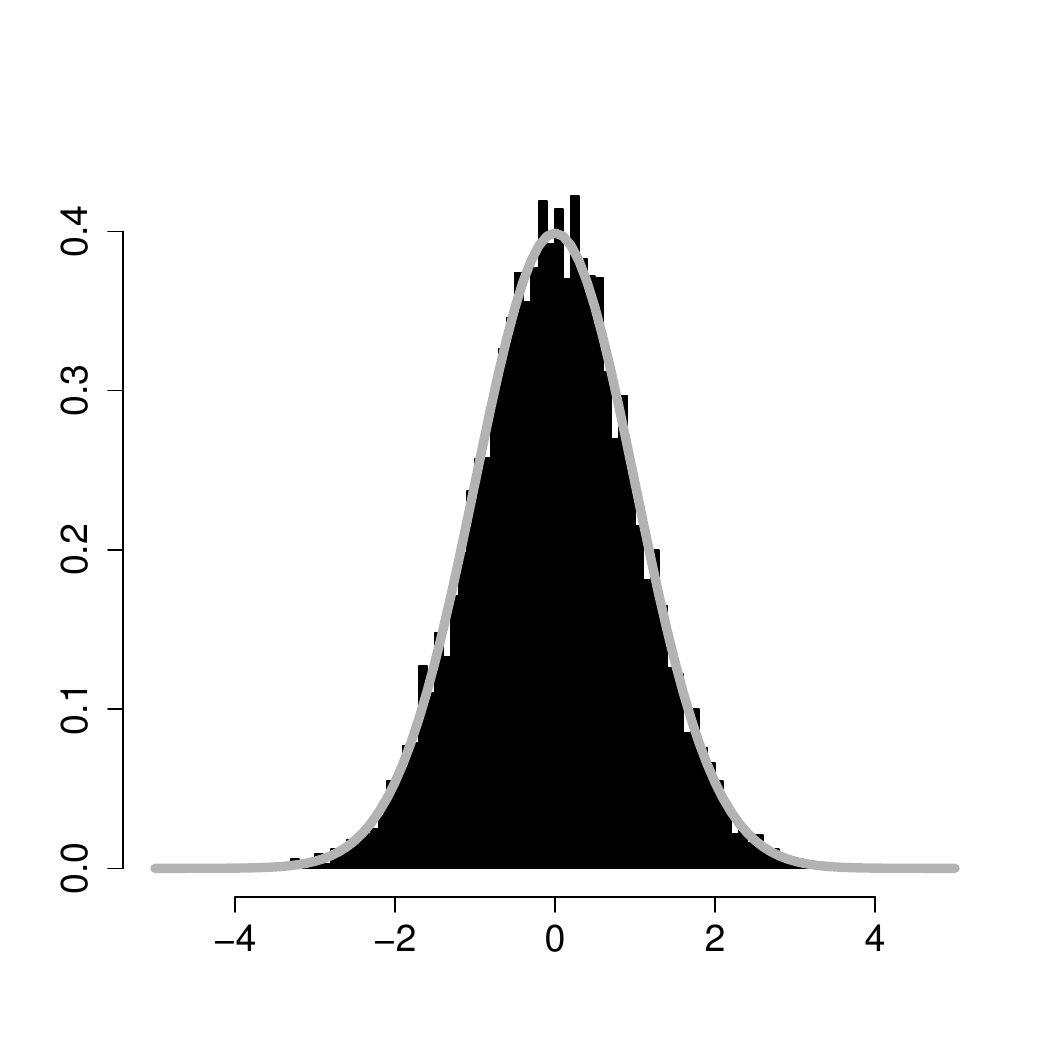} 
\caption{
Left: $\alpha=0.25$ centre: $\alpha=0.5$ and right: $\alpha=1$.
Top row: examples of simulated YOUj process trajectories,
bottom row: histograms of sample averages, 
left: scaled by $n^{0.25}\sqrt{5\Gamma(3/2)/2}$, 
centre: scaled by $\sqrt{n\ln^{-1}n/2}$, 
right: scaled by $\sqrt{n/3}$.
In all three cases, $p=0.5$, $\sigma_{c}^{2}=1$, $\sigma_{a}^{2}=1$,
$X_{0}=\theta=0$. The phylogenetic trees are pure birth trees
with $\lambda=1$ conditioned on number of tips, $n=30$ for the trajectory plots and 
$n=200$ for the histograms. 
The histograms are based on $10000$ simulated trees. 
The sample mean and variances of the scaled data in the histograms are
left: $(-0.015,2.037)$, centre: $(-0.033,1.481)$ and right: $(0.004,1.008)$.
The gray curve painted on the histograms is the standard normal distribution.
The phylogenies are simulated by the TreeSim R package
(\citet{TreeSim1,TreeSim2}) and simulations of phenotypic evolution and trajectory plots
are done by functions of the, available on CRAN, mvSLOUCH R package.
We can see that as $\alpha$ decreases the sample
variance is further away from the asymptotical $1$ (after scaling) and the histogram from
normality (though when $\alpha=0.25$ we should not expect normality). This is as with 
smaller $\alpha$ convergence is slower.
}\label{figYOUjCLT}
\end{center}
\end{figure}

\section{A series of technical lemmata}\label{secKeyConvLem}
We will now prove a series of technical lemmata describing the asymptotics of
driving components of the considered YOUj process. 
For two sequences $a_{n}$, $b_{n}$ the notation 
$a_{n} \lesssim b_{n}$ will mean that 
$a_{n}/b_{n} \to C \neq 0$ with $n$ and 
$a_{n} \le (1+o(1)) b_{n}$. 
Notice that always
when $a_{n} \lesssim b_{n}$ is used a defined or undefined constant
$C$ is present within $b_{n}$. The key property is that
the asymptotic behaviour with $n$ does not change
after the $\lesssim$ sign.
The general approach to proving these lemmata is related to that in the proof 
of \citet{KBarSSag2015}'s Lemma $11$.
What changes here is that we need to take into account
the effects of the jumps \citep[which were not considered in][]{KBarSSag2015}.
However, we noticed that there is an error in the proof of 
\citet{KBarSSag2015}'s Lemma $11$. Hence, below for the convenience 
of the reader, we do not only cite the lemma but also provide
the whole corrected proof. In Remark \ref{remLem11}, following the proof,
we briefly point the problem in the original wrong proof and
explain why it does not influence the rest of 
\citet{KBarSSag2015}'s results.

\begin{lemma}{(Lemma $11$ of \citet{KBarSSag2015})}\label{lem11}

\be
\variance{\Expectation{e^{-2\alpha \tau^{(n)}} \vert \mathcal{Y}_{n}}} =
\left \{
\begin{array}{rcl}
O(n^{-4\alpha}) & 0 < \alpha < 0.75, \\
O(n^{-3}\ln n) & \alpha = 0.75, \\
O(n^{-3}) & 0.75 < \alpha .
\end{array}
\right.
\ee
\end{lemma}
\begin{proof}
For a given realization of the Yule $n$-tree we  denote by $\tau^{(n)}_{1}$ and  
$\tau^{(n)}_{2}$ two versions of $\tau^{(n)}$ that are independent conditional 
on $\mathcal{Y}_{n}$. In other words
$\tau^{(n)}_{1}$ and  $\tau^{(n)}_{2}$  correspond to two 
independent choices of pairs of tips out of $n$ available. Conditional
on $\mathcal{Y}_{n}$ all heights in the tree are known---the randomness
is only in the choice out of the $\binom{n}{2}$ pairs or equivalently sampling out of the 
set of $n-1$ coalescent heights. We have,

$$
\Expectation{\Big(\Expectation{e^{-2\alpha \tau^{(n)}}\vert \mathcal{Y}_{n}}\Big)^2}=\Expectation{\Expectation{e^{-2\alpha (\tau^{(n)}_{1}+\tau^{(n)}_{2})}\vert \mathcal{Y}_{n}}}
=\Expectation{e^{-2\alpha (\tau^{(n)}_{1}+\tau^{(n)}_{2})}}.
$$
Let $\pi_{n,k}$ be the probability that two randomly chosen tips coalesced at the $k$--th speciation event.
We know that (cf. \citet{TreeSim1}'s proof of her Theorem 4.1, using $m$ for our $n$
or \citet{KBarSSag2015}'s Lemma 1 for a more general statement)

$$
\pi_{n,k}=2\frac{n+1}{n-1}\frac{1}{(k+1)(k+2)}.
$$
Writing

$$
f_{\alpha}(k,n):=\frac{k+1}{\alpha+k+1}\cdots\frac{n}{\alpha+n} = 
\frac{\Gamma(n+1)\Gamma(\alpha+k+1)}{\Gamma(k+1)\Gamma(\alpha+n+1)}
$$ 
and as the times between speciation events are independent and exponentially distributed
we obtain

$$
\Expectation{\Big(\Expectation{e^{-2\alpha \tau^{(n)}}\vert \mathcal{Y}_{n}}\Big)^2}
=\sum_{k=1}^{n-1}f_{4\alpha}(k,n)\pi_{n,k}^2
+2\sum\limits_{k_{1}=1}^{n-1}\sum\limits_{k_{2}=k_{1}+1}^{n-1}f_{2\alpha}(k_1,k_2)f_{4\alpha}(k_2,n)\pi_{n,k_1}\pi_{n,k_2}.
$$
On the other hand,

$$
\Big(\Expectation{e^{-2\alpha\tau^{(n)}}} \Big)^{2}=\Big(\sum\limits_{k_{1}=1}^{n-1}f_{2\alpha}(k_{1},n)\pi_{n,k_{1}}\Big)\Big(\sum\limits_{k_{2}=1}^{n-1}f_{2\alpha}(k_{2},n)\pi_{n,k_{2}}\Big).
$$
Taking the difference between the last two expressions we find 

$$
\begin{array}{ll}
 \variance{\Expectation{e^{-2\alpha \tau^{(n)}}\vert \mathcal{Y}_{n}}}=&\sum\limits_{k}\Big(f_{4\alpha}(k,n)-f_{2\alpha}(k,n)^{2}\Big)\pi_{n,k}^{2}\\
& +2\sum\limits_{k_{1}=1}^{n-1}\sum\limits_{k_{2}=k_{1}+1}^{n-1}f_{2\alpha}(k_{1},k_{2})
\Big(f_{4\alpha}(k_{2},n)-f_{2\alpha}(k_{2},n)^{2}\Big)\pi_{n,k_{1}}\pi_{n,k_{2}}.
\end{array}
$$
Noticing that we are dealing with a telescoping sum and hence using the relation 

\be \label{eqTelescopSum}
a_{1}\cdots a_{n}-b_{1}\cdots b_{n}=\sum_{i=1}^{n}b_{1}\cdots b_{i-1}(a_{i}-b_{i})a_{i+1}\cdots a_{n}
\ee
we see that it suffices to study the asymptotics of,

$$
\sum\limits_{k=1}^{n-1} A_{n,k}\pi_{n,k}^2 
~~\mathrm{and}~~
\sum\limits_{k_{1}=1}^{n-1}\sum\limits_{k_{2}=k_{1}+1}^{n-1}f_{2\alpha}(k_{1},k_{2})A_{n,k_{2}}\pi_{n,k_{1}}\pi_{n,k_{2}},
$$
where

$$A_{n,k}:=\sum\limits_{j=k+1}^{n}f_{2\alpha}(k,j)^{2} \Big(\frac{4\alpha^{2}}{j(j+4\alpha)}\Big)f_{4\alpha}(j,n).$$
To consider these two asymptotic relations we observe that for large $n$

$$A_{n,k}\lesssim 4\alpha^{2}\frac{b_{n,4\alpha}}{b_{k,2\alpha}^{2}}
\sum\limits_{j=k+1}^{n}\frac{b_{j,2\alpha}^{2}}{b_{j,4\alpha}}\frac{1}{j(4\alpha+j)}
\lesssim C\frac{b_{n,4\alpha}}{b_{k,2\alpha}^{2}} \sum\limits_{j=k+1}^{n} j^{-2}
\lesssim C\frac{b_{n,4\alpha}}{b_{k,2\alpha}^{2}}k^{-1}.$$
Now since
$\pi_{n,k} = \frac{2(n+1)}{(n-1)(k+2)(k+1)}$,
it follows 

$$
\sum\limits_{k=1}^{n-1} A_{n,k}\pi_{n,k}^{2} 
\lesssim Cb_{n,4\alpha} \sum\limits_{k=1}^{n-1} \frac{1}{k^{5} b_{k,2\alpha}^{2}}
\lesssim C n^{-4\alpha} \sum\limits_{k=1}^{n} k^{4\alpha-5}
\lesssim
C\left\{
\begin{array}{cc}
n^{-4\alpha}& 0 < \alpha < 1 \\
n^{-4}\ln n& \alpha = 1 \\
n^{-4} & 1 < \alpha
\end{array}
\right.
$$
and

$$
\begin{array}{l}
\sum\limits_{k_{1}=1}^{n-1}\sum\limits_{k_{2}=k_1+1}^{n-1}f_{2\alpha}(k_{1},k_{2})A_{n,k_{2}}\pi_{n,k_1}\pi_{n,k_2}
\lesssim C b_{n,4\alpha} \sum\limits_{k_{1}=1}^{n-1}\sum\limits_{k_{2}=k_{1}+1}^{n-1} 
\frac{1}{b_{k_{1},2\alpha} b_{k_{2},2\alpha}} \frac{1}{k_{1}^{2} k_{2}^{3}}
\end{array}
$$

$$
\begin{array}{l}
\lesssim 
 Cn^{-4\alpha} \sum\limits_{k_{1}=1}^{n-1}
k_{1}^{2\alpha-2} \sum\limits_{k_{2}=k_{1}+1}^{n-1} k_{2}^{2\alpha-3}
\lesssim
C\left\{
\begin{array}{cc}
n^{-4\alpha} \sum\limits_{k_{1}=1}^{n-1} k_{1}^{4\alpha-4}& 0 < \alpha < 1 \\
n^{-4} \sum\limits_{k_{2}=2}^{n}k_{2}^{-1}\sum\limits_{k_{1}=1}^{k_{2}} 1 & \alpha =1 \\
n^{-4\alpha} \sum\limits_{k_{2}=2}^{n} k_{2}^{4\alpha-4}& 1 < \alpha
\end{array}
\right.
\end{array}
$$

$$
\begin{array}{l}
\lesssim
C\left\{
\begin{array}{cc}
n^{-4\alpha} & 0 < \alpha < 0.75 \\
n^{-3}\ln n & \alpha=0.75  \\
n^{-3} & 0.75 < \alpha < 1 \\
n^{-4} \sum\limits_{k_{2}=2}^{n}1  & \alpha =1 \\
n^{-3} & 1 < \alpha
\end{array}
\right.
\lesssim
C\left\{
\begin{array}{cc}
n^{-4\alpha} & 0 < \alpha < 0.75 \\
n^{-3}\ln n & \alpha=0.75  \\
n^{-3} & 0.75 < \alpha < 1 \\
n^{-3}  & \alpha =1 \\
n^{-3} & 1 < \alpha.
\end{array}
\right.
\end{array}
$$
Summarizing

$$
\begin{array}{l}
\sum\limits_{k_{1}=1}^{n-1}\sum\limits_{k_{2}=k_1+1}^{n-1}f_{2\alpha}(k_{1},k_{2})A_{n,k_{2}}\pi_{n,k_1}\pi_{n,k_2}
\lesssim 
C\left\{
\begin{array}{cc}
n^{-4\alpha} & 0 < \alpha < 0.75 \\
n^{-3}\ln n & \alpha=0.75  \\
n^{-3} & 0.75 < \alpha < 1 .
\end{array}
\right.
\end{array}
$$
\end{proof}

\begin{remark}\label{remLem11}
\citet{KBarSSag2015}
wrongly stated in their Lemma $11$ that 
$\variance{\Expectation{e^{-2\alpha \tau^{(n)}} \vert \mathcal{Y}_{n}}}=O(n^{-3})$ for all $\alpha>0$. 
From the above we can see that this holds only for $\alpha>3/4$. This does not however change
\citet{KBarSSag2015}'s main results.
 If one inspects the proof of Theorem $1$ therein, then one can see
that for $\alpha>0.5$ it is required that 
$\variance{\Expectation{e^{-2\alpha \tau^{(n)}} \vert \mathcal{Y}_{n}}}=O(n^{-(2+\epsilon)})$, 
where $\epsilon>0$.
This by Lemma \ref{lem11} holds. \citet{KBarSSag2015}'s Thm. $2$  does not depend on the rate of 
convergence,
only that $n^{2}\variance{\Expectation{e^{-2\alpha \tau^{(n)}} \vert \mathcal{Y}_{n}}}\to 0$ with $n$. 
This remains true, just with a different rate.
\end{remark}
Let $\mathrm{I}^{(n)}$ be the sequence of speciation events on a random lineage
and $\left(J_{i}\right)$ be the jump pattern (binary sequence $1$ jump took place, $0$ did not take place
just after speciation event $i$)
on a randomly selected lineage.

\begin{lemma}\label{lemvarEexpPsi}
For random variables $(\Upsilon^{(n)},\mathrm{I}^{(n)},\left(J_{i}\right)_{i=1}^{\Upsilon^{(n)}})$ derived
from the same random lineage and 
a fixed jump probability $p$
we have 

\be
\begin{array}{l}
\variance{\Expectation{
\sum\limits_{i=1}^{\Upsilon^{(n)}}J_{i}    
e^{-2\alpha (T_{n}+\ldots+T_{\mathrm{I}^{(n)}_{i}+1})}\vert \mathcal{Y}_{n}}}

\lesssim  pC
\left \{
\begin{array}{lc}
n^{-4\alpha}  & 0< \alpha < 0.25  \\
n^{-1}\ln n   & \alpha = 0.25 \\
n^{-1} & 0.25 < \alpha .
\end{array}
\right.
\end{array}
\ee
\end{lemma}
\begin{proof}
We introduce the random variables

$$
\Psi^{\ast^{(n)}}:=\sum\limits_{i=1}^{\Upsilon^{(n)}} J_{i} e^{-2\alpha (T_{n}+\ldots+T_{\mathrm{I}^{(n)}_{i}+1})}
$$
and

$$
\phi^{\ast}_{i}:= Z_{i} e^{-2\alpha(T_{n}+\ldots+T_{i+1})} \Expectation{\mathbf{1}_{i} \vert \mathcal{Y}_{n}},
$$
where $Z_{i}$ is the binary random variable if a jump took place at the $i$--th speciation event 
of the tree for our considered random lineage. 
Obviously   

$$
\Expectation{\Psi^{\ast^{(n)}} \vert \mathcal{Y}_{n}} =
\sum_{i=1}^{n-1} \phi^{\ast}_{i}.
$$
Immediately (for $i<j$)

$$
\begin{array}{rcl}
\Expectation{\phi^{\ast}_{i}} & = & \frac{2p}{i+1} \frac{b_{n,2\alpha}}{b_{i,2\alpha}}, \\
\Expectation{\phi^{\ast}_{i}\phi^{\ast}_{j}} & = &  \frac{4p^{2}}{(i+1)(j+1)} \frac{b_{n,4\alpha}}{b_{j,4\alpha}}\frac{b_{j,2\alpha}}{b_{i,2\alpha}}, \\
\Expectation{{\phi^{\ast}_{i}}^{2}} & = & p \frac{b_{n,4\alpha}}{b_{i,4\alpha}} \Expectation{\left(\Expectation{\mathbf{1}_{i} \vert \mathcal{Y}_{n}}\right)^{2}}.
\end{array}
$$
We illustrate the random objects defined above in Fig. \ref{figPhiPsi}.
The term $\Expectation{\left(\Expectation{\mathbf{1}_{i} \vert \mathcal{Y}_{n}}\right)^{2}}$ can be 
expressed as $\Expectation{\mathbf{1}^{(1)}_{i}\mathbf{1}^{(2)}_{i}}$
(same as with
$\Expectation{\left(\Expectation{e^{-2\alpha \tau^{(n)}} \vert \mathcal{Y}_{n}}\right)^{2}}$ in Lemma \ref{lem11}),
where $\mathbf{1}^{(1)}_{i}$ and $\mathbf{1}^{(2)}_{i}$ are two copies of
$\mathbf{1}_{i}$ that are independent given $\mathcal{Y}_{n}$, 
i.e. for a given tree we sample two lineages and ask if the $i$--th  speciation event
is on both of them. This will occur if these lineages coalesced at a speciation event
$k \ge i$. Therefore,

$$
\begin{array}{l}
\Expectation{\mathbf{1}^{(1)}_{i}\mathbf{1}^{(2)}_{i}}
= \frac{2}{i+1} \sum\limits_{k=i+1}^{n-1} \pi_{k,n} + \pi_{i,n}
= \frac{n+1}{n-1}\frac{2}{i+1}\left(\sum\limits_{k=i+1}^{n-1} \frac{2}{(k+1)(k+2)} + \frac{1}{i+2} \right)
\\ = \frac{n+1}{n-1}\frac{2}{i+1}\left(\frac{2}{i+2} - \frac{2}{n+1} + \frac{1}{i+2} \right)
= \frac{n+1}{n-1}\frac{6}{(i+1)(i+2)}  - \frac{2}{n-1}\frac{2}{i+1}. 
\end{array}
$$
Together with the above

$$
\Expectation{{\phi^{\ast}_{i}}^{2}}  =  p \frac{b_{n,4\alpha}}{b_{i,4\alpha}} 
\left( \frac{n+1}{n-1}\frac{6}{(i+1)(i+2)}  - \frac{1}{n-1}\frac{4}{i+1} \right).
$$
Now

\begin{eqnarray}
\label{eqVarPsiStar}
\variance{\sum\limits_{i=1}^{n-1} \phi^{\ast}_{i}} & = &
\sum\limits_{i=1}^{n-1}\left(\Expectation{{\phi^{\ast}_{i}}^{2}} - \left(\Expectation{\phi^{\ast}_{i}}\right)^{2} \right)
+2\sum\limits_{i=1}^{n-1}\sum\limits_{j=i+1}^{n-1}\left(\Expectation{\phi^{\ast}_{i}\phi^{\ast}_{j}} - \Expectation{\phi^{\ast}_{i}}\Expectation{\phi^{\ast}_{j}}\right)
\\ \notag & =&
\sum\limits_{i=1}^{n-1}\left(
p \frac{b_{n,4\alpha}}{b_{i,4\alpha}}
\left(
\frac{n+1}{n-1}\frac{6}{(i+1)(i+2)}  - \frac{1}{n-1}\frac{4}{i+1}
\right)
-\frac{4p^{2}}{(i+1)^{2}} \left(\frac{b_{n,2\alpha}}{b_{i,2\alpha}}\right)^{2}
\right)
\\ \notag&& +2\sum\limits_{i=1}^{n-1}\sum\limits_{j=i+1}^{n-1}\left(
\frac{4p^{2}}{(i+1)(j+1)} \frac{b_{n,4\alpha}}{b_{j,4\alpha}}\frac{b_{j,2\alpha}}{b_{i,2\alpha}}
-
\frac{4p^{2}}{(i+1)(j+1)} \frac{b_{n,2\alpha}}{b_{i,2\alpha}} \frac{b_{n,2\alpha}}{b_{j,2\alpha}}
\right)
\\ \notag& \lesssim &
2p 
\sum\limits_{i=1}^{n-1} \frac{1}{(i+1)^{2}}\left(
 3\frac{b_{n,4\alpha}}{b_{i,4\alpha}}
-2p\left(\frac{b_{n,2\alpha}}{b_{i,2\alpha}}\right)^{2}
\right)
\text{\circledchar{\ref{rEexpPsi1}}}\refstepcounter{rEexpPsi}\label{rEexpPsi1}
\\ \notag&& 
+ 4p (n-1)^{-1} 
\sum\limits_{i=1}^{n-1} 
\frac{b_{n,4\alpha}}{b_{i,4\alpha}} \left(\frac{3}{(i+1)^{2}} - \frac{1}{i+1} \right)
\text{\circledchar{\ref{rEexpPsi2}}}\refstepcounter{rEexpPsi}\label{rEexpPsi2}
\\ \notag&&
+8p^{2}\sum\limits_{i=1}^{n-1}\sum\limits_{j=i+1}^{n-1}\left(
\frac{1}{(i+1)(j+1)} \frac{b_{j,2\alpha}}{b_{i,2\alpha}}
\left(
\frac{b_{n,4\alpha}}{b_{j,4\alpha}}
-
\left(\frac{b_{n,2\alpha}}{b_{j,2\alpha}}\right)^{2} \right)
\right).
\text{\circledchar{\ref{rEexpPsi3}}}\refstepcounter{rEexpPsi}\label{rEexpPsi3}
\end{eqnarray}
We notice that we are dealing with a telescoping sum, we take advantage of Eq. \eqref{eqTelescopSum} again
and consider the three parts in turn. 

\begin{enumerate}[label=\Roman*]
\item[\text{\circledchar{\ref{rEexpPsi1}}}] 

$$
\begin{array}{l}
\sum\limits_{i=1}^{n-1} \frac{1}{(i+1)^{2}}\left(
 3\frac{b_{n,4\alpha}}{b_{i,4\alpha}}
-2p\left(\frac{b_{n,2\alpha}}{b_{i,2\alpha}}\right)^{2}
\right)
\\= 
\sum\limits_{i=1}^{n-1} \frac{1}{(i+1)^{2}}\left(
\left(\frac{b_{n-1,2\alpha}}{b_{i,2\alpha}}\right)^{2}\left(
\frac{3n}{n+4\alpha} - \frac{2pn^{2}}{(n+2\alpha)^{2}}
\right)
+
3\sum\limits_{k=i+1}^{n-1}\left(\frac{b_{k-1,2\alpha}}{b_{i,2\alpha}}\right)^{2}
\left(\frac{k}{k+4\alpha} - \frac{k^{2}}{(k+2\alpha)^{2}} \right)
\frac{b_{n,4\alpha}}{b_{k,4\alpha}}\right)
\end{array}
$$

$$
\begin{array}{l}
= 
\sum\limits_{i=1}^{n-1} \frac{1}{(i+1)^{2}}\left(
\left(\frac{b_{n-1,2\alpha}}{b_{i,2\alpha}}\right)^{2}
\frac{n^{2}}{(n+2\alpha)^{2}}\frac{(3-2p)n+(3-2p)4\alpha+n^{-1}12\alpha^{2}}{n+4\alpha}
+
3\sum\limits_{k=i+1}^{n-1}\left(\frac{b_{k-1,2\alpha}}{b_{i,2\alpha}}\right)^{2}
\frac{k^{2}}{(k+2\alpha)^{2}}\frac{4\alpha^{2}}{k(k+4\alpha)}
\frac{b_{n,4\alpha}}{b_{k,4\alpha}}
\right)
\\ \lesssim
C((3-2p) 
n^{-4\alpha}\sum\limits_{i=1}^{n}i^{4\alpha-2}
+
12\alpha^{2}
n^{-4\alpha}\sum\limits_{i=1}^{n}i^{4\alpha-3})
\end{array}
$$

$$
\begin{array}{l}
\sim C
 \left \{
\begin{array}{cc}
n^{-4\alpha}& 0< \alpha < 0.25  \\
n^{-1}\ln n & \alpha = 0.25 \\
n^{-1} & 0.25< \alpha   .
\end{array}
\right.
\end{array}
$$
\item[\text{\circledchar{\ref{rEexpPsi2}}}]

$$
\begin{array}{l}
n^{-1}\sum\limits_{i=1}^{n-1} 
\frac{b_{n,4\alpha}}{b_{i,4\alpha}} \left(\frac{3}{(i+1)^{2}} - \frac{1}{i+1} \right) 
\sim
C(3 n^{-4\alpha-1} 
\sum\limits_{i=1}^{n}  i^{4\alpha -2}
-
n^{-4\alpha-1}\sum\limits_{i=1}^{n}  i^{4\alpha -1})
 \sim -C n^{-1}
\end{array}
$$
\item[\text{\circledchar{\ref{rEexpPsi3}}}] 

\be \label{eqExpPsi3}
\begin{array}{l}
\sum\limits_{i=1}^{n-1}\sum\limits_{j=i+1}^{n-1}\left(
\frac{1}{(i+1)(j+1)} \frac{b_{j,2\alpha}}{b_{i,2\alpha}}
\left(
\frac{b_{n,4\alpha}}{b_{j,4\alpha}}
-
\left(\frac{b_{n,2\alpha}}{b_{j,2\alpha}}\right)^{2} \right)
\right)
=
\sum\limits_{i=1}^{n-1}\sum\limits_{j=i+1}^{n-1}
\frac{1}{(i+1)(j+1)} f_{2\alpha}(i,j)A_{n,j}
\end{array}
\ee

$$
\begin{array}{l}
\lesssim
C n^{-4\alpha}
\sum\limits_{i=1}^{n}\sum\limits_{j=i+1}^{n} i^{-1+2\alpha}j^{-2+2\alpha}
\lesssim C
 \left \{
\begin{array}{cc}
n^{-4\alpha}& 0< \alpha < 0.25\\
n^{-1}\ln n& \alpha = 0.25\\
n^{-1} & 0.25 < \alpha .
\end{array}
\right.
\end{array}
$$
\end{enumerate}
Putting these together we obtain 

$$
\begin{array}{rcl}
\variance{\sum\limits_{i=1}^{n-1} \phi^{\ast}_{i}} & \lesssim &
p C
\left \{
\begin{array}{lc}
n^{-4\alpha}  & 0< \alpha < 0.25  \\
n^{-1}\ln n   & \alpha = 0.25 \\
n^{-1} & 0.25 < \alpha .
\end{array}
\right.
\end{array}
$$
On the other hand the variance is bounded from below by 
\ref{rEexpPsi3}.
Its asymptotic behaviour is tight
as the calculations there are accurate up to a constant (independent of $p$).
This is further illustrated by graphs
in Fig. \ref{figEPsiStar}.

\begin{figure}[!ht]
\begin{center}
\includegraphics[width=0.3\textwidth]{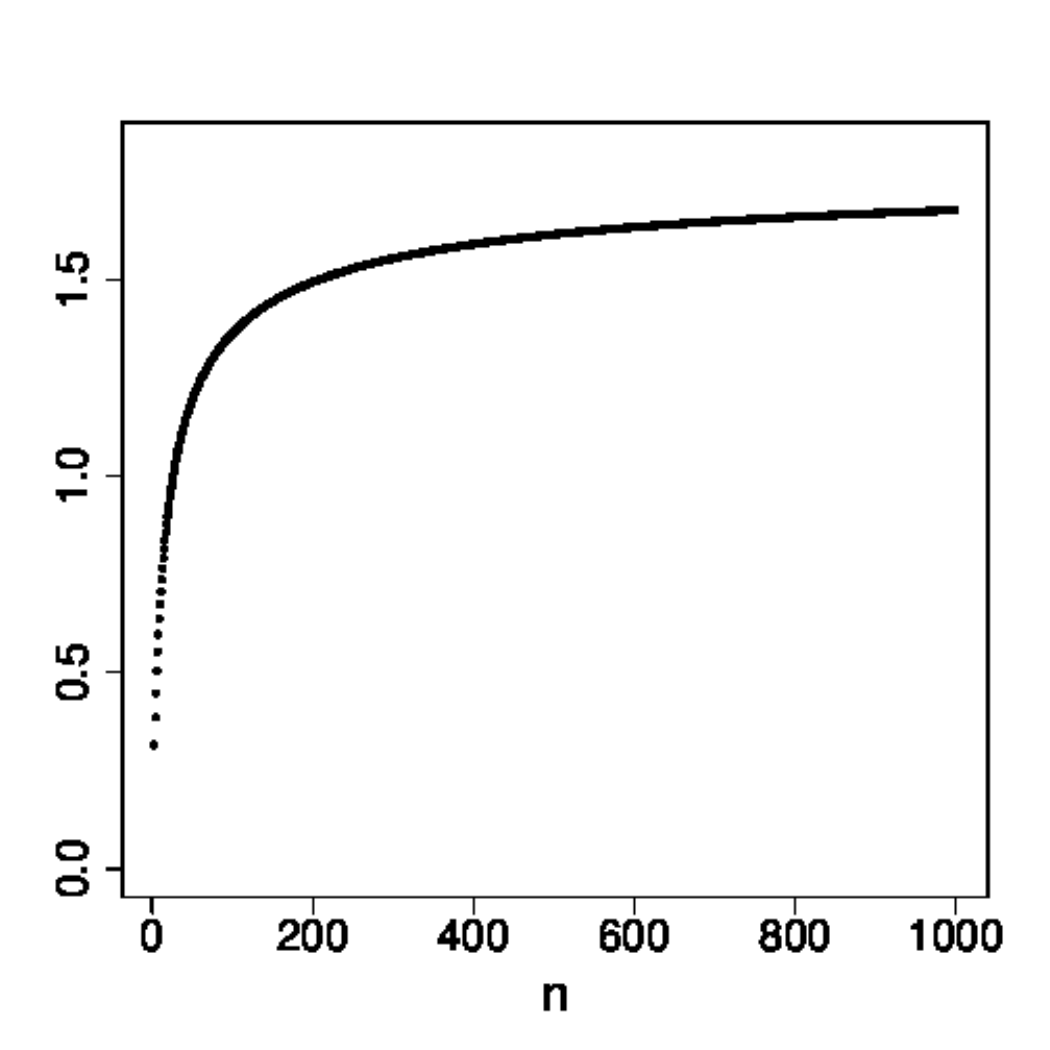}
\includegraphics[width=0.3\textwidth]{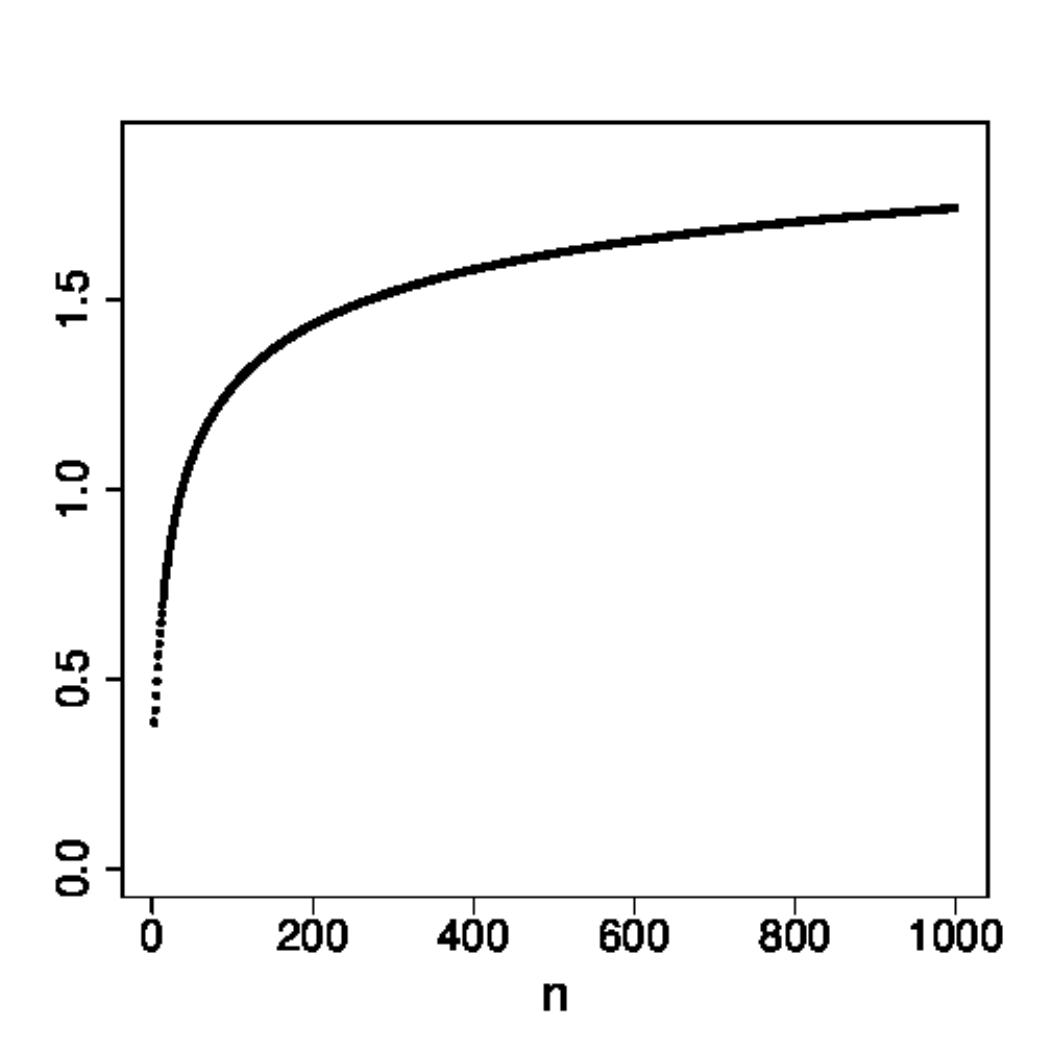}
\includegraphics[width=0.3\textwidth]{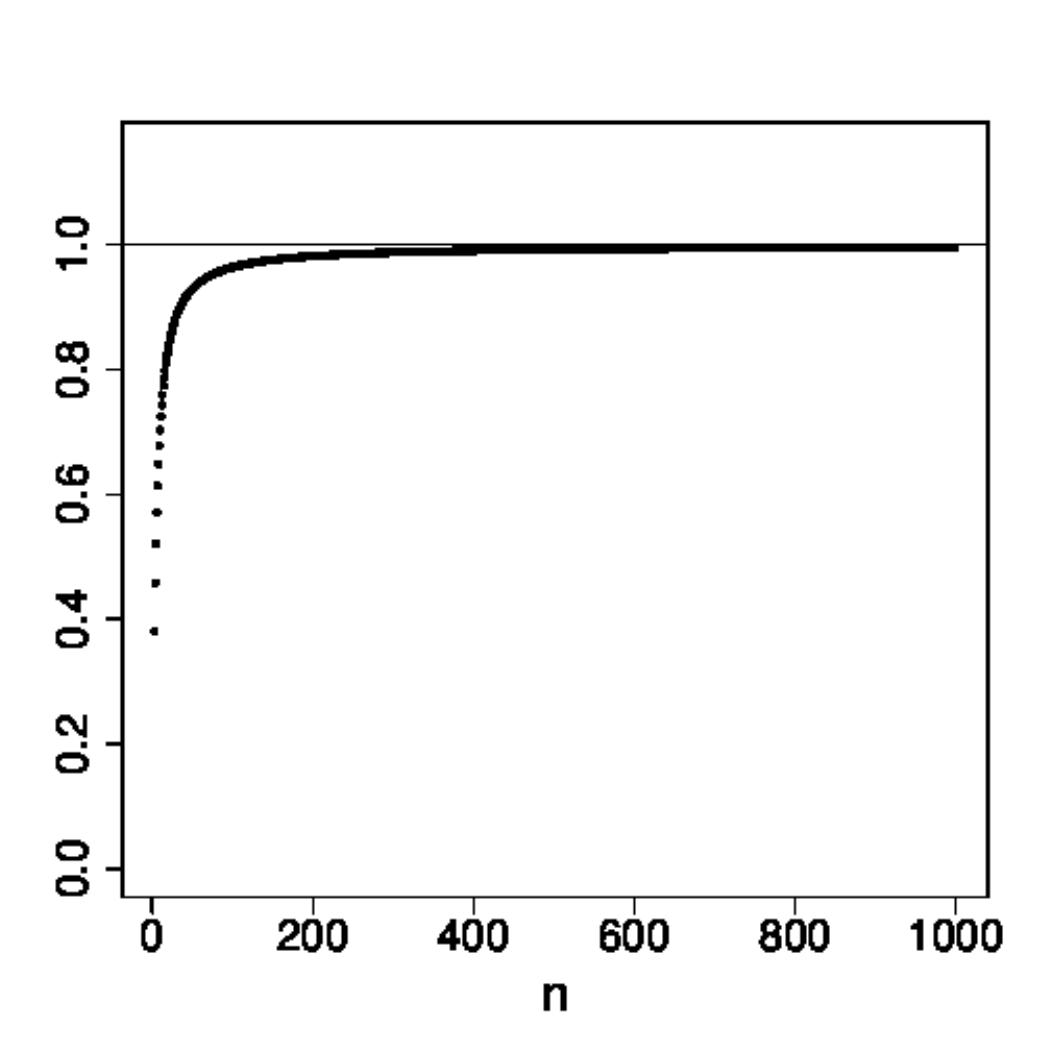}
\caption{Numerical evaluation of scaled Eq. \eqref{eqVarPsiStar} for
different values of $\alpha$. The scaling for
left: $\alpha=0.1$ equals $n^{-4\alpha}$,
centre: $\alpha=0.25$ equals $n^{-1}\log n$
and right $\alpha=1$ equals 
$(2p(3-2p)/(4\alpha-1)-4p/(4\alpha)+32p^{2}\alpha^{2}(1/(8\alpha^{2})+1/(2\alpha(2\alpha-1))-1/(4\alpha^{2})^{-1}-1/((2\alpha-1)(4\alpha-1))))n^{-1}.$
In all cases, $p=0.5$.
The value of the leading constant comes from a careful treatment 
of the summation in Lemma \ref{lemvarEexpPsi}.
The sums are approximated by definite integrals and the leading constant
resulting from the integration is remembered (in the panel on the right).
}\label{figEPsiStar}
\end{center}
\end{figure}
\end{proof}

\begin{corollary}\label{corvarEexpPsi}
Let $p_{k}$ and $\sigma_{c,k}^{2}$ be respectively the jump probability and 
variance at the $k$--th speciation event, such that the sequence $\sigma_{c,k}^{4}p_{k}$ is bounded.
We have 

$$
\begin{array}{rcccl}
n\ln^{-1}n \variance{\sum\limits_{i=1}^{n-1} \sigma_{c,i}^{2}\phi^{\ast}_{i}}& \to & 0 & \mathrm{for}~ \alpha = 0.25, \\
n \variance{\sum\limits_{i=1}^{n-1} \sigma_{c,i}^{2}\phi^{\ast}_{i}}& \to & 0 & \mathrm{for}~  0.25 < \alpha.
\end{array}
$$
iff $\sigma_{c,k}^{4}p_{k} \to 0$ with density $1$.
\end{corollary}

\begin{proof}
We consider the case, $\alpha>0.25$.
Notice that in the proof of Lemma \ref{lemvarEexpPsi}
$\variance{\sum_{i=1}^{n-1} \phi^{\ast}_{i}} \lesssim pn^{-4\alpha}\sum_{i=1}^{n-1}i^{4\alpha-2}$.
If the jump probability and variance are not constant, but as in the Corollary, then

$$\variance{\sum_{i=1}^{n-1} \sigma_{c,i}^{2}\phi^{\ast}_{i}} 
\lesssim n^{-4\alpha}\left(\sum_{i=1}^{n-1} p_{i}\sigma_{c,i}^{4}i^{4\alpha-2}+\sum_{i=1}^{n-1} p_{i}\sigma_{c,i}^{2}i^{4\alpha-2}\right).$$
Notice that if $p_{i}\sigma_{c,i}^{4} \to 0$ with density $1$, then so will
$p_{i}\sigma_{c,i}^{2}$.

The Corollary is a consequence of a more general ergodic property, similar to \citet{KPet1983}'s Lemma $6.2$ (p. $65$).
Namely take $u>0$ and if a bounded sequence
$a_{i} \to 0$ with density $1$, then

$$
n^{-u}\sum\limits_{i=1}^{n-1} a_{i} i^{u-1} \to 0.
$$
To show this say the sequence $a_{i}$ is bounded by $A$, let $E \subset \mathbb{N}$ be 
the set of natural numbers such that $a_{i} \to 0$ if $i \in E^{c}$ and
define $E_{n}=E \cup \{1,\ldots,n\}$. Then

$$
n^{-u} \sum\limits_{i=1}^{n-1}a_{i}i^{u-1} = 
n^{-u}\sum\limits_{\stackrel{i=1}{i \in E_{n-1}}}^{n-1} a_{i}i^{u-1}
+ n^{-u}\sum\limits_{\stackrel{i=1}{i \notin E_{n-1}}}^{n-1} a_{i}i^{u-1}.
$$
Denoting by $\vert E_{i} \vert$ the cardinality of a set $E_{i}$,
the former sum is bounded above by $A\frac{\vert E_{n-1}\vert}{n}$, which, by assumption, 
tends to $0$ as $n \to \infty$. For the latter sum, given $\epsilon > 0$, 
if we choose $N_{1}$ such that $\vert a_{n} \vert < \epsilon/2$ for all $n>N_{1}$ and $N_{2}$ 
such that $\left(N_{1}/n \right)^{u} < \epsilon/(2A)$ for all $n > N_{2}$, 
then for all $n>N=\max\{N_{1},N_{2} \}$, one has that

$$
n^{-u} \sum\limits_{\stackrel{i=1}{i\notin E_{n-1}}}^{n-1}a_{i}i^{u-1} = 
n^{-u}\sum\limits_{\stackrel{i=1}{i \notin E_{n-1}}}^{N_{1}} a_{i}i^{u-1}
+ n^{-u}\sum\limits_{\stackrel{i=N_{1}+1}{i \notin E_{n-1}}}^{n-1} a_{i}i^{u-1},
$$
and now one has that the former sum is bounded above by $An^{-u}N_{1}N_{1}^{u-1} < \epsilon/2$ 
and the latter by $n^{-u}n^{u-1}(n-N_{1})(\epsilon/2) < \epsilon/2$. This proves the result.

On the other hand if $a_{i}$ does not go to $0$ with density $1$, then 
$
\limsup\limits_{n} n^{-u}\sum\limits_{i=1}^{n-1} a_{i} i^{u-1} > 0.
$

When $\alpha=0.25$ we obtain the Corollary using the same ergodic argumentation for 

$$\ln^{-1}n\left( \sum_{i=1}^{n-1}p_{i}\sigma_{c,i}^{4}i^{-1}+\sum_{i=1}^{n-1}p_{i}\sigma_{c,i}^{2}i^{-1} \right).$$
\end{proof}
Let $\mathrm{\tilde{I}}^{(n)}$ be the sequence of speciation events on 
the lineage from the origin of the tree to the most recent common ancestor
of a pair of randomly selected tips
and $\left(\tilde{J}_{i}\right)$ be the jump pattern 
(binary sequence $1$ jump took place, $0$ did not take place just after speciation event $i$)
on the lineage from the origin of the tree to the most recent common ancestor
of a pair of randomly selected tips.

\begin{lemma}\label{lemvarEexpPhi}
For random variables $(\upsilon^{(n)},\mathrm{\tilde{I}}^{(n)},\left(\tilde{J}_{i}\right)_{i=1}^{\upsilon^{(n)}})$ derived
from the same random pair of lineages and a fixed jump probability $0<p<1$

\be
\variance{\Expectation{\sum\limits_{i=1}^{\upsilon^{(n)}} \tilde{J}_{i}
e^{-2\alpha (\tau^{(n)}+\ldots+T_{\mathrm{\tilde{I}}^{(n)}_{i}+1})}\vert \mathcal{Y}_{n}}}
\lesssim p(1-p) C
\left \{
\begin{array}{cc}
n^{-4\alpha} & 0 < \alpha < 0.5, \\
n^{-2}\ln n &  \alpha = 0.5,  \\
n^{-2} & 0.5 < \alpha.
\end{array}
\right.
\ee
\end{lemma}
\begin{proof}
We introduce the notation

$$
\Psi^{(n)} := \sum\limits_{i=1}^{\upsilon^{(n)}} \tilde{J}_{i}e^{-2\alpha (\tau^{(n)}+\ldots+T_{\mathrm{\tilde{I}}^{(n)}_{i}+1})}
$$
and by definition we have 

$$
\variance{\Expectation{\sum\limits_{i=1}^{\upsilon^{(n)}} \tilde{J}_{i}
e^{-2\alpha (\tau^{(n)}+\ldots+T_{\mathrm{\tilde{I}}^{(n)}_{i}})} \vert \mathcal{Y}_{n}}}
= \Expectation{\left(\Expectation{\Psi^{(n)}  \vert \mathcal{Y}_{n}}\right)^{2}}
- \left(\Expectation{\Psi^{(n)} }\right)^{2}.
$$
We introduce the random variable 

$$
\phi_{i} = \tilde{Z}_{i}\mathbf{\tilde{1}}_{i}e^{-2\alpha \left(T_{n}+\ldots+T_{i+1} \right)},
$$
where $\tilde{Z}_{i}$ is the binary random variable if a jump took place just after the $i$--th speciation event 
of the tree for our considered lineage and obviously (for $i_{1}<i_{2}$)

$$
\begin{array}{rcl}
\Expectation{\phi_{i}}&=& \frac{2p}{i+1}b_{n,2\alpha}/b_{i,2\alpha},\\
\Expectation{\phi_{i}^{2}}&=& \frac{2p}{i+1}b_{n,4\alpha}/b_{i,4\alpha},\\
\Expectation{\phi_{i_{1}}\phi_{i_{2}}}&=& \frac{4p^{2}}{(i_{1}+1)(i_{2}+1)}\frac{b_{n,4\alpha}}{b_{i_{2},4\alpha}}\frac{b_{i_{2},2\alpha}}{b_{i_{1},2\alpha}}.
\end{array}
$$
We illustrate the random objects defined above in Fig. \ref{figPhiPsi}.
We can write similarly (but not exactly the same) as for $\Psi^{\ast^{(n)}}$

$$
\Psi^{(n)} = \sum\limits_{i=1}^{k-1}\phi_{i}.
$$
As usual (just as for $\tau^{(n)}_{1}, \tau^{(n)}_{2}$ in Lemma \ref{lem11}) 
let $(\tau^{(n)}_{1},\upsilon^{(n)}_{1},\Psi^{(n)}_{1})$ and $(\tau^{(n)}_{2},\upsilon^{(n)}_{2},\Psi^{(n)}_{2})$ 
be two conditionally on $\mathcal{Y}_{n}$ independent copies of $(\tau^{(n)},\upsilon^{(n)},\Psi^{(n)})$
and now

$$
\begin{array}{l}
\Expectation{\left(\Expectation{\Psi^{(n)}\vert \mathcal{Y}_{n}}\right)^{2}} = 
\Expectation{\Expectation{\Psi^{(n)}_{1}\vert \mathcal{Y}_{n}}\Expectation{\Psi^{(n)}_{2}\vert \mathcal{Y}_{n}}}
=\Expectation{\Expectation{\Psi^{(n)}_{1}\Psi^{(n)}_{2}\vert \mathcal{Y}_{n}}}
=\Expectation{\Psi^{(n)}_{1}\Psi^{(n)}_{2}}.
\end{array}
$$
Writing out a product of two sums, for $k_{1}<k_{2}$, as

$$
\begin{array}{rcl}
\left(\sum\limits_{i_{1}=1}^{k_{1}-1} a_{i_{1}}\right)
\left(\sum\limits_{i_{2}=1}^{k_{2}-1} a_{i_{2}}\right)
& = &
\left(\sum\limits_{i=1}^{k_{1}-1} a_{i}\right)^{2}
+
\left(\sum\limits_{i_{1}=1}^{k_{1}-1} a_{i_{1}}\right)
\left(\sum\limits_{i_{2}=k_{1}}^{k_{2}-1} a_{i_{2}}\right)
\\ &&
=
\left(\sum\limits_{i=1}^{k_{1}-1} a_{i}^{2}\right)
+
2
\left(\sum\limits_{i_{1}=1}^{k_{1}-1} \sum\limits_{i_{2}=i_{1}+1}^{k_{1}-1} a_{i_{1}}a_{i_{2}} \right)
+
\left(\sum\limits_{i_{1}=1}^{k_{1}-1} a_{i_{1}}\right)
\left(\sum\limits_{i_{2}=k_{1}}^{k_{2}-1} a_{i_{2}}\right)
\end{array}
$$
and using the law of total probability to condition on the speciation event at which the two nodes
coalesced, we have

\begin{eqnarray}
\label{eqVarPsi}
\variance{\Expectation{\Psi^{(n)}\vert \mathcal{Y}_{n}}}  &=& \Expectation{\Psi^{(n)}_{1}\Psi^{(n)}_{2}} - \left(\Expectation{\Psi^{(n)}} \right)^{2}
\\ \notag & = & \sum\limits_{k=1}^{n-1} \pi_{k,n}^{2} \Bigg(
\stackrel{\text{\circledchar{\ref{rEPPp1}}}}{\refstepcounter{rEPP}\label{rEPPp1}\sum\limits_{i=1}^{k-1}{\left(\Expectation{\phi_{i}^{2}}-\Expectation{\phi_{i}}^{2}\right)}}
+ 2\stackrel{\text{\circledchar{\ref{rEPPp2}}}}{\refstepcounter{rEPP}\label{rEPPp2}
\sum\limits_{i_{1}=1}^{k-1}\sum\limits_{i_{2} = i_{1}+1}^{k-1} \left(\Expectation{\phi_{i_{1}}\phi_{i_{2}}}-\Expectation{\phi_{i_{1}}}\Expectation{\phi_{i_{2}}}\right)}
 \Bigg)
\\ \notag&& + 2 \sum\limits_{k_{1}=1}^{n-1}\sum\limits_{k_{2} = k_{1}+1}^{n-1} \pi_{k_{1},n}\pi_{k_{2},n} 
\Bigg(
\sum\limits_{i=1}^{k_{1}-1}\left(\Expectation{\phi_{i}^{2}}-\Expectation{\phi_{i}}^{2}\right){\tiny\text{\circledchar{\ref{rEPPp3}}}}\refstepcounter{rEPP}\label{rEPPp3}
\\ \notag&& 
+ 2 \stackrel[{\text{\circledchar{\ref{rEPPp4}}}}]{}{\refstepcounter{rEPP}\label{rEPPp4}
\sum\limits_{i_{1}=1}^{k_{1}-1}\sum\limits_{i_{2} = i_{1}+1}^{k_{1}-1}\left( \Expectation{\phi_{i_{1}}\phi_{i_{2}}}-\Expectation{\phi_{i_{1}}}\Expectation{\phi_{i_{2}}}\right)}
+
\stackrel[{\text{\circledchar{\ref{rEPPp5}}}}]{}{\refstepcounter{rEPP}\label{rEPPp5}
\sum\limits_{i_{1}=1}^{k_{1}-1}
\sum\limits_{i_{2}=k_{1}}^{k_{2}-1} \left(\Expectation{\phi_{i_{1}}\phi_{i_{2}}}-\Expectation{\phi_{i_{1}}}\Expectation{\phi_{i_{2}}}\right)}
 \Bigg).
\end{eqnarray}
To aid intuition, we point out that cases \ref{rEPPp1} and \ref{rEPPp2} correspond to the case when
the two pairs of tips coalesce at the same node $k$ while cases \ref{rEPPp3}--\ref{rEPPp5}
when at different nodes, $k_{1}<k_{2}$.
We first observe 

$$
\begin{array}{l}
\Expectation{\phi_{i}^{2}} - \Expectation{\phi_{i}}^{2} = \frac{2p}{i+1}\left(\frac{b_{n,4\alpha}}{b_{i,4\alpha}}-\frac{2p}{i+1}\left(\frac{b_{n,2\alpha}}{b_{i,2\alpha}} \right)^{2} \right)
 =
\frac{2p}{i+1}\left(\frac{(i+1)^{2}}{(i+1+2\alpha)^{2}}\frac{(i+1)+(4\alpha-1)+(i+1)^{-1}4\alpha(\alpha-1)}{(i+1+4\alpha)}\frac{b_{n,4\alpha}}{b_{i+4\alpha}} 
\right. \\ \left.
 + 4\alpha^{2}\frac{b_{n,4\alpha}}{b_{i,2\alpha}^{2}}\sum\limits_{j=i+2}^{n-1}\frac{b_{j,2\alpha}^{2}}{b_{j,4\alpha}}
\frac{1}{j(j+4\alpha)}
+\left(\frac{b_{n,2\alpha}}{b_{i,2\alpha}}\right)^{2} 
\frac{n(1-2p)+4\alpha(1-2p)+n^{-1}4\alpha^{2}}{n+4\alpha}
\right) 
\end{array}
$$
and

\be\label{eqCovphi1phi2}
\begin{array}{l}
\Expectation{\phi_{i_{1}}\phi_{i_{2}}} - \Expectation{\phi_{i_{1}}}\Expectation{\phi_{i_{2}}} = 
\frac{4p^{2}}{(i_{1}+1)(i_{2}+1)}\left(
\frac{b_{n,4\alpha}}{b_{i_{2},4\alpha}}\frac{b_{i_{2},2\alpha}}{b_{i_{1},2\alpha}} -
\left(\frac{b_{n,2\alpha}}{b_{i_{1},2\alpha}}\right)\left(\frac{b_{n,2\alpha}}{b_{i_{2},2\alpha}}\right)
\right)
\\ =
\frac{4p^{2}}{(i_{1}+1)(i_{2}+1)}\frac{b_{n,4\alpha}b_{i_{2},2\alpha}}{b_{i_{1},2\alpha}b_{i_{2},2\alpha}^{2}}\left(
\sum\limits_{j=i_{2}+1}^{n}\frac{b_{j,2\alpha}^{2}}{b_{j,4\alpha}}
\frac{4\alpha^{2}}{j(j+4\alpha)}
\right) .
\end{array}
\ee
Using the above, we consider each of the five components in this sum separately. 
\begin{enumerate}[label=\Roman*]
\item[\text{\circledchar{\ref{rEPPp1}}}]

$$
\begin{array}{ll}
&\sum\limits_{k=1}^{n-1} \pi_{k,n}^{2} \sum\limits_{i=1}^{k-1}\left(\Expectation{\phi_{i}^{2}} - \Expectation{\phi_{i}}^{2}\right)
\\ ~ \\ \lesssim&
pC n^{-4\alpha}
\sum\limits_{i=1}^{n} 
\left(i^{4\alpha-1}+(4\alpha-1)i^{4\alpha-2}+4\alpha(\alpha-1)i^{4\alpha-3}
 + 4\alpha^{2}i^{4\alpha-2} 
+(1-2p) i^{4\alpha-1} \right)
\sum\limits_{k=i+1}^{n}k^{-4}
\\ ~ \\ \lesssim&
p  C
\left \{
\begin{array}{cc}
 n^{-4\alpha}& 0 < \alpha < 0.75    \\
n^{-3}\ln n& \alpha = 0.75 \\
n^{-3} &  0.75< \alpha
  \end{array} \right.
\end{array}
$$
\item[\text{\circledchar{\ref{rEPPp2}}}]

$$
\begin{array}{ll}
&\sum\limits_{k=1}^{n-1} \pi_{k,n}^{2} \sum\limits_{i_{1}=1}^{k-1}\sum\limits_{i_{2}=i_{1}+1}^{k-1}\left(\Expectation{\phi_{i_{1}}\phi_{i_{2}}} - \Expectation{\phi_{i_{1}}}\Expectation{\phi_{i_{2}}}\right)
\lesssim
p^{2}Cn^{-4\alpha}\sum\limits_{k=1}^{n} k^{-4}
\sum\limits_{i_{1}=1}^{k}i_{1}^{2\alpha-1}\sum\limits_{i_{2}=i_{1}+1}^{k}i_{2}^{2\alpha-2}
\\ ~ \\ \lesssim &
C p^{2}\left \{
\begin{array}{cc}
n^{-4\alpha} \sum\limits_{i_{1}=1}^{n}i_{1}^{4\alpha-2} \sum\limits_{k=i_{1}+1}^{n} k^{-4}   & 0 < \alpha < 0.5 \\
n^{-2}\sum\limits_{k=1}^{n} k^{-4} \sum\limits_{i_{2}=2}^{k}  1  &  \alpha=0.5  \\
n^{-4\alpha} \sum\limits_{i_{1}=1}^{n}i_{1}^{4\alpha-2}  \sum\limits_{k=i_{1}+1}^{n} k^{-4}  & 0.5 < \alpha  
\end{array} \right.
\lesssim 
C p^{2}
\left \{
\begin{array}{cc}
n^{-4\alpha}  & 0 < \alpha < 1 \\
n^{-4}\ln n & \alpha = 1 \\
n^{-4}  & 1< \alpha
\end{array} \right.
\end{array}
$$

\item[\text{\circledchar{\ref{rEPPp3}}}]

$$
\begin{array}{ll}
& \sum\limits_{k_{1}=1}^{n-1} \sum\limits_{k_{2} = k_{1}+1}^{n-1} \pi_{k_{1},n}\pi_{k_{2},n} 
\sum\limits_{i=1}^{k_{1}-1} \left(\Expectation{\phi_{i}^{2}} - \Expectation{\phi_{i}}^{2}\right)
\\ ~ \\ \lesssim & 
pCn^{-4\alpha} \sum\limits_{i=1}^{n}  \left(i^{4\alpha-1} +(4\alpha-1)i^{4\alpha -2}+4\alpha(\alpha-1)i^{4\alpha -3}  + 4\alpha^{2}i^{4\alpha-2}+ (1-2p)i^{4\alpha-1}\right)\sum\limits_{k_{1}=i+1}^{n} k_{1}^{-3}
\\ ~ \\ \lesssim &
p(1-p) C
\left \{
\begin{array}{cc}
n^{-4\alpha} & 0< \alpha <0.5 \\
n^{-2}\ln n& \alpha = 0.5 \\
n^{-2}& 0.5 < \alpha
\end{array} \right.
\end{array}
$$
\item[\text{\circledchar{\ref{rEPPp4}}}]

$$
\begin{array}{ll}
&\sum\limits_{k_{1}=1}^{n-1}\sum\limits_{k_{2} = k_{1}+1}^{n-1} \pi_{k_{1},n}\pi_{k_{2},n} 
\sum\limits_{i_{1} = 1}^{k_{1}-1}\sum\limits_{i_{2} = i_{1}+1}^{k_{1}-1}
\left( \Expectation{\phi_{i_{1}}\phi_{i_{2}}}-\Expectation{\phi_{i_{1}}}\Expectation{\phi_{i_{2}}}\right)
\\ \lesssim &
p^{2}Cn^{-4\alpha} \sum\limits_{k_{1} =1 }^{n}\sum\limits_{k_{2} = k_{1}+1}^{n} k_{1}^{-2}k_{2}^{-2}
\sum\limits_{i_{1} =1}^{k_{1}}\sum\limits_{i_{2} = i_{1}+1}^{k_{1}}\left(  i_{1}^{2\alpha-1}i_{2}^{2\alpha-2} \right)
\lesssim 
p^{2}C\left \{
\begin{array}{cc}
n^{-4\alpha}  & 0< \alpha< 0.75   \\
n^{-3}\ln n & \alpha  = 0.75 \\
n^{-3} & 0.75 < \alpha  
\end{array} \right.
\end{array}
$$
\item[\text{\circledchar{\ref{rEPPp5}}}]

$$
\begin{array}{ll}
&\sum\limits_{k_{1} =1}^{n-1}\sum\limits_{k_{2} = k_{1}+1}^{n-1} \pi_{k_{1},n}\pi_{k_{2},n} 
\sum\limits_{i_{1}=1}^{k_{1}-1}\sum\limits_{i_{2}=k_{1}}^{k_{2}-1} \left(\Expectation{\phi_{i_{1}}\phi_{i_{2}}}-\Expectation{\phi_{i_{1}}}\Expectation{\phi_{i_{2}}}\right)
\\ ~ \\ \lesssim & 
p^{2}Cn^{-4\alpha} \sum\limits_{k_{1}=1}^{n}
\sum\limits_{k_{2} = k_{1}+1}^{n} k_{1}^{-2}k_{2}^{-2} \sum\limits_{i_{1}=1}^{k_{1}}\sum\limits_{i_{2}=k_{1}}^{k_{2}} \left( i_{1}^{2\alpha-1}i_{2}^{2\alpha-2} \right)
\\ ~\\ \lesssim & 
p^{2}Cn^{-4\alpha}\left\{
\begin{array}{cc}
\sum\limits_{i_{1} = 1}^{n} i_{1}^{2\alpha-1}\sum\limits_{k_{1}=i_{1}+1}^{n}   k_{1}^{-2}
\left( 
\sum\limits_{i_{2} = k_{1} }^{n} i_{2}^{2\alpha-2}  \sum\limits_{k_{2}=i_{2}+1}^{n} k_{2}^{-2}
\right)
& \alpha \notin \{0.5,1\} \\
\sum\limits_{k_{1} = 1}^{n} k_{1}^{-1}
\left( 
\sum\limits_{k_{2} = k_{1}+1}^{n} k_{2}^{-2}H_{k_{2}}  
\right)
& \alpha = 0.5\\
\frac{1}{2}
\sum\limits_{1=k_{1} < k_{2}}^{n} k_{2}^{-1}  & \alpha =1 
\end{array}
\right.
\\ ~ \\ \lesssim &
p^{2} C \left \{
\begin{array}{cc}
n^{-2} & \alpha = 0.5 \\
n^{-4\alpha}\sum\limits_{i_{1} = 1}^{n} i_{1}^{2\alpha-1}\sum\limits_{k_{1}=i_{1}+1}^{n}   k_{1}^{-2\alpha-4} &  \alpha\in (0,1)\setminus \{0.5\} \\
n^{-3} &  \alpha = 1  \\
n^{-4\alpha}\sum\limits_{i_{1} = 1}^{n} i_{1}^{2\alpha-1}\sum\limits_{k_{1}=i_{1}+1}^{n}   k_{1}^{2\alpha-4}  & 1 < \alpha  
\end{array} \right.
\lesssim  
p^{2} C\left \{
\begin{array}{cc}
n^{-4\alpha} &  0 < \alpha \le 0.75 \\
n^{-3}\ln n &  \alpha = 0.75 \\
n^{-3} &  0.75 \le \alpha   .
\end{array} \right.
\end{array}
$$
\end{enumerate}
Putting \ref{rEPPp1}--\ref{rEPPp5} together we obtain

$$
\variance{\Expectation{\Psi^{(n)}\vert \mathcal{Y}_{n}}} \lesssim
p(1-p)C \left\{
\begin{array}{cc}
n^{-4\alpha} & 0 < \alpha < 0.5 \\
n^{-2}\ln n &  \alpha = 0.5  \\
n^{-2} & 0.5 < \alpha.
\end{array}
\right.
$$
The variance is bounded from below by 
\ref{rEPPp3} and as these derivations 
are correct up to a constant (independent of $p$)
the variance behaves as above. This is further illustrated by graphs
in Fig. \ref{figEPsi}.

\begin{figure}[!ht]
\begin{center}
\includegraphics[width=0.3\textwidth]{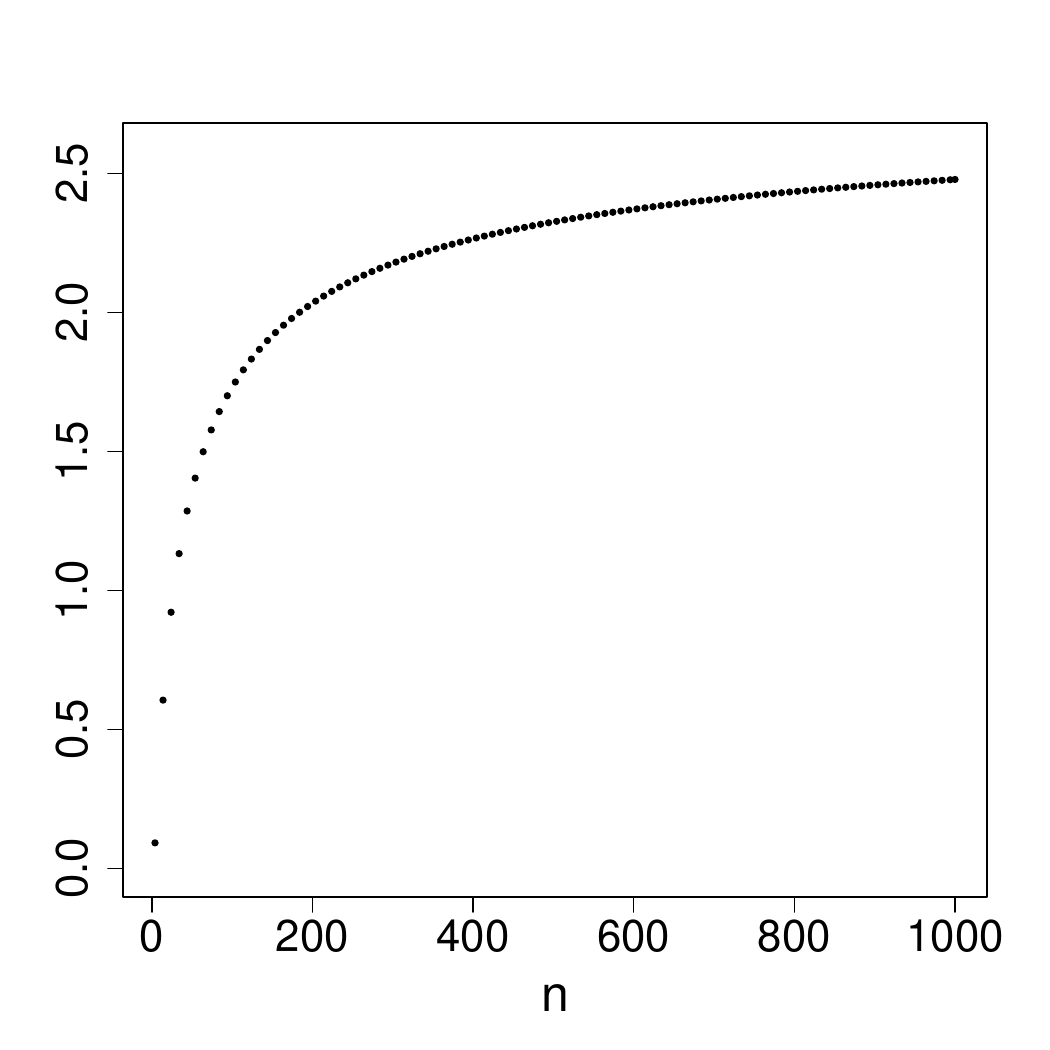}
\includegraphics[width=0.3\textwidth]{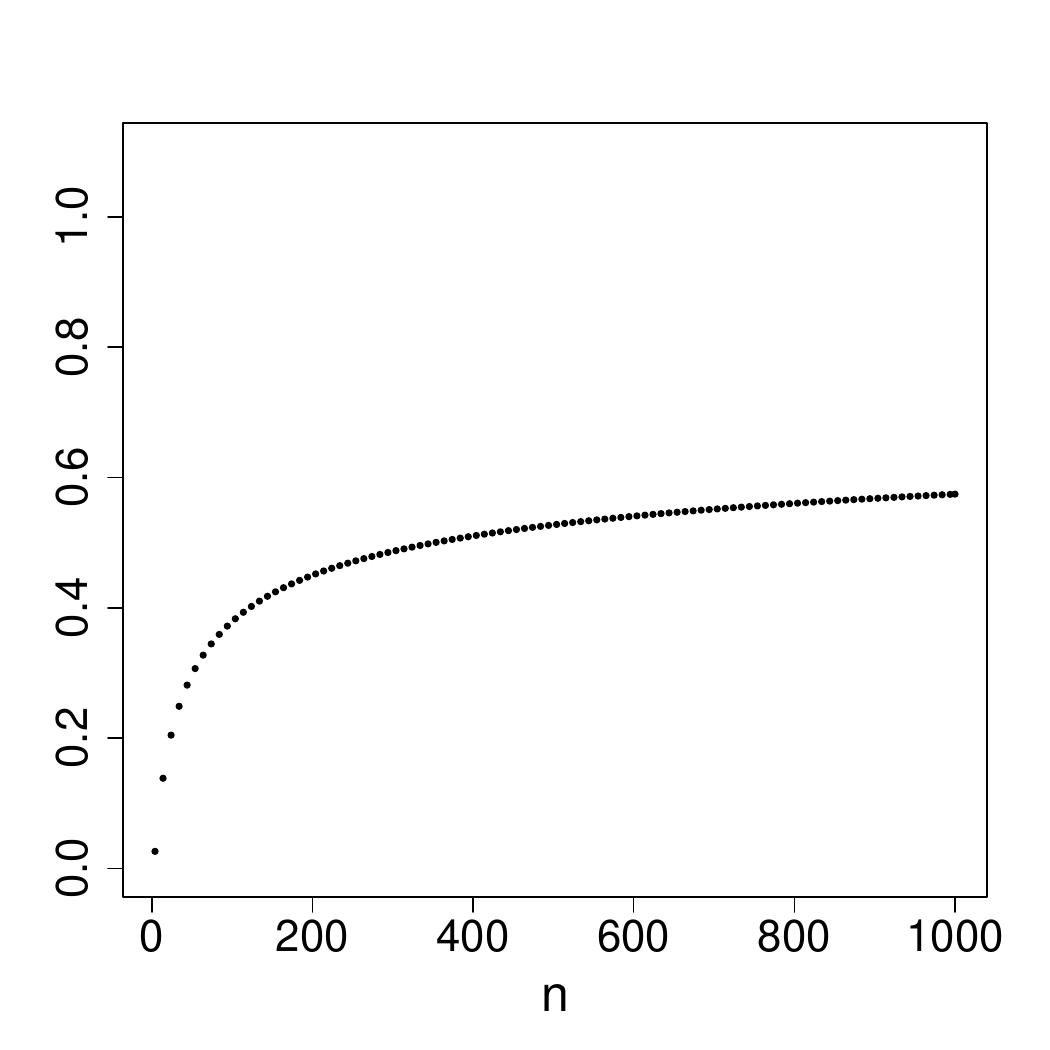}
\includegraphics[width=0.3\textwidth]{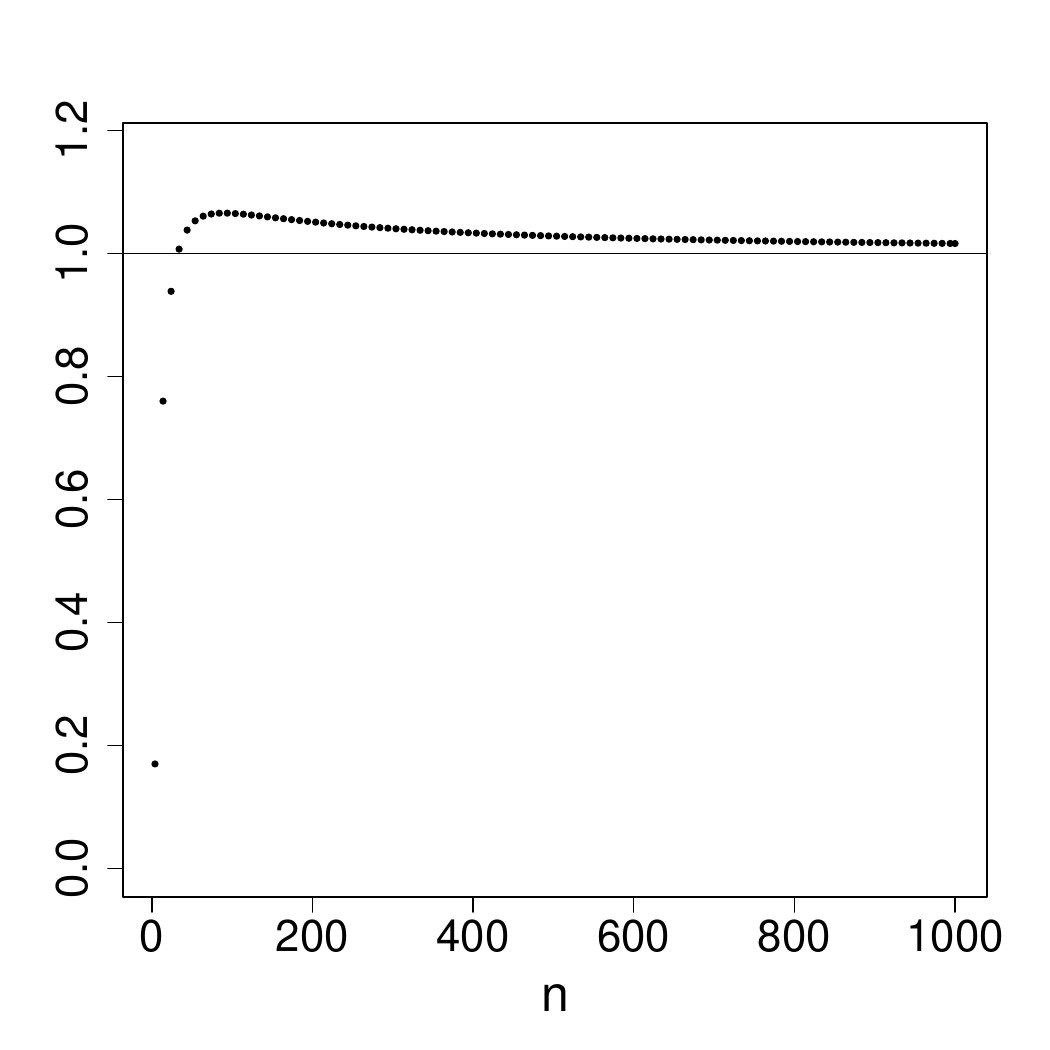}
\caption{
Numerical evaluation of scaled Eq. \eqref{eqVarPsi} for
different values of $\alpha$. The scaling for
left: $\alpha=0.35$ equals $n^{-4\alpha}$,
centre: $\alpha=0.5$ equals $16p(1-p)n^{-2}\log n$
and right $\alpha=1$ equals 
$(32p(1-p)/((4\alpha-2)(4\alpha-1)(4\alpha)))n^{-2}.$
In all cases, $p=0.5$.
The value of the leading constant comes from a careful treatment 
of the summation in Lemma \ref{lemvarEexpPhi}, component \ref{rEPPp3}.
The sums (centre and right panel) are approximated by definite integrals and the leading constant
resulting from the integration is remembered. In the $\alpha=0.5$ case the convergence is very slow.
}\label{figEPsi}
\end{center}
\end{figure}
\end{proof}

\begin{figure}[!ht]
\begin{center}
\includegraphics[width=0.5\textwidth]{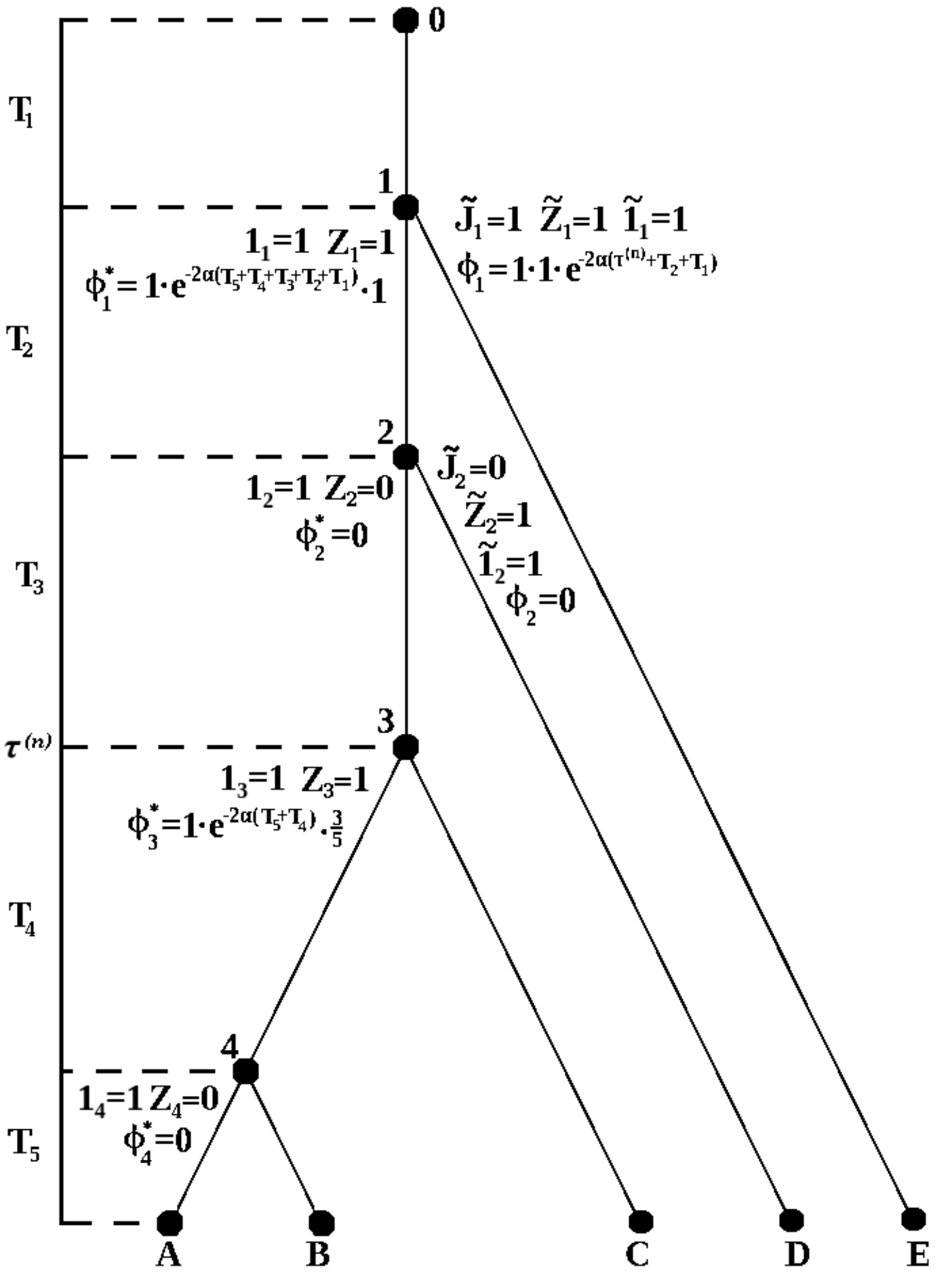}
\caption{
Illustration of the key random variables used in
Lemmata \ref{lemvarEexpPsi}, \ref{lemvarEexpPhi}
and defined in Section \ref{secNotation}.
We ``randomly sample'' (out of the five) the lineage leading to
tip A and the pair of tips (A,C) out of $\binom{5}{2}$ possible.
As jumps take place just \emph{after} speciation events there
is no associated jump at the third speciation event for
the (A,C) pair. We have 
$\Expectation{\mathbf{1}_{3} \vert \mathcal{Y}_{n}}=0.6$
as it would be one for three (A,B, or C) randomly 
sampled lineages out of the five possible. 
One should remember that for an OU process, $X(\cdot)$, for $s<t$ 
one has $\Expectation{X(t) \vert X(s)} = e^{-\alpha (t-s)}X(s)+(1-e^{-\alpha (t-s)})\theta$,
hence all contributions of the jumps to the variance and covariance are modified  by $e^{-2\alpha t}$,
where $t$ is the distance from the jump to the tip. Intuitively writing, the variable 
$\Psi^{\ast^{(n)})}$ will then be (for $\sigma_{c,k}^{2}=1$) 
the contribution of the jumps to the variance of the randomly sampled lineage,
while $\Psi^{(n)}$ will then be (for $\sigma_{c,k}^{2}\equiv 1$) 
the contribution of the jumps to the variance of the randomly sampled lineage.
}\label{figPhiPsi}
\end{center}
\end{figure}

\begin{remark}
In Lemma \ref{lemvarEexpPhi} we assumed that $0<p<1$. The case of 
$p=0$ is trivial, as then for all $i$, $\tilde{J}_{i}=0$ and hence the 
variance will be $0$. The case $p=1$ is more interesting. It means that there
will be a jump on each lineage after each speciation event. This however
implies that the variability due to the uncertainty, if a jump did or did not
take place, disappears. Hence, a faster rate of convergence will be present in 
component \ref{rEPPp3}. It will be $n^{-4\alpha}$ for $0<\alpha<0.75$, $n^{-3}\ln n$
for $\alpha=0.75$ and $n^{-3}$ for $\alpha>0.75$, i.e. same as in components
\ref{rEPPp1}, \ref{rEPPp4} and \ref{rEPPp5}.
\end{remark}

The proof of the next Corollary, \ref{corvarEexpPhi},
is exactly the same as of Corollary \ref{corvarEexpPsi}.
\begin{corollary}\label{corvarEexpPhi}
Let $p_{k}$ and $\sigma_{c,k}^{2}$ be respectively the jump probability and 
variance at the $n$--th speciation event, such that the sequence $\sigma_{c,k}^{4}p_{k}(1-p_{k})$ is bounded. 
We have 

$$
\begin{array}{rcccl}
n^{2}\ln^{-1}n \variance{\sum\limits_{i=1}^{n-1} \sigma_{c,i}^{2}\phi_{i}}& \to & 0 & \mathrm{for}~ \alpha = 0.5, \\
n^{2} \variance{\sum\limits_{i=1}^{n-1} \sigma_{c,i}^{2}\phi_{i}}& \to & 0 & \mathrm{for}~  0.5 < \alpha.
\end{array}
$$
iff $\sigma_{c,k}^{4}p_{k}(1-p_{k}) \to 0$ with density $1$.
\end{corollary}

\begin{lemma}\label{lemcovEexpUnPhi}
For random variables $U^{(n)}$, $\Psi^{(n)}$
and a fixed jump probability $p$

\be
\covariance{e^{-2\alpha U^{(n)}}}{\Expectation{\Psi^{(n)}\vert \mathcal{Y}_{n}}} 
\lesssim 
p C \left \{
\begin{array}{cc}
n^{-4\alpha} & \alpha < 0.5 \\
n^{-2}\ln n &  \alpha = 0.5  \\
n^{-(2\alpha +1)}  & 0.5 < \alpha  
\end{array} \right. .
\ee
\end{lemma}
\begin{proof}
We introduce the random variable 

$$
\bar{\phi}_{i} = \tilde{Z}_{i}\mathbf{\tilde{1}}_{i}e^{-4\alpha \left(T_{n}+\ldots+T_{i+1} \right)-2\alpha \left(T_{i}+\ldots + T_{1}\right)}
$$
and obviously 

$$
\begin{array}{rcl}
\Expectation{\bar{\phi}_{i}}&=& \frac{2p}{i+1}\left(b_{n,4\alpha}/b_{i,4\alpha}\right) b_{i,2\alpha}.
\end{array}
$$
Writing out 

\begin{eqnarray}
\label{eqCovUPsi}
\covariance{e^{-2\alpha U^{(n)}}}{\Expectation{\Psi^{(n)}\vert \mathcal{Y}_{n}}}  &=& 
\Expectation{e^{-2\alpha U^{(n)}}\Psi^{(n)}} - \left(\Expectation{e^{-2\alpha U^{(n)}}}\right)\left(\Expectation{\Psi^{(n)}} \right)
\\ \notag & = & 
\sum\limits_{k=1}^{n-1} \pi_{k,n}\left( \sum\limits_{i=1}^{k-1} \left(\Expectation{\bar{\phi}_{i}} - b_{n,2\alpha}\Expectation{\phi_{i}}\right)
\right)
\\ \notag & = & 
\sum\limits_{k=1}^{n-1} \pi_{k,n}\left( \sum\limits_{i=1}^{k-1} \frac{2p}{i+1} \left(
\frac{b_{n,4\alpha}b_{i,2\alpha}}{b_{i,4\alpha}} -  \frac{b_{n,2\alpha}^{2}}{b_{i,2\alpha}}
\right)\right)
\\ \notag & = & 
\sum\limits_{k=1}^{n-1} \pi_{k,n}\left( \sum\limits_{i=1}^{k-1} \frac{2p}{i+1}b_{i,2\alpha} \left(
\frac{b_{n,4\alpha}}{b_{i,4\alpha}} -  \left(\frac{b_{n,2\alpha}}{b_{i,2\alpha}}\right)^{2}
\right)\right)
\\ \notag & = &  \mathrm{see~Eq.~}\eqref{eqExpPsi3}
\\ \notag & = & 
2pb_{n,4\alpha}\sum\limits_{k=1}^{n-1} \pi_{k,n}\left( \sum\limits_{i=1}^{k-1} 
\frac{1}{i+1}\frac{1}{b_{i,2\alpha}}
\sum\limits_{j=i+1}^{n}\frac{b_{j,2\alpha}^{2}}{b_{j,4\alpha}}\frac{4\alpha^{2}}{j(j+4\alpha)} \right)
\\ \notag & \lesssim & Cpn^{-4\alpha} \sum\limits_{i=1}^{n} i^{2\alpha-1}
\sum\limits_{k=i+1}^{n-1} k^{-2}
\\ \notag & \lesssim &  Cpn^{-4\alpha} \sum\limits_{i=1}^{n} i^{2\alpha-2}
\lesssim
p C \left \{
\begin{array}{cc}
n^{-4\alpha} & \alpha < 0.5 \\
n^{-2}\ln n &  \alpha = 0.5 \\
n^{-2\alpha -1}  & 0.5 < \alpha  
\end{array} \right. .
\end{eqnarray}

\end{proof}

\begin{lemma}\label{lemcovEexptaunPhi}
For random variables
$\tau^{(n)},\Psi^{(n)}$
and a fixed jump probability $p$

\be
\covariance{\Expectation{e^{-2\alpha \tau ^{(n)}}\vert \mathcal{Y}_{n}}}{\Expectation{\Psi^{(n)}\vert \mathcal{Y}_{n}}} \ge 0
\ee
\end{lemma}
\begin{proof}
We introduce the random variable for $i<k$

$$
\phi_{k,i} = \tilde{Z}_{i}\mathbf{\tilde{1}}_{i}e^{-4\alpha \left(T_{n}+\ldots+T_{k+1} \right)
-2\alpha \left(T_{k}+\ldots + T_{i}\right)}
$$
and obviously 

$$
\begin{array}{rcl}
\phi_{k,i}&=& \frac{2p}{i+1}\frac{b_{n,4\alpha}}{b_{k,4\alpha}}\frac{b_{k,2\alpha}}{b_{i,2\alpha}}.
\end{array}
$$
As in the proofs of previous lemmata we denote by $\tau^{(n)}_{1}$ and  
$\Psi^{(n)}_{2}$ realizations of $\tau^{(n)}$ and $\Psi^{(n)}$ that are conditionally
independent given $\mathcal{Y}_{n}$. In other words, given a particular Yule tree
$\tau^{(n)}_{1}$ and  $\Psi^{(n)}_{2}$ will correspond to two independent
choices of pairs of tip species. 
In the below derivations $k_{1}$ will correspond to the node where the random pair,
connected to $\tau^{(n)}_{1}$, coalesced and $k_{2}$ will correspond to the node where the random pair
$\Psi^{(n)}_{2}$ coalesced. Notice that the conditional expectation of
$e^{-2\alpha \tau^{(n)}}$ given that the coalescent took place at node $k_{1}$ is
$b_{n,2\alpha}/b_{k_{1},2\alpha}$. Writing out 

\begin{eqnarray*}
\label{eqCovtauPsi}
\covariance{\Expectation{e^{-2\alpha \tau^{(n)}} \vert \mathcal{Y}_{n}}}{\Expectation{\Psi^{(n)}\vert \mathcal{Y}_{n}}}  &=& 
\Expectation{e^{-2\alpha \tau^{(n)}_{1}}\Psi^{(n)}_{2}} - \left(\Expectation{e^{-2\alpha \tau^{(n)}}}\right)\left(\Expectation{\Psi^{(n)}} \right)
\\ \notag & = & 
\sum\limits_{k=1}^{n-1} \pi_{k,n}^{2}\left( \sum\limits_{i=1}^{k-1} \left(\Expectation{\phi_{k,i}} - \frac{b_{n,2\alpha}}{b_{k,2\alpha}}\Expectation{\phi_{i}}\right)
\right) \stackrel{\text{\circledchar{\ref{rEtauPp1}}}}{\refstepcounter{rEtauP}\label{rEtauPp1}}
\\ \notag && +
\sum\limits_{k_{1}=2}^{n}\sum\limits_{k_{2}=1}^{k_{1}-1} \pi_{k_{1},n}\pi_{k_{2},n}
\left( \sum\limits_{i=1}^{k_{2}-1} \left(\Expectation{\phi_{k_{1},i}} - \frac{b_{n,2\alpha}}{b_{k_{1},2\alpha}}\Expectation{\phi_{i}}\right)
\right) \stackrel{\text{\circledchar{\ref{rEtauPp2}}}}{\refstepcounter{rEtauP}\label{rEtauPp2}}
\\ \notag &&+
\sum\limits_{k_{1}=1}^{n-1}\sum\limits_{k_{2}=k_{1}+1}^{n} \pi_{k_{1},n}\pi_{k_{2},n}
\left( \sum\limits_{i=1}^{k_{1}} \left(\Expectation{\phi_{k_{1},i}} - \frac{b_{n,2\alpha}}{b_{k_{1},2\alpha}}\Expectation{\phi_{i}}\right)
\right) \stackrel{\text{\circledchar{\ref{rEtauPp3}}}}{\refstepcounter{rEtauP}\label{rEtauPp3}}
\\ \notag &&+
\sum\limits_{k_{1}=1}^{n-1}\sum\limits_{k_{2}=k_{1}+1}^{n} \pi_{k_{1},n}\pi_{k_{2},n}
\left( \sum\limits_{i=k_{1}+1}^{k_{2}-1} \left(\Expectation{\phi_{i,k_{1}}} - \frac{b_{n,2\alpha}}{b_{k_{1},2\alpha}}\Expectation{\phi_{i}}\right)
\right) \stackrel{\text{\circledchar{\ref{rEtauPp4}}}}{\refstepcounter{rEtauP}\label{rEtauPp4}}
\end{eqnarray*}
\begin{eqnarray*}
 \notag & = & 
\sum\limits_{k=1}^{n-1} \pi_{k,n}^{2}
\left( \sum\limits_{i=1}^{k-1} \frac{2p}{i+1}\left(
\frac{b_{n,4\alpha}}{b_{k,4\alpha}}\frac{b_{k,2\alpha}}{b_{i,2\alpha}}
 - \frac{b_{n,2\alpha}}{b_{k,2\alpha}}\frac{b_{n,2\alpha}}{b_{i,2\alpha}}\right)\right)
\\ \notag && +
\sum\limits_{k_{1}=2}^{n}\sum\limits_{k_{2}=1}^{k_{1}-1} \pi_{k_{1},n}\pi_{k_{2},n}
\left( \sum\limits_{i=1}^{k_{2}-1}\frac{2p}{i+1} \left(
\frac{b_{n,4\alpha}}{b_{k_{1},4\alpha}}\frac{b_{k_{1},2\alpha}}{b_{i,2\alpha}}
 - \frac{b_{n,2\alpha}}{b_{k_{1},2\alpha}}\frac{b_{n,2\alpha}}{b_{i,2\alpha}}\right)\right)
\\ \notag &&+
\sum\limits_{k_{1}=1}^{n-1}\sum\limits_{k_{2}=k_{1}+1}^{n} \pi_{k_{1},n}\pi_{k_{2},n}
\left( \sum\limits_{i=1}^{k_{1}} \frac{2p}{i+1}\left(
\frac{b_{n,4\alpha}}{b_{k_{1},4\alpha}}\frac{b_{k_{1},2\alpha}}{b_{i,2\alpha}}
 - \frac{b_{n,2\alpha}}{b_{k_{1},2\alpha}}\frac{b_{n,2\alpha}}{b_{i,2\alpha}}\right)\right)
\\ \notag &&+
\sum\limits_{k_{1}=1}^{n-1}\sum\limits_{k_{2}=k_{1}+1}^{n} \pi_{k_{1},n}\pi_{k_{2},n}
\left( \sum\limits_{i=k_{1}+1}^{k_{2}-1} \frac{2p}{i+1}\left(
\frac{b_{n,4\alpha}}{b_{i,4\alpha}}\frac{b_{i,2\alpha}}{b_{k_{1},2\alpha}}
 - \frac{b_{n,2\alpha}}{b_{k_{1},2\alpha}}\frac{b_{n,2\alpha}}{b_{i,2\alpha}}\right)
\right).
\end{eqnarray*}
We may recognize that, after bounding $(i+1)^{-1}$ from below
by appropriately $k^{-1}$, $(k_{1}+1)^{-1}$ or $k_{2}^{-1}$, 
under the sums over $i$ we will have a difference corresponding
to a telescoping sum, i.e. Eq. \eqref{eqTelescopSum}.
This implies that the whole covariance must be positive.
Notice the similarity to the sums 
present in Eqs. \eqref{eqExpPsi3} and \eqref{eqCovphi1phi2}.

We also give intuition how all the individual sums arose. Component
\ref{rEtauPp1} corresponds to the case where both randomly sampled pairs coalesce at the 
same node. Component \ref{rEtauPp2} corresponds to the situation where the random pair
of tips associated with $\tau^{(n)}$ coalesced later (further away from the origin of the tree),
than the random pair associated with $\upsilon^{(n)}$. Components \ref{rEtauPp3} and 
\ref{rEtauPp4} correspond to the opposite situation. In particular 
component \ref{rEtauPp3} is when the ``$i$'' node on the path from the origin 
to node ``$\upsilon^{(n)}$'' is earlier than or at the same node as the coalescent associated with 
$\tau^{(n)}$ and component \ref{rEtauPp4} when later.
\end{proof}

\begin{remark}
Notice that the proof of Lemma \ref{lemcovEexptaunPhi} can easily be continued, in the same fashion
as the proofs of Lemmata \ref{lem11}--\ref{lemcovEexpUnPhi} to
find the rate of the decay to $0$ of $\covariance{\Expectation{e^{-2\alpha \tau ^{(n)}}\vert \mathcal{Y}_{n}}}{\Expectation{\Psi^{(n)}\vert \mathcal{Y}_{n}}}$.
However, in order not to further lengthen the technicalities we remain at showing the sign of the
covariance, as we require only this property.
\end{remark}

\section{Proof of the Central Limit Theorems \ref{thmCLTYOUjpsConst} and \ref{thmCLTYOUjpsae0}}\label{secProof} 
To avoid unnecessary notation it will be always assumed 
that under a given summation sign
the random variables $(\Upsilon^{(n)},\mathrm{I}^{(n)},\left(J_{i}\right)_{i=1}^{\Upsilon^{(n)}})$ are derived
from the same random lineage and 
also $(\upsilon^{(n)},\mathrm{\tilde{I}}^{(n)},\left(\tilde{J}_{i}\right)_{i=1}^{\upsilon^{(n)}})$ are derived
from the same random pair of lineages 
\newpage
\begin{lemma}\label{condYOUjMom}
Conditional on $\mathcal{Y}_{n}$
the first two moments of the scaled sample average are

$$
\begin{array}{rcl}
\Expectation{\overline{Y}_{n} \vert \mathcal{Y}_{n}} & = & \delta e^{-\alpha U^{(n)}} \\
\Expectation{\overline{Y}_{n}^{2} \vert \mathcal{Y}_{n}} & = & 
n^{-1}-(1-\delta^{2})e^{-2\alpha U^{(n)}} +
(1-n^{-1})\Expectation{e^{-2\alpha \tau^{(n)}}\vert\mathcal{Y}_{n}}
\\&&
+
n^{-1} (\sigma_{a}^{2}/(2\alpha))^{-1} \Expectation{\sum\limits_{k=1}^{\Upsilon^{(n)}}
\sigma_{c,\mathrm{I}^{(n)}_{k}}^{2} J_{k} e^{-2\alpha (T_{n}+\ldots+T_{\mathrm{I}^{(n)}_{k}+1})} 
\vert\mathcal{Y}_{n}}
\\&&+(1-n^{-1}) (\sigma_{a}^{2}/(2\alpha))^{-1}
\Expectation{
\sum\limits_{k=1}^{\upsilon^{(n)}}\sigma_{c,\mathrm{\tilde{I}}^{(n)}_{k}}^{2} \tilde{J}_{k} 
e^{-2\alpha (\tau^{(n)}+\ldots+T_{\mathrm{\tilde{I}}_{k}+1})}
\vert \mathcal{Y}_{n} },
\\
\variance{\overline{Y}_{n} \vert \mathcal{Y}_{n}} & = & 
n^{-1}-e^{-2\alpha U^{(n)}} +
(1-n^{-1})\Expectation{e^{-2\alpha \tau^{(n)}}\vert\mathcal{Y}_{n}}
\\&&
+
n^{-1} (\sigma_{a}^{2}/(2\alpha))^{-1} \Expectation{\sum\limits_{k=1}^{\Upsilon^{(n)}}
\sigma_{c,\mathrm{I}^{(n)}_{k}}^{2}J_{k}e^{-2\alpha (T_{n}+\ldots+T_{\mathrm{I}_{k}+1})}  
\vert\mathcal{Y}_{n}}
\\&&+(1-n^{-1}) (\sigma_{a}^{2}/(2\alpha))^{-1}
\Expectation{
\sum\limits_{k=1}^{\upsilon^{(n)}}\sigma_{c,\mathrm{\tilde{I}}^{(n)}_{k}}^{2}
\tilde{J}_{k}e^{-2\alpha (\tau^{(n)}+\ldots+T_{\mathrm{\tilde{I}}^{(n)}_{k}+1})} 
\vert \mathcal{Y}_{n} }.
\end{array}
$$
\end{lemma}
\begin{proof}
The first equality is immediate. The variance follows from 

$$
\begin{array}{rcl}
\variance{Y_{1}+\ldots+Y_{n} \vert \mathcal{Y}_{n}} &=&
n(1-e^{-2\alpha U^{(n)}}) + 
(\sigma_{a}^{2}/(2\alpha))^{-1} 
\sum\limits_{i=1}^{n} \sum\limits_{k=1}^{\Upsilon^{(i,n)}}
\sigma_{c,\mathrm{I}_{k}^{(i,n)}}^{2}J_{k}^{(i,n)}e^{-2\alpha (T_{n}+\ldots+T_{\mathrm{I}_{k}^{(i,n)}})}
\\&&
+2\sum\limits_{i=1}^{n}\sum\limits_{j=i+1}^{n}\left(
(e^{-2\alpha \tau^{(i,j,n)}} - e^{-2\alpha U^{(n)}})+
\right.\\&&\left.
(\sigma_{a}^{2}/(2\alpha))^{-1} 
\sum\limits_{k=1}^{\upsilon^{(i,j,n)}} \sigma_{c,\mathrm{I}_{k}^{(i,j,n)}}^{2} J_{k}^{(i,j,n)} e^{-2\alpha (\tau^{(i,j,n)}+\ldots+T_{\mathrm{I}_{k}^{(i,j,n)}})}
\right)
\\&=&
n-n^{2}e^{-2\alpha U^{(n)}} +
n(n-1)\Expectation{e^{-2\alpha \tau^{(n)}}\vert\mathcal{Y}_{n}}
\\~\\&&
+
n (\sigma_{a}^{2}/(2\alpha))^{-1} 
\Expectation{\sum\limits_{k=1}^{\Upsilon^{(n)}}
\sigma_{c,\mathrm{I}^{(n)}_{k}}^{2}J_{k}e^{-2\alpha (T_{n}+\ldots+T_{\mathrm{I}^{(n)}_{k}+1})}
\vert\mathcal{Y}_{n}}
\\~\\&&+n(n-1)(\sigma_{a}^{2}/(2\alpha))^{-1}  
\Expectation{
\sum\limits_{k=1}^{\upsilon^{(n)}}\sigma_{c,\mathrm{\tilde{I}}^{(n)}_{k}}^{2}\tilde{J}_{k}e^{-2\alpha (\tau^{(n)}+\ldots+T_{\mathrm{I}^{(n)}_{k}})}
\vert \mathcal{Y}_{n} }.
\end{array}
$$
This immediately entails the second moment.

\end{proof}

Before stating the next lemma we remind the reader of a key, for this manuscript, result presented in
\citet{KBar2014}'s Appendix A.$2$ (top of second column, p. $55$) in the case of $p$ constant

\be\label{eqEPsin}
\Expectation{\Psi^{(n)}} = p\left\{
\begin{array}{cc}
\frac{1}{\alpha}\left(\frac{2-(2\alpha+1)(2\alpha n -2\alpha+2)b_{n,2\alpha}}{(n-1)(2\alpha-1)}\right) & \alpha \neq 0.5\\
\frac{4}{n-1}\left(H_{n}-\frac{5n-1}{2(n+1)} \right) & \alpha = 0.5
\end{array}
\right. .
\ee
\begin{lemma}\label{lemSubmartEexpPhi}
Assume that the jump probability is constant, equalling $0<p<1$, at every speciation event.
Let 

$$
a_{n}(\alpha) = \left\{
\begin{array}{cc}
n^{2\alpha} & 0 < \alpha < 0.5, \\
n\ln ^{-1}n & 0.5=\alpha, \\
n & 0.5 < \alpha
\end{array}
\right.
$$
and then for all $\alpha>0$ and $n$ greater than some $n(\alpha)$

$$
W_{n}:= a_{n}(\alpha)\Expectation{\Psi^{(n)}\vert \mathcal{Y}_{n} },
$$
converges a.s. and in $L^{1}$ to a random variable
$W_{\infty}$ with expectation 

$$
\Expectation{W_{\infty}}=
\left\{
\begin{array}{cc}
\frac{2p(2\alpha+1)\Gamma(2\alpha+1)}{(1-2\alpha)} & 0 < \alpha < 0.5, \\
4p & 0.5=\alpha, \\
2p/(\alpha(2\alpha-1))  & 0.5 < \alpha .
\end{array}
\right.
$$
In particular for $\alpha=0.5$ (and also $p=1$, see Remark \ref{remSubmartEexpPhip1})
$W_{\infty}$ is a constant and the convergence is a.s. and $L^{2}$.
\end{lemma}
\begin{proof}
\noindent {\sc for $\alpha>0.5$}
We know that $\Expectation{W_{n}} < C_{E}$
for some constant $C_{E}$, as  $\Expectation{W_{n}} \to 2p/(\alpha(2\alpha-1))$ by Eq. \eqref{eqEPsin}.
Furthermore, by Lemma \ref{lemvarEexpPhi} $\variance{W_{n}} < C_{V}$, for some constant $C_{V}$.
Looking in detail, one can see from Eq. \eqref{eqEPsin},
that $\Expectation{W_{n}}$ will (from $n$ large enough) converge monotonically to its limit.
It will be decreasing with $n$ for $\alpha>1$ and increasing for $0.5<\alpha \le 1$.
If one considers the asymptotic behaviour, then
the leading term will be $4p/(\alpha(2\alpha-1))(1+1/(n-1))(1-\alpha\Gamma(2\alpha+2)n^{-2\alpha+1})$.
Direct calculations show that for $\alpha>1$ it will be decreasing, as it behaves as 
$4p/(\alpha(2\alpha-1))(1+1/(n-1))$, for $\alpha=1$ it will be increasing as it behaves as 
$4p(1-5n^{-1})$, while for for $0.5<\alpha<1$ it will be increasing as it behaves as 
$4p/(\alpha(2\alpha-1))(1-\alpha\Gamma(2\alpha+2)n^{-2\alpha+1})$.

Therefore, if one studies the proof of the downcrossing inequality
and submartingale convergence theorem (e.g. Thm. $1.71$, Cor. $1.72$, p. $44$, \citet{PMed2007}) one
will notice that only the monotonicity (which in the classical submartingale convergence theorem is 
a consequence of the sequence being a submartingale) and boundedness 
of the expectations of the sequence of positive random variables are required for the 
almost sure convergence. All of the above is met in our case for $W_{n}$. 

Hence, by the above $W_{n} \to W_{\infty}$ a.s.
for some random variable $W_{\infty}$ and as all expectations are finite,
and the variance is uniformly bounded we have $\Expectation{W_{\infty}} < \infty$.  
This entails
$\Expectation{W_{n}} \to \Expectation{W_{\infty}} = 2p/(\alpha(2\alpha-1))$.
Also we have  
uniform integrability of $\{W_{n} \}$ and hence $L^{1}$ convergence. 
\\
\noindent {\sc Proof for $\alpha=0.5$} 
By Lemma \ref{lemvarEexpPhi} we know that $\variance{\Expectation{W_{n} \vert \mathcal{Y}_{n}}}$
behaves as $n^{-2}\ln n$. Therefore, 
$\variance{W_{n}}=a_{n}^{2}(0.5)\cdot\variance{\Expectation{W_{n} \vert \mathcal{Y}_{n}}}
\sim C (n^{2}\ln^{-2} n) (n^{-2}\ln n) = C \ln^{-1} n \to 0$.
Therefore, $W_{n}$ converges a.s. and in $L^{2}$ to a constant $W_{\infty}=4p$.

\noindent {\sc Proof for $0<\alpha<0.5$}
is the same as the proof for $\alpha>0.5$, except that
now the leading terms in the asymptotic behaviour of
$\Expectation{W_{n}}$ will be 
$p(2\alpha\Gamma(2\alpha+2)+2n^{2\alpha-1})/(\alpha(1-2\alpha))$.
This causes the sequence of expectations to be increasing (from $n$ large enough) and we may
argue similar as when  $\alpha > 0.5$.
From Eq. \eqref{eqEPsin} 
we obtain $\Expectation{W_{n}}\to 2p(2\alpha+1)\Gamma(2\alpha+1)/(1-2\alpha)$ and
$\variance{W_{n}}$ is bounded by a constant by Lemma \ref{lemvarEexpPhi}.
\end{proof}

\begin{remark}\label{remSubmartEexpPhip1}
If $p=0$, we are in the trivial case of no jumps. When $p=1$, in $\alpha>0.5$ regime
we will have $W_{n}$ converging a.s. and in $L^{2}$ to a constant, denoted above as 
$\Expectation{W_{\infty}}$, by the same argument
that takes place for $\alpha=0.5$, i.e. as the rate of decay to $0$ of 
$\variance{\Expectation{W_{n} \vert \mathcal{Y}_{n}}}$ is faster than $n^{-2}$. 
In the $\alpha <0.5$ regime the argumentation presented above
holds for $p=1$ and no convergence to a constant can be deduced, as Lemma \ref{lemvarEexpPhi}
does not provide a different rate of decay of $\variance{\Expectation{W_{n} \vert \mathcal{Y}_{n}}}$.
\end{remark}

\begin{remark}
It is worth noticing that $W_{n}$ has a very interesting recursive structure. 
Denote by $\Psi^{(n+1)}_{ij}$ the value that $\Psi^{(n+1)}$ would 
take if the randomly chosen pair of species would be tips $i$ and $j$
and by $\Psi^{\ast^{(n)}}_{i}$ the value that $\Psi^{\ast^{(n)}}$ would 
take if tip $i$ is sampled.

$$
\begin{array}{rcl}
W_{n+1} & = & (n+1)\frac{2}{(n+1)n}\sum\limits_{i=1}^{n}\sum\limits_{j=i+1}^{n+1}\Psi^{(n+1)}_{ij}
\\ &=&  e^{-2\alpha T_{n+1}}
\left( \frac{n-1}{n}W_{n} + \frac{2}{n}\sum\limits_{i=1}^{n}\xi_{i}
\sum\limits_{k=1}^{\Upsilon^{(i,n)}}J_{k}^{(i,n)}e^{-2\alpha 
(T_{n}+\ldots + T_{\mathrm{I}^{(i,n)}_{k}+1})}
\right. \\&& \left. 
+ \frac{2}{n}\sum\limits_{i=1}^{n}\xi_{i}\sum\limits_{j\neq i}^{n}
\sum\limits_{k=1}^{\upsilon^{(i,j,n)}}J_{k}^{(i,j,n)}
e^{-2\alpha (\tau^{(i,j,n)}+\ldots+T_{\mathrm{I}^{(i,j,n)}_{k}+1})}
\right)
\\ &=& e^{-2\alpha T_{n+1}}
\left( \frac{n-1}{n}W_{n} + \frac{2}{n}\sum\limits_{i=1}^{n}\xi_{i}\Psi^{\ast^{(n)}}_{i}
+\frac{2}{n}\sum\limits_{i=1}^{n}\xi_{i}\sum\limits_{j\neq i}^{n}\Psi^{(n)}_{ij} \right),
\end{array}
$$
where $\xi_{i}$ is a binary random variable indicating whether it  is the $i$--th lineage that split
(see Fig. \ref{figSplit}). It is worth emphasizing that the sum defining $W_{n+1}$ splits
according to whether one picks both members of the pair of species splitting in the last
speciation event or only one of them.

\begin{figure}[!ht]
\begin{center}
\includegraphics[width=1\textwidth]{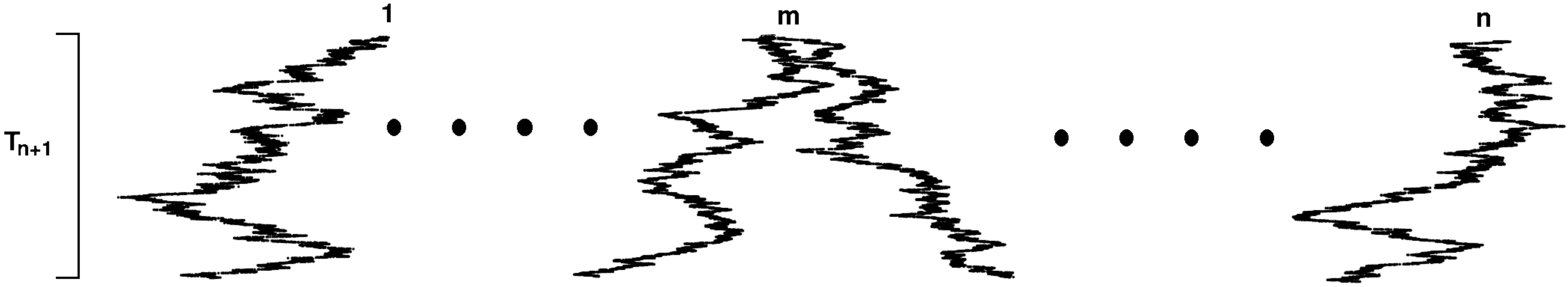}
\caption{The situation of the process between the $n$--th and $n+1$--st split. Node $m$ split
so $\xi_{m}=1$ and $\xi_{i}=0$ for $i\neq m$. The time between the splits is $T_{n+1}\sim \exp(n+1)$.
}\label{figSplit}
\end{center}
\end{figure}

Obviously the distribution of the vector $(\xi_{1},\ldots,\xi_{n})$ is uniform
on the $n$--element set \\ $\{(1,0,\ldots,0),\ldots,(0,\ldots,0,1)\}$. In particular note

$$
\begin{array}{rcl}
\Expectation{W_{n+1} \vert \mathcal{Y}_{n}} 
&=&  
\frac{n+1}{n+1+2\alpha}
\left( \frac{n-1}{n}W_{n} + 
\frac{2}{n^{2}}\Expectation{\sum\limits_{i=1}^{n}\Psi^{\ast^{(n)}}_{i}\vert \mathcal{Y}_{n}}
+\frac{2}{n^{2}}\Expectation{\sum\limits_{i=1}^{n}\sum\limits_{j\neq i}^{n}\Psi^{(n)}_{ij}\vert \mathcal{Y}_{n}}
\right)
\\ & = &
\frac{n+1}{n+1+2\alpha}
\left( \frac{n-1}{n}W_{n} + 
\frac{2}{n}\Expectation{\Psi^{\ast^{(n)}}\vert \mathcal{Y}_{n}}
+\frac{2n(n-1)}{n^{2}}\Expectation{\Psi^{(n)}\vert \mathcal{Y}_{n}}
\right)
\\ & = &
\frac{n+1}{n+1+2\alpha}
\left( \frac{n-1}{n}W_{n} + 
\frac{2}{n}\Expectation{\Psi^{\ast^{(n)}}\vert \mathcal{Y}_{n}}
+\frac{2(n-1)}{n^{2}}W_{n}
\right)
\\ & = &
\frac{n+1}{n+1+2\alpha}
\left( \frac{(n-1)(n+2)}{n^{2}}W_{n} + 
\frac{2}{n}\Expectation{\Psi^{\ast^{(n)}}\vert \mathcal{Y}_{n}}
\right)
\\ & = &
\frac{(n-1)(n+1)(n+2)}{n^{2}(n+1+2\alpha)}W_{n} + 
\frac{2(n+1)}{n(n+1+2\alpha)}\Expectation{\Psi^{\ast^{(n)}}\vert \mathcal{Y}_{n}}.
\end{array}
$$
Furthermore, $W_{n}$ shows resemblance to a martingale as the coefficient 
$\frac{(n-1)(n+1)(n+2)}{n^{2}(n+1+2\alpha)}$ converges to $1$ monotonically, depending on $\alpha$ 
from above or below, while 
$\frac{2(n+1)}{n(n+1+2\alpha)}\Expectation{\Psi^{\ast^{(n)}}\vert \mathcal{Y}_{n}} \xrightarrow{\mathbb{P},L^{2}} 0$.
\end{remark}

\noindent 
{\sc Proof of Theorem \ref{thmCLTYOUjpsConst}, Part \ref{rCLTthmpIpConst}, $\alpha > 0.5$} \\
We will show convergence in probability of the conditional mean and variance

$$
\begin{array}{rcccl}
\mu_{n}&:=&
\sqrt{n}\Expectation{\overline{Y}_{n} \vert \mathcal{Y}_{n}}
&\xrightarrow{\mathbb{P}} &
0 ~~~n\to \infty \\
\sigma_{n}^{2}&:= &
n\variance{\overline{Y}_{n} \vert \mathcal{Y}_{n}}
& \xrightarrow{\mathbb{P}} & \sigma_{\infty}^{2}
 ~~~n\to \infty ,
\end{array}
$$
for a finite mean and variance random variable $\sigma_{\infty}^{2}$.
Then, due to the conditional normality of $\overline{Y}_{n}$ this will give the 
convergence of 
characteristic functions 
and the desired weak convergence, i.e.

$$
\Expectation{e^{i x \sqrt{n} \cdot\overline{Y}_{n}}}=\Expectation{e^{i\mu_{n}x-\sigma_{n}^{2}x^{2}/2}}\to 
\Expectation{e^{-\sigma_{\infty}^{2}x^{2}/2}}.
$$

Using Lemma \ref{condYOUjMom} and that the Laplace transform 
of the average coalescent time \cite[Lemma $3$ in][]{KBarSSag2015} is

\be
\Expectation{e^{-2\alpha \tau^{(n)}_{ij}}} = \frac{2-(n+1)(2\alpha+1)b_{n,2\alpha}}{(n-1)(2\alpha-1)} = \frac{2}{2\alpha-1} n^{-1} + O(n^{-2\alpha})
\ee
we can calculate 

$$
\begin{array}{rcl}
\Expectation{\mu_{n}} & = & \delta \Expectation{e^{-\alpha U^{(n)}}} =  \delta b_{n,\alpha} = O(n^{-\alpha}), \\
\variance{\mu_{n}} & = & n\left(\Expectation{\mu_{n}^{2}} - \left(\Expectation{\mu_{n}}\right)^{2}\right)
=  \delta^{2}n\left(\Expectation{e^{-2\alpha U^{(n)}}} - \left(\Expectation{e^{-\alpha U^{(n)}}}\right)^{2}\right)
=  \delta^{2}n\left(b_{n,2\alpha} - b_{n,\alpha}^{2}\right)
\\
& = &  \delta^{2}\alpha n b_{n,2\alpha} \sum\limits_{j=1}^{n} \frac{b_{j,\alpha}^{2}}{b_{j,2\alpha}}\frac{1}{j(j+2\alpha)}
 = O(n^{-2\alpha+1}).  
\end{array} $$
Therefore we have $\mu_{n} \to 0$ in $L^{2}$ and hence in $\mathbb{P}$.

Remembering that 
$\sigma_{c,k}^{2}$ was assumed constant, equalling $\sigma_{c}^{2}$,
Lemma \ref{condYOUjMom} states that

$$
\begin{array}{rcl}
\sigma_{n}^{2} & = & 
1-ne^{-2\alpha U^{(n)}} +
n(1-n^{-1})\Expectation{e^{-2\alpha \tau^{(n)}}\vert\mathcal{Y}_{n}}
\\&&
+(\sigma_{a}^{2}/(2\alpha))^{-1}\sigma_{c}^{2}
\Expectation{\Psi^{\ast^{(n)}}\vert \mathcal{Y}_{n} }
+n(1-n^{-1}) (\sigma_{a}^{2}/(2\alpha))^{-1}\sigma_{c}^{2}
\Expectation{\Psi^{(n)}\vert \mathcal{Y}_{n} }
\end{array}
$$
Remembering that $p_{k}$ was assumed constant, equalling $p$,
we know that 
\begin{enumerate}
\item $n\Expectation{e^{-2\alpha \tau^{(n)}}} \to 2/(2\alpha-1)$
(Eq. $(4)$ in Lemma $3$, \citet{KBarSSag2015}),
\item $n^{2}\variance{\Expectation{e^{-2\alpha \tau^{(n)}}\vert\mathcal{Y}_{n}}} \to 0$
(Lemma \ref{lem11}),
\item $\Expectation{\Psi^{\ast^{(n)}}} \to 2p/(2\alpha)$
(Appendix A.$2$, p. $54$ just above Fig. A.$8$., \citet{KBar2014}),
\item $\variance{\Expectation{\Psi^{\ast^{(n)}}\vert \mathcal{Y}_{n} }}\to 0$ (Lemma \ref{lemvarEexpPsi}),
\item $n\Expectation{\Psi^{(n)}\vert \mathcal{Y}_{n} } \xrightarrow{\mathbb{P}} W_{\infty}$
(Lemmata \ref{lemvarEexpPhi}, \ref{lemSubmartEexpPhi}).
\end{enumerate}
Hence, we have $n\Expectation{e^{-2\alpha \tau^{(n)}}\vert\mathcal{Y}_{n}} \xrightarrow{\mathbb{P},L^{2}} 2/(2\alpha-1)$
and $\Expectation{\Psi^{\ast^{(n)}}\vert \mathcal{Y}_{n} }\xrightarrow{\mathbb{P},L^{2}} 2p/(2\alpha)$.
Putting these individual components together we obtain

$$
\sigma_{n}^{2} \xrightarrow{\mathbb{P}} 1 + \frac{2}{2\alpha-1} + \frac{2p\sigma^{2}_{c}}{\sigma_{a}^{2}}+ 
\frac{\sigma^{2}_{c}W_{\infty}}{\sigma_{a}^{2}/(2\alpha)} =: \sigma_{\infty}^{2}. 
$$
By Lemma \ref{lemSubmartEexpPhi} 

$$
\Expectation{\sigma_{\infty}^{2}} = 
1+ \frac{2}{2\alpha-1}+
\frac{2p\sigma^{2}_{c}}{\sigma_{a}^{2}}+ 
\frac{4p\sigma^{2}_{c}}{(2\alpha-1)\sigma_{a}^{2}}.
$$
\\~\\
{\sc Proof of Part \ref{rCLTthmpIIpConst}, $\alpha=0.5$}\\
We again show convergence in probability of the conditional mean and variance

$$
\begin{array}{rcccl}
\mu_{n}&:=&
\sqrt{(n\ln^{-1}n)}\Expectation{\overline{Y}_{n} \vert \mathcal{Y}_{n}}
&\xrightarrow{\mathbb{P}} &
0 ~~~n\to \infty \\
\sigma_{n}^{2}&:= &
(n\ln^{-1}n)\variance{\overline{Y}_{n} \vert \mathcal{Y}_{n}}
& \xrightarrow{\mathbb{P}} & 2+4p\sigma_{c}^{2}/\sigma_{a}^{2}
 ~~~n\to \infty.
\end{array}
$$
As all the steps are the same as in {\sc Part \ref{rCLTthmpIpConst}} we just explicitly write the key part
concerning 

$$
\begin{array}{rcl}
\sigma_{n}^{2} & = & 
\left(\ln^{-1}n\right)\left(1-ne^{- U^{(n)}}\right) +
(n\ln^{-1} n)(1-n^{-1})\Expectation{e^{- \tau^{(n)}}\vert\mathcal{Y}_{n}}
\\&&
+(\sigma_{c}^{2}/\sigma_{a}^{2})(\ln^{-1}n)
\Expectation{\Psi^{\ast^{(n)}}\vert \mathcal{Y}_{n} }
+(n\ln^{-1} n)(1-n^{-1}) (\sigma_{c}^{2}/\sigma_{a}^{2})
\Expectation{\Psi^{(n)}\vert \mathcal{Y}_{n} }
\end{array}
$$
As before $p_{k}\equiv p$ and
we know that 
\begin{enumerate}
\item $(n\ln^{-1}n)\Expectation{e^{- \tau^{(n)}}} \to 2$
(Eq. $(4)$ in Lemma $3$, \citet{KBarSSag2015}),
\item $(n^{2}\ln^{-2}n)\variance{e^{- \tau^{(n)}}\vert\mathcal{Y}_{n}} \to 0$
(Lemma \ref{lem11}),
\item $(\ln^{-1}n)\Expectation{\Psi^{\ast^{(n)}}} \to 0$
(Appendix A.$2$, p. $54$ just above Fig. A.$8$., \citet{KBar2014}),
\item $(\ln^{-2}n)\variance{\Expectation{\Psi^{\ast^{(n)}}\vert \mathcal{Y}_{n} }}\to 0$ (Lemma \ref{lemvarEexpPsi}),
\item $(n\ln^{-1} n)\Expectation{\Psi^{(n)}} \to 4p$ (Eq. \ref{eqEPsin}),
$(n^{2}\ln^{-2} n)\variance{\Expectation{\Psi^{(n)}\vert \mathcal{Y}_{n} }} \to 0$
(Lemma \ref{lemvarEexpPhi}).
\end{enumerate}
Putting these individual components together we obtain the $L^{2}$ convergence and hence

$$
\sigma_{n}^{2} \xrightarrow{\mathbb{P}} 2 + 4p\sigma_{c}^{2}/\sigma_{a}^{2}. 
$$
\\~\\
{\sc Proof of Part \ref{rCLTthmpIIIpConst}, $0<\alpha<0.5$}\\
We notice that the martingale (with respect to $\mathcal{F}_{n}$)
$H_{n}=(n+1)e^{(\alpha-1)U^{(n)}}\overline{Y}_{n}$ has
uniformly bounded second moments. Namely by Lemma \ref{condYOUjMom},
a modification of Lemma \ref{lemcovEexpUnPhi}, Cauchy--Schwarz, bounding $\Expectation{\left(\Psi^{\ast^{(n)}} \right)^{2}}$ by a constant
and remembering that in this case $\sigma_{c}^{2}$ is constant

$$
\begin{array}{l}
\Expectation{H_{n}^{2}}=(n+1)^{2}\Expectation{e^{2(\alpha-1)U^{(n)}}\Expectation{\overline{Y}_{n}^{2} \vert \mathcal{Y}_{n}}}
\le C
n^{2}\left(
n^{-1}\Expectation{e^{-2(1-\alpha)U^{(n)}}}+\Expectation{e^{-2(1-\alpha)U^{(n)}-2\alpha \tau^{(n)}}}
\right. \\ \left.
+n^{-1}(\sigma_{a}^{2}/(2\alpha))^{-1}\Expectation{e^{-2(1-\alpha)U^{(n)}}\Psi^{\ast^{(n)}}}
+(\sigma_{a}^{2}/(2\alpha))^{-1}\Expectation{e^{-2(1-\alpha)U^{(n)}}\Psi^{(n)}}
\right)
\\
\le Cn^{2}\left(
n^{-1}n^{-2(1-\alpha)}+n^{-2(1-\alpha)}n^{-2\alpha}
+n^{-1}n^{-2(1-\alpha)}
+ n^{-2}
\right)
\\
\le C
\left(
n^{-1+2\alpha}+1
+n^{-1+2\alpha}
+1
\right) \to C < \infty.
\end{array}
$$
To deal with $\Expectation{e^{-2(1-\alpha)U^{(n)}}\Psi^{(n)}}$ one slightly modifies the proof
of Lemma \ref{lemcovEexpUnPhi}. Namely instead of considering the random variable $\bar{\phi}_{i}$,
consider 

$$\tilde{Z}_{i}\tilde{\mathbf{1}}_{i}\exp\left(-(2(T_{n}+\ldots+T_{i+1})+2(1-\alpha)(T_{i}+\ldots+T_{1})) \right)$$
and then doing similar calculations one will obtain a decay of order $n^{-2}$.
It is also worth pointing out that using \citet{KBarSSag2015}'s Lemma 3
for a more detailed consideration of $\Expectation{e^{-2(1-\alpha)U^{(n)}-2\alpha \tau^{(n)}}}$,
would not result in a different rate of decay, than what Cauchy--Schwarz provides, i.e. $n^{-2}$.
Hence, $\sup_{n} \Expectation{H_{n}^{2}} < \infty$ and by the martingale convergence theorem,
$H_{n} \to H_{\infty}$ a.s. and in $L^{2}$. 
We obtain 
$n^{\alpha}\overline{Y}_{n} \to V^{(\alpha-1)}H_{\infty}$
a.s. and in $L^{2}$, where $V^{(x)}$ is the a.s. and $L^{2}$ limit of
$V_{n}^{(x)}=b_{n,x}^{-1}e^{-xU^{(n)}}$ (cf. Lemma 9 in \citet{KBarSSag2015}).
Notice that for the convergence to hold in the $0<\alpha<0.5$ regime, 
it is not required that 
$\sigma_{c,k}^{2}$ is constant, only bounded.
We may also obtain directly the first two moments of 
$n^{\alpha}\overline{Y}_{n}$ (however, for these formul\ae\ to hold, $\sigma_{c,k}^{2}$ has to be constant)

$$
\begin{array}{rcl}
n^{\alpha}\Expectation{\overline{Y}_{n}} & = & \delta n^{\alpha} b_{n,\alpha} \to \delta \Gamma(1+\alpha) \\
n^{2\alpha}\Expectation{\overline{Y}_{n}^{2}} & = & 
n^{2\alpha-1} - (1-\delta^{2})n^{2\alpha}b_{n,2\alpha}
+n^{2\alpha}(1-n^{-1})\Expectation{e^{-2\alpha \tau^{(n)}}}
\\ &&
+n^{2\alpha-1}\sigma_{c}^{2}(\sigma_{a}^{2}/(2\alpha))^{-1}\Expectation{\Psi^{\ast^{(n)}}}
+ n^{2\alpha}\sigma_{c}^{2}(\sigma_{a}^{2}/(2\alpha))^{-1}\Expectation{\Psi^{(n)}}
\\ & \to &
- (1-\delta^{2})\Gamma(2\alpha+1)
+\frac{1+2\alpha}{1-2\alpha}\Gamma(1+2\alpha)(1+2p\sigma_{c}^{2}(\sigma_{a}^{2}/(2\alpha))^{-1}).
\end{array}
$$
\\~\\ \noindent
{\sc Proof of Theorem \ref{thmCLTYOUjpsae0}, Part \ref{rCLTthmpI}, $\alpha > 0.5$} 
\\~\\
From the proof of Part \ref{rCLTthmpIpConst}, Theorem \ref{thmCLTYOUjpsConst} 
we know that $\mu_{n} \to 0$ in probability.
Then, by the assumptions of the theorem on the expectations

$$
\begin{array}{rcl}
\Expectation{\sigma_{n}^{2}} & = & n\left(
n^{-1}-\Expectation{e^{-2\alpha U^{(n)}}} +
(1-n^{-1})\Expectation{e^{-2\alpha \tau^{(n)}}}
\right.\\~\\ &&\left.
+
n^{-1}(\sigma_{a}^{2}/(2\alpha))^{-1} 
\Expectation{
\sum\limits_{k=1}^{\Upsilon^{(n)}} \sigma_{c,\mathrm{I}^{(n)}_{k}}^{2}J_{k}e^{-2\alpha (T_{n}+\ldots+T_{\mathrm{I}^{(n)}_{k}+1})}
}
\right.\\~\\&&\left.
+(1-n^{-1})(\sigma_{a}^{2}/(2\alpha))^{-1}
\Expectation{
\sum\limits_{k=1}^{\upsilon^{(n)}}\sigma_{c,\mathrm{\tilde{I}}^{(n)}_{k}}^{2}\tilde{J}_{k}e^{-2\alpha (\tau^{(n)}+\ldots+T_{\mathrm{\tilde{I}}^{(n)}_{k}+1})}
}\right)
\\ & \to & \frac{2\alpha+1}{2\alpha-1} + (\sigma_{\Upsilon}^{2}+\sigma_{\upsilon}^{2})/(\sigma_{a}^{2}/(2\alpha)).
\end{array}
$$
Furthermore, Lemma 
\ref{lemvarEexpPsi} (the sequence $\sigma_{c,k}^{4}p_{k}$ is bounded by assumption) 
and Corollary \ref{corvarEexpPhi} ($\sigma_{c,k}^{4}p_{k}(1-p_{k})\to 0$ with density $1$ by assumption) imply 

$$
\begin{array}{rcl}
\variance{\sigma_{n}^{2}} & = & n^{2}\variance{\variance{\overline{Y}_{n} \vert \mathcal{Y}_{n}}}
= n^{-2}\variance{\variance{Y_{1}+ \ldots + Y_{n} \vert \mathcal{Y}_{n}}}
\\ & \le & 
C\left( 
n^{2}\variance{e^{-2\alpha U^{(n)}}} +
(n-1)^{2}\variance{\Expectation{e^{-2\alpha \tau^{(n)}}\vert\mathcal{Y}_{n}}}
\right. \\~\\ && \left.
+
(\sigma_{a}^{2}/(2\alpha))^{-2} 
\variance{\Expectation{
\sum\limits_{k=1}^{\Upsilon^{(n)}} \sigma_{c,\mathrm{I}^{(n)}_{k}}^{2}J_{k}e^{-2\alpha (T_{n}+\ldots+T_{\mathrm{I}^{(n)}_{k}+1})}
\vert\mathcal{Y}_{n} }}
\right. \\~\\ && \left.
+(n-1)^{2}(\sigma_{a}^{2}/(2\alpha))^{-2}
\variance{\Expectation{
\sum\limits_{k=1}^{\upsilon^{(n)}}\sigma_{c,\mathrm{\tilde{I}}^{(n)}_{k}}^{2}\tilde{J}_{k}e^{-2\alpha (\tau^{(n)}+\ldots+T_{\mathrm{\tilde{I}}^{(n)}_{k}+1})}
\vert \mathcal{Y}_{n} }} 
\right)
\\ &  \to & 0.
\end{array}
$$
Therefore we obtain that $\sigma_{n}^{2} \to (2\alpha+1)/(2\alpha-1)+ (\sigma_{\Upsilon}^{2}+\sigma_{\upsilon}^{2})/(\sigma_{a}^{2}/(2\alpha))$ in probability
and by convergence of characteristic functions

$$
\Expectation{e^{i x \sqrt{n} \cdot\overline{Y}_{n}}}=\Expectation{e^{i\mu_{n}x-\sigma_{n}^{2}x^{2}/2}}\to 
\Expectation{e^{-((2\alpha+1)/(2\alpha-1)+ (\sigma_{\Upsilon}^{2}+\sigma_{\upsilon}^{2})/(\sigma_{a}^{2}/(2\alpha)))x^{2}/2}}
$$
we obtain the asymptotic normality. Notice that on the other hand using
the Cauchy--Schwarz inequality, Lemmata \ref{lemcovEexpUnPhi} and \ref{lemcovEexptaunPhi}
we obtain

$$
\begin{array}{l}
\variance{\sigma_{n}^{2}} \ge 
n^{2}\variance{e^{-2\alpha U^{(n)}}} +(n-1)^{2}(\sigma_{a}^{2}/(2\alpha))^{-2}
\variance{
\Expectation{
\sum\limits_{k=1}^{\upsilon^{(n)}}\sigma_{c,\mathrm{\tilde{I}}^{(n)}_{k}}^{2}\tilde{J}_{k}e^{-2\alpha (\tau^{(n)}+\ldots+T_{\mathrm{\tilde{I}}^{(n)}_{k}+1})}
\vert \mathcal{Y}_{n} }} + b_{n}(\alpha),
\end{array}
$$
where $b_{n}(\alpha)$ is some sequence decaying to $0$ with a rate depending on $\alpha$.
Assume now that $p_{k}(1-p_{k})\sigma_{c,k}^{4}$ does not converge $0$ with density $1$.
Then, by Corollary \ref{corvarEexpPhi} we will have 

$$
\limsup\limits_{n\to\infty}\variance{
\Expectation{
\sum\limits_{k=1}^{\upsilon^{(n)}}\sigma_{c,\mathrm{\tilde{I}}^{(n)}_{k}}^{2}\tilde{J}_{k}e^{-2\alpha (\tau^{(n)}+\ldots+T_{\mathrm{\tilde{I}}^{(n)}_{k}+1})}
\vert \mathcal{Y}_{n} }}>0$$ implying $\limsup\limits_{n\to\infty}\variance{\sigma_{n}^{2}}>0$
and hence, the convergence of the characteristic functions as above does not hold.
Therefore, the convergence $p_{k}(1-p_{k})\sigma_{c,k}^{4}\to 0$ with density $1$ is 
a necessary assumption for the asymptotic normality.

{\sc Proof of Part \ref{rCLTthmpII},  $\alpha = 0.5$}
This is proved in the same way as {\sc Part \ref{rCLTthmpI}}. 
Due to the boundedness of $\sigma_{c,k}^{4}p_{k}$ (implying $\sigma_{c,k}^{2}p_{k}$ is bounded)

$$
\begin{array}{l}
(\ln^{-1}n)\Expectation{
\sum\limits_{k=1}^{\Upsilon^{(n)}} \sigma_{c,\mathrm{I}^{(n)}_{k}}^{2}J_{k}e^{- (T_{n}+\ldots+T_{\mathrm{I}^{(n)}_{k}+1})}
}\to 0 
\end{array}
$$
and as before

$$
(\ln^{-2}n)\variance{\Expectation{
\sum\limits_{k=1}^{\Upsilon^{(n)}} \sigma_{c,\mathrm{I}^{(n)}_{k}}^{2}J_{k}e^{- (T_{n}+\ldots+T_{\mathrm{I}^{(n)}_{k}+1})}
\vert\mathcal{Y}_{n} }}\to 0.
$$
Then, due to the assumption on the expectation we have that 
$\sigma^{2}_{n} \to 2 + \sigma_{\upsilon}^{2}/\sigma^{2}_{a}$
as due to the boundedness of $\sigma_{c,k}^{4}p_{k}$ by Lemma \ref{lemvarEexpPhi}

$$
(n^{2} \ln^{-2}n)
\variance{\Expectation{
\sum\limits_{k=1}^{\upsilon^{(n)}}\sigma_{c,\mathrm{\tilde{I}}^{(n)}_{k}}^{2}\tilde{J}_{k}e^{- (\tau^{(n)}+\ldots+T_{\mathrm{\tilde{I}}^{(n)}_{k}+1})}
\vert \mathcal{Y}_{n} }} \to 0.
$$

\begin{remark}\label{remBound}
The boundedness assumption for $\sigma_{c,k}^{4}p_{k}$ ($\alpha \ge 0.5$), together with the convergence to $0$
with density $1$ of $\sigma_{c,k}^{4}p_{k}(1-p_{k})$ (for $\alpha > 0.5$) allows for controlling
$\variance{\sigma_{n}^{2}}\to 0$. The boundedness assumption would still allow for showing that
$\Expectation{\sigma_{n}^{2}}$ is bounded but would not suffice for convergence. For example,
consider $p_{k}\equiv 1$ constant and $\sigma_{c,k}^{2}=1$ for $k$ odd and $2$ for $k$ even. 
Then, for $\alpha>0.5$ we would have in probability

$$
\liminf\limits_{n\to \infty}\sigma_{n}^{2} = 
1+ \frac{2}{2\alpha-1}+
\frac{2}{\sigma_{a}^{2}}+ 
\frac{4}{(2\alpha-1)\sigma_{a}^{2}}
~~\mathrm{and}~~
\limsup\limits_{n\to \infty}\sigma_{n}^{2} = 
1+ \frac{2}{2\alpha-1}+
\frac{4}{\sigma_{a}^{2}}+ 
\frac{8}{(2\alpha-1)\sigma_{a}^{2}}.
$$

It does seem that for $\alpha>0.5$ the assumption that $\sigma_{c,k}^{4}p_{k}$ is bounded
could be relaxed. However, it would essentially require that one explicitly
assumes that the sequences $\{\sigma_{c,k}^{2}\}$, $\{p_{k}\}$ are such that the 
sequences of the variances of the conditional expectations converge to $0$
(as was needed for the sequences of the expectations). 
\end{remark}

\begin{exam}\label{exPsae0}
Assume that $\sigma_{c,k}^{4}p_{k}\to 0$ with density $1$.
Then,
by the same ergodic argument as in Corollary \ref{corvarEexpPsi} 
and following the steps in the proof of Thm. \ref{thmCLTYOUjpsae0}
we obtain 
that for $\alpha>0.5$ we have 

$$
\begin{array}{l}
\Expectation{
\sum\limits_{k=1}^{\Upsilon^{(n)}} \sigma_{c,\mathrm{I}^{(n)}_{k}}^{2}J_{k}e^{-2\alpha (T_{n}+\ldots+T_{\mathrm{I}^{(n)}_{k}+1})}
}\to 0
\\ n
\Expectation{
\sum\limits_{k=1}^{\upsilon^{(n)}}\sigma_{c,\mathrm{\tilde{I}}^{(n)}_{k}}^{2}\tilde{J}_{k}e^{-2\alpha (\tau^{(n)}+\ldots+T_{\mathrm{\tilde{I}}^{(n)}_{k}+1})}
} \to 0
\end{array}
$$
resulting in 
$\sigma_{n}^{2}\xrightarrow{\mathbb{P}} (2\alpha+1)/(2\alpha-1)$
and for $\alpha=0.5$ we have also

$$
(n\ln^{-1}n)
\Expectation{
\sum\limits_{k=1}^{\upsilon^{(n)}}\sigma_{c,\mathrm{\tilde{I}}^{(n)}_{k}}^{2}\tilde{J}_{k}e^{- (\tau^{(n)}+\ldots+T_{\mathrm{\tilde{I}}^{(n)}_{k}+1})}
} \to 0
$$
resulting in
$\sigma_{n}^{2}\xrightarrow{\mathbb{P}} 2$.
Hence, to recover \citep{KBarSSag2015}'s CLTs one needs the stronger 
assumption of $\sigma_{c,k}^{4}p_{k}\to 0$ with density $1$.
\end{exam}

\section*{Acknowledgements}
A significant part of this work was done at the Department of 
Mathematics, Uppsala University, Sweden, during my postdoc which was supported 
by the Knut and Alice Wallenberg Foundation. Currently I am supported by the
Swedish Research Council (Vetenskapsr\aa det) grant no. $2017$--$04951$.
I would like to acknowledge Olle Nerman for his suggestion on adding jumps to the 
branching OU process, Wojciech Bartoszek for helpful suggestions concerning 
ergodic arguments, Venelin Mitov and Tanja Stadler for many discussions.
I would like to thank anonymous reviewers who found a number of errors and 
whose comments immensely improved the manuscript. I am especially grateful
for pointing out a number of flaws in the original proof 
of Corollary \ref{corvarEexpPsi} and suggesting a significantly more elegant way
of proving it. I am also grateful to Torkel Erhardsson and our collaboration
\citep[][]{KBarTErh2019arXiv} which allowed for a correct formulation
of Thms. \ref{thmCLTYOUjpsConst} and \ref{thmCLTYOUjpsae0}.

\bibliographystyle{plainnat}
\bibliography{OUjCLT}

\end{document}